\documentclass[11pt]{amsart}
\usepackage{latexsym,amssymb,calc,amsmath,amsthm,amssymb,color,graphicx,graphics,enumerate}
\newcommand{\ds}{\displaystyle}
\newcommand{\pt}{\partial_t}
\newcommand{\stress}{\mathbf{S}}

\newcommand{\esssup}[1]{\mathop{\rm ess\ sup}}
\newcommand{\essinf}[1]{\mathop{\rm ess\ inf}}
\newcommand{\D}{{\rm {-}\kern - 6pt D}}

\newcommand{\rr}{\mathbb R}

\newcommand{\R}{\mathbb{R}}

\newcommand{\Lp}{L_{\lower0.05cm\hbox{\scriptsize X}}^{\lower-0.05cm\hbox{\scriptsize 1}}}
\newcommand{\Linf}{L_{\lower0.05cm\hbox{\scriptsize X}}^{\lower-0.05cm\hbox{\scriptsize $\infty$}}}
\newcommand{\CX}{C_{\lower0.05cm\hbox{\scriptsize X}}}
\newcommand{\TAK}{T^{\lower-0.05cm \hbox{$\scriptstyle A$}}_{\lower 0.05cm \hbox{$\scriptstyle K$}}}
\newcommand{\TK}{T_{\lower 0.05cm \hbox{$\scriptstyle K$}}}
\newcommand{\TC}{T_{\lower 0.05cm \hbox{$\scriptstyle C$}}}
\newcommand{\PET}{P_T^{\lower-0.05cm \hbox{$\scriptstyle \epsilon$}}}
\newcommand{\ea}[1]{\lower 1mm \hbox{$|_{#1}$}}
\newcommand{\lims}[1]{\lower0.2cm\hbox{${\overline{\ds{\lim}_{}}\atop {#1}}$}}
\newcommand{\limi}[1]{\lower0.2cm\hbox{${\underline{\ds{\lim}_{}}\atop {#1}}$}}

\numberwithin{equation}{section}

\newtheorem{theorem}{Theorem}[section]

\newtheorem{lemma}[theorem]{Lemma}

\theoremstyle{definition}

\theoremstyle{definition} 
\newtheorem{remark}[theorem]{Remark}

\renewcommand{\div}{{\rm div \,}}

\DeclareMathSymbol{\complement}{\mathord}{AMSa}{"7B}

\def\vv<#1>{\langle #1\rangle}
\def\Vv<#1>{\bigl\langle #1\bigr\rangle}

\begin{document}
\title[Multicomponent Diffusion Modeling]
{On the structure of continuum thermodynamical diffusion fluxes -- A novel closure scheme and its relation to the Maxwell-Stefan and the 
Fick-Onsager approach}

\author[Dieter Bothe]{Dieter Bothe \vspace{0.1in}}
\address{Mathematical Modeling and Analysis\\
Technical University of Darmstadt\\
Alarich-Weiss-Str.~10\\
D-64287 Darmstadt, Germany}
\email{bothe@mma.tu-darmstadt.de}

\author[Pierre-\'Etienne Druet]{Pierre-\'Etienne Druet \vspace{0.1in}}
\address{Weierstrass Institute Berlin\\
Mohrenstr.~39\\
D-10117 Berlin, Germany}
\email{PierreEtienne.Druet@wias-berlin.de}

\thanks{The first author gratefully acknowledged funding within the SFB 1194 "Interaction of Transport and Wetting Processes", Project-ID 265191195, subproject B01.
The second author was supported by the grant DR-1117/1-1 of the German Science Foundation}

\begin{abstract}
This paper revisits the modeling of multicomponent diffusion within the framework of thermodynamics of irreversible processes.
We briefly review the two well-known main approaches, leading to the generalized Fick-Onsager multicomponent diffusion fluxes
or to the generalized Maxwell-Stefan equations. The latter approach has the advantage that the resulting fluxes are consistent with non-negativity of the partial mass densities for non-singular and non-degenerate Maxwell-Stefan diffusivities.
On the other hand, this approach requires computationally expensive matrix inversions since the fluxes are only
implicitly given.
We propose and discuss a novel and more direct closure which avoids the inversion of the Maxwell-Stefan equations.
It is shown that all three closures are actually equivalent under the natural requirement
of positivity for the concentrations, thus revealing the general structure of continuum thermodynamical diffusion fluxes.
%
As a special case, the new closure also gives rise to a core-diagonal diffusion model in which only those cross-effects are present that are necessary to guarantee consistency with total mass conservation, plus a compositional dependence of the diffusivity.
This core-diagonal closure turns out to provide a rigorous fundament for recent extensions of the Darken equation from binary mixtures to the general multicomponent case.
As an outcome of our investigation, we also address different questions related to the sign of multicomponent thermodynamic or Fickian diffusion coefficients. We show rigorously that in general the second law requires positivity properties for tensors and operators rather than for scalar diffusivities.

\end{abstract}
\maketitle
{\small\noindent
 \vspace{0.1in} \vspace{0.2in}
\section{Introduction}
In 1855, Adolf Fick summarized his findings on diffusion of matter in liquids in his seminal paper
\cite{Fick} and formulated (in one dimension)
what is today known as Fick's first law, namely the closure relation
\begin{equation}
\label{Fick-1}
{\bf j}_i^{\rm mol} = - D_i \nabla c_i,
\end{equation}
where ${\bf j}_i^{\rm mol}$ denotes the molar mass flux and $c_i$ is the molar concentration of constituent $A_i$,
$i\in\{1,\ldots ,N\}$, of a mixture.
The phenomenological coefficient $D_i$ is called diffusion coefficient or diffusivity of the species.
It turned out that this ''law'' applies well for dilute species, i.e.\ a species $A_i$ such that its molar
fraction $x_i=c_i/c$, with $c=\sum_{k=1}^N c_k$ the total molar concentration, satisfies $x_i \ll 1$.
But outside of this dilute case, closure laws for multicomponent diffusion fluxes get much more involved, mainly in two ways:
(i) even for moderate concentrations, the activity of a species is not linearly related to its concentration and (ii)
interactions between species $i$ and all other constituents--instead of interactions only between a dilute solute and the solvent--lead to cross-effects. Concerning (i), a more deeper thermodynamical study showed that--still in simple cases--the so-called driving force of diffusion is not given by $-\nabla c_i$, but rather by $-c_i \nabla \frac{\mu_i^{\rm mol}}{RT}$ with $\mu_i^{\rm mol}$
denoting the chemical potential of constituent $i$.
Regarding (ii) above, let us note that the continuity equation
(see \eqref{continuity} below) as the total mass balance imposes a constraint on the set of all diffusion fluxes, which
necessarily leads to a cross-coupling between all mass fluxes.

In fact, only in very particular cases,
use of Fick's law for all constituents of a mixture is consistent with the continuity equation, usually
requiring the introduction of \emph{ad hoc} convective velocities and/or the rather unrealistic assumption that all diffusivities coincide.

The general setting of the Theory of Irreversible Processes (T.I.P.) as summarized in \cite{dGM} allowed to overcome these deficiencies and has lead to the Fick-Onsager form of diffusion fluxes, replacing the diffusion
coefficients $D_i$ by a matrix of phenomenological coefficients $L_{ij}$, also called Onsager coefficients (for diffusion). This matrix will be dense in real applications and the $L_{ij}$ strongly depend on the composition of the system as well as on temperature, i.e.\ they are functions of the thermodynamic state
variables and in a highly non-trivial manner.
This makes it difficult to infer realistic models for the phenomenological coefficients from
experimental measurements or molecular dynamics (MD) simulation, although it is rather easy to compute
an approximation of the value of $L_{ij}$ from MD simulations for a specific thermodynamic state.
The difficulty comes in because of the complicated dependencies in particular on the composition.

Another approach employs the Maxwell-Stefan equations which--in the diffusive approximation--assumes
a force balance between friction forces occurring because of the relative motion of the different constituents
and so-called thermodynamic driving forces. The coefficients appearing in these force balances are hence
friction coefficients and these are always associated with a pair of constituents. The reciprocal of these
coefficients, possibly up to a common factor, are called Maxwell-Stefan (MS) diffusivities and denoted as
$\D_{ij}$. The $\D_{ij}$ also depend on temperature and especially on the composition, but in a much more moderate
way as compared to the Onsager coefficients. An indication for this is the fact that even with constant
$\D_{ij}$, the system of partial differential equations resulting from the partial mass balances
allows for non-negative solutions (which exist at least locally in time). The latter is not true if constant $L_{ij}$ are used in the Fick-Onsager approach. {\color{black} At least in the case that non-conservative mass transfer mechanisms are present (chemical reactions, boundary interactions), the complete extinction of constituents in finite time is possible even for smooth solutions: See \cite{dredrugagu20} for some mathematical results including this aspect.}

The thermodynamic driving forces used in the Maxwell-Stefan approach have been obtained outside of the T.I.P.\
via approximate solutions to the multicomponent Boltzmann equations. It took a while until it was shown in \cite{Standart} that this approach is consistent with the second law of thermodynamics.
On the other hand, for small--up to quaternary--systems, the Maxwell-Stefan diffusivities are nowadays routinely used to compute (Fick-)Onsager coefficients by using explicit inversion formulae; cf.\ \cite{TK-book}. Interestingly, MD simulations are often employed for computing Maxwell-Stefan diffusivities in which case,
in the first place, Green-Kubo integral representations are used to compute Onsager coefficients.
These are then converted into Maxwell-Stefan diffusivities since the latter are  better suited to describe the dependence on composition.
This also demonstrates that a close link between both approaches is of course present.
Actually, it has already been stated explicitly (see, e.g., \cite{Cussler}) that both approaches are equivalent - albeit without a precise statement how this equivalence should look like, not to mention a proof of the equivalence.

It is one of the main results in the present paper that the Fick-Onsager and the Maxwell-Stefan approach to multicomponent diffusion are indeed equivalent to each other in a rigorous, mathematical sense. Furthermore, as another main contribution we introduce a third way of performing a thermodynamically consistent closure of the diffusion fluxes which combines the advantages of the Fick-Onsager and the Maxwell-Stefan approach.
This novel closure separates those cross-couplings between the diffusion fluxes which are necessary to
simultaneously fulfill the continuity equation, i.e.\ the fact that all mass diffusion fluxes taken
against the barycentric velocity must sum up to zero, from ''true cross-effects'' which may be present or not.
The outcome, again, is proven to be equivalent to both the Fick-Onsager and the Maxwell-Stefan approach.

In recent years, the topic of accurate and consistent modeling of multicomponent diffusion fluxes has received an even increased popularity. On the one hand, the ever growing computing power allows for detailed and fully resolved numerical simulations of transport processes in multicomponent fluid mixtures, mainly with chemical reactions.

This is relevant to understand, quantitatively describe, optimize and control a multitude of phenomena and processes, including
ocean-atmosphere exchange, cloud physics, combustion, separation, materials processing, water treatment, electrolysis and many more.
For this purpose, highly accurate models for the diffusion coefficients are mandatory. Accordingly, usually in the context of the Maxwell-Stefan approach, (semi-)empirical or theoretical models for the $\D_{ij}$ have been further developed, often based on existing relations for binary systems. For instance, the classical Vignes equation,
\[
\D_{12} \, =\, \big( \D_{12}^{x_1 \to 1} \big)^{x_1} \, \big( \D_{12}^{x_2 \to 1} \big)^{x_2}
\]
has been extended in \cite{KvB2005} to ternary systems based on empirical information. 

\noindent
Another example is the classical Darken equation,
\[
\D_{12} \, =\, x_2 D_{1,\rm self} +  x_1 D_{2,\rm self},
\]
which has been extended in \cite{Bardow2011} to the multicomponent case based on theoretical ad hoc arguments.
In this respect and as another main result, we show that our novel closure approach in the special case of a core-diagonal closure,
i.e.\ if only the necessary cross-effects enforced by the total mass balance are included, provides a
consistent derivation of this multicomponent extension of the Darken equation.

In the applied literature, there are ongoing discussions on the sign of Maxwell-Stefan diffusivities; see, e.g., \cite{chakrawangeapen}, \cite{kraaijeveldwesselingh}. We therefore also ask in this paper which consequences can be derived from the second law of thermodynamics for the main scalar diffusion coefficients occurring in the three closure schemes. For special cases with simple compositional dependencies -- as for instance purely binary dependence of MS-coefficients -- we show that thermodynamic diffusivities are sometimes required to be positive. But as already shown for ternary systems by counterexamples, the general picture is that the second law does not impose the sign of diffusion coefficients. We prove that the same is essentially valid for the multicomponent Fickian diffusion coefficients expressing the proportionality between diffusion fluxes and concentration gradients.

In the last decade, new interest in the mathematical analysis of multicomponent diffusion systems, in particular within the Maxwell-Stefan formulation, arose. This comprises both the proof of short time existence of strong or even classical solutions (see \cite{JP-MS}, \cite{Giovan}, \cite{piashiba19}, \cite{BD-COMP}) as well as global existence of weak solutions (see \cite{chenjuengel}, \cite{mupoza15}, \cite{dredrugagu20}, \cite{bondesanbriant}, \cite{PED-MS}). In most of these papers, the underlying physics of multicomponent diffusion is exploited for the way to treat these systems mathematically.

Let us therefore note that our estimate of ellipticity constants for diffusion matrices in Theorem~7.2 below
is a new finding out of the equivalence proof, but also of interest for the mathematical analysis.
Currently, the existence of unique global solutions without restriction on the initial data is not known, even for the isothermal and isobaric case where the mechanics simplifies to trivial. Here the core-diagonal case, where cross-effects are reduced to minimal, might be an interesting case to start with.


In order to focus on the aspect of multicomponent diffusion, we restrict our attention to the \emph{iso-thermal case} throughout the entire paper. Indeed, variable temperature implies that the forces driving diffusion have a more complex expression than in the isothermal case. The question of finding consistent proportional closure relations between diffusion fluxes and these driving forces is not essentially affected, though.

%

For more information about the history of research on diffusion in liquids see, e.g., \cite{Cussler}, \cite{Mehrer}
and the primary literature cited there.\\[2ex]

%
%
%

\section{Continuum Thermodynamical Framework}
\label{TIP}
A thermodynamically consistent model for multicomponent diffusion has to be based on continuum thermodynamics
of irreversible processes, for instance in the framework of the by now classical Theory of Irreversible
Processes (T.I.P.) as summarized in the excellent book by De Groot and Mazur \cite{dGM}.
Within Continuum Thermodynamics, the modeling of multicomponent diffusion fluxes is based on the
partial mass balances for the mass densities $\rho_i$ of all constituents $A_1,\ldots ,A_N$ of the considered
mixture. The partial mass balances read as
\begin{equation}
\label{partial-mass}
\pt \rho_i + \div (\rho_i {\bf v}_i)= r_i,
\end{equation}
where $r_i$ are the rates of mass production due to chemical reactions and ${\bf v}_i$ are the individual
continuum mechanical velocities.
The mass production rate for species $A_i$ is
\begin{equation}
r_i = M_i \sum_{a=1}^{N_R} \nu_i^a R_a,
\end{equation}
where $M_i$ is the molar mass of $A_i$, $R_a=R_a^f - R_a^b$ denotes the molar rate of the $a^{\rm th}$ reaction
(forward minus backward rate) and the $\nu_i^a$ are the stoichiometric coefficients for this reaction, i.e.\ $\nu_i^a=\beta_i^a - \alpha_i^a$ if the $a^{\rm th}$ chemical reaction is of the type
\begin{equation}
\alpha_1^a A_1 + \ldots + \alpha_N^a A_N \rightleftharpoons \beta_1^a A_1 + \ldots + \beta_N^a A_N.
\end{equation}
In this section, we briefly recall the class-I model for a multicomponent fluid mixture, i.e.\
we consider partial mass balances but only a single balance for the total momentum and energy, respectively;
see \cite{BD} and the references given there for class-II modeling with partial momenta.

To formulate the common momentum balance,
we define the total mass density $\rho$ and the barycentric velocity ${\bf v}$ of the mixture according to
\begin{equation}
\label{barycentric}
\rho :=  \sum_{i=1}^N \rho_i
\quad \mbox{ and } \quad
\rho\, {\bf v} := \sum_{i=1}^N \rho_i {\bf v}_i.
\end{equation}
Exploiting the conservation of total mass,
summation of \eqref{partial-mass} for $i=1,\ldots ,N$ then yields the total mass balance
\begin{equation}\label{continuity}
\pt \rho  +  \div ( \rho {\bf v} )=0,
\end{equation}
i.e.\ the continuity equation. The total momentum balance reads as
\begin{equation}
\label{total-momentum-balance}
\pt (\rho {\bf v}) + \div (\rho {\bf v} \otimes {\bf v} - \stress)
= \rho {\bf b},
\end{equation}
where $\stress$ denotes the stress tensor, and the total body force ${\bf b}$
is given via
\begin{equation}
\rho {\bf b}:=\sum_{i=1}^N \rho_i {\bf b}_i
\end{equation}
in the general case of individual body forces ${\bf b}_i$.
Next, we let
\begin{equation}
\label{mech-pressure}
P:= - \frac 1 3 {\rm tr} (\stress)
\quad \mbox{ and } \quad
\stress^\circ := \stress + P \, {\bf I}
\end{equation}
denote the \emph{mechanical pressure} and the traceless part of the stress, respectively.
We only consider non-polar fluids for which the balance of angular momentum reduces to the symmetry of $\stress$,
i.e.\ we have
\begin{equation}
\label{S-symm}
\stress^{\sf T} = \stress.
\end{equation}

In the class-I modeling, the individual velocities ${\bf v}_i$ are decomposed as
${\bf v}_i = {\bf v} + {\bf u}_i$ with the barycentric velocity ${\bf v}$ and the diffusion velocities
\begin{equation}
{\bf u}_i = {\bf v}_i -{\bf v}.
\end{equation}
Accordingly, the partial mass balances are written as
\begin{equation}
\label{partial-mass2}
\pt \rho_i + \div (\rho_i {\bf v} + {\bf j}_i)= r_i,
\end{equation}
with the mass diffusion fluxes defined as
\begin{equation}
{\bf j}_i =\rho_i {\bf u}_i.
\end{equation}
Due to \eqref{barycentric} and the continuity equation \eqref{continuity}, the diffusive mass fluxes satisfy the relation
\begin{equation}\label{flux-constraint}
 \sum_{i=1}^N {\bf j}_i =0,
\end{equation}
which is a constraint for the modeling of the fluxes ${\bf j}_i$.

Molar-based (instead of mass-based) variants are also useful, employing the molar
mass balances
\begin{equation}\label{molar-mass-balance}
\pt c_i + \div (c_i {\bf v} + {\bf j}_i^{\rm mol}) = \sum_{a=1}^{N_R} \nu_i^a R_a
\end{equation}
for the molar concentrations $c_i := \rho_i / M_i$ with molar diffusion fluxes
${\bf j}_i^{\rm mol}:={\bf j}_i / \rho_i$.
Note that these molar diffusion fluxes are still defined relative to the
barycentric velocity and have to obey the constraint
\begin{equation}\label{molar-flux-constraint}
 \sum_{i=1}^N M_i {\bf j}_i^{\rm mol} =0.
\end{equation}
Throughout the entire paper, the superscript ``mol'' is used for molar-based quantities to distinguish them from the mass-based variants.

The description of diffusion processes requires a full thermodynamical treatment of the considered mixture.
For this purpose, the balance of internal energy is required as well which reads as
\begin{equation}
\label{internal-energy-balance}
\pt (\rho e) + \div (\rho e {\bf v} + {\bf q})
= {\bf D}:\stress + \sum_{i=1}^N {\bf j}_i \cdot {\bf b}_i,
\end{equation}
where $e$ is the specific internal energy, ${\bf q}$ the heat flux and
${\bf D}=\frac 1 2 (\nabla {\bf v} + (\nabla {\bf v})^{\sf T} )$ the symmetric part of the
velocity gradient. Equation \eqref{internal-energy-balance} is also called the first law of (continuum)
thermodynamics.

In addition to these balances of conserved quantities, we finally add the balance of entropy, i.e.\
\begin{equation}
\label{entropy-balance}
\pt (\rho s) + \div (\rho s {\bf v} + {\bf \Phi}) = \zeta,
\end{equation}
where $s$ denotes the specific entropy, $\Phi$ the entropy flux and $\zeta$ the rate of entropy production.
The second law of (continuum) thermodynamics then postulates that
\begin{equation}
\zeta \geq 0 \quad \mbox{\rm for any thermodynamic process}.
\end{equation}
We consider the simplest class of (fluid) materials for which
\begin{equation}\label{specific-entropy}
\rho s = \widetilde{\rho s}(\rho e, \rho_1, \ldots ,\rho_N)
\end{equation}
with a strictly concave function $\widetilde{\rho s}$ which is strictly increasing in $\rho e$.
We then define the (absolute) temperature $T$ and the chemical potentials $\mu_i$ as
\begin{equation}\label{entropy-derivatives}
\frac 1 T = \frac{\partial \widetilde{\rho s}}{\partial \rho e},
\qquad
- \frac{\mu_i}{T} = \frac{\partial \widetilde{\rho s}}{\partial \rho_i}.
\end{equation}
Exploiting \eqref{entropy-balance}, \eqref{specific-entropy}, \eqref{entropy-derivatives} and eliminating
time derivatives by means of the balance equations, a straightforward calculation yields
\begin{align}\label{entropy-inequality-0}
\zeta  &  =\,  \div (\Phi - \frac{\bf q}{T} + \sum_{i=1}^N  \frac{\mu_i {\bf j}_i}{T} ) - \frac 1 T (\rho e - \rho s T + P - \sum_{i=1}^N \rho_i \mu_i )
 \, \div {\bf v}\\
& \, + {\bf q} \cdot \nabla \frac 1 T + \frac 1 T  \stress^\circ : {\bf D}^\circ
- \sum_{i=1}^N {\bf j}_i \cdot \Big( \nabla \frac{\mu_i}{T}
-\frac{{\bf b}_i}{T} \Big) - \frac 1 T \sum_{a=1}^{N_R} R_a \mathcal{A}_a\nonumber
\end{align}
for the entropy production, where $\mathcal{A}_a=\sum_{i=1}^N M_i \nu_i^a \mu_i$ are the so-called affinities
and ${\bf A}^\circ$ denotes the traceless part of a second-rank tensor ${\bf A}$.
Now the entropy principle, see \cite{BD} for a strengthened version, requires that a closure for the entropy flux
leads to a reduced entropy production being a sum over binary products.
Here, we choose the entropy flux as
\begin{equation}\label{entropy-flux}
{\bf \Phi} := \frac{\bf q}{T} - \sum_{i=1}^N  \frac{\mu_i {\bf j}_i}{T}
\end{equation}
and obtain the reduced entropy production rate as
\begin{align}\label{entropy-inequality-1}
\zeta  &  =\,   - \frac 1 T (\rho e - \rho s T + P - \sum_{i=1}^N \rho_i \mu_i )
 \, \div {\bf v}\\
& \, + {\bf q} \cdot \nabla \frac 1 T + \frac 1 T  \stress^\circ : {\bf D}^\circ
- \sum_{i=1}^N {\bf j}_i \cdot \Big( \nabla \frac{\mu_i}{T}
-\frac{{\bf b}_i}{T} \Big) - \frac 1 T \sum_{a=1}^{N_R} R_a \mathcal{A}_a.\nonumber
\end{align}
The entropy production is a sum of binary products, running over all dissipative mechanisms which are acting in the mixture.
In the considered case they correspond to -- in the order of their appearance in \eqref{entropy-inequality-1} --
volume variations, heat conduction, shear strain, multicomponent diffusion and chemical reactions.

The discussion of different approaches for a thermodynamically consistent closure of the multicomponent diffusion fluxes is the
main subject of the present paper and starts in the next section. Here, for the sake of completeness,
let us briefly explain the treatment of the
other binary products. For technical simplicity we do not consider cross-effects between the different dissipative mechanisms\footnote{i.e., for technical simplicity, the Dufour- and Soret-effects of thermo-diffusion (among other
possible couplings) are not treated.}
which can be treated easily using entropy neutral mixing on the level of the entropy production rate; see \cite{BD} for more
details on this new concept to introduce cross-effects which are then accompanied by associated Onsager-Casimir relations for the phenomenological coefficients.

For the first three binary products we employ the standard linear (in the co-factors) closure, i.e.\ we let
\begin{align}
\rho e - \rho s T + P - \sum_{i=1}^N \rho_i \mu_i  &  :=\,   - \lambda \, \div {\bf v},\vspace{-0.1in}\label{closure-1}\\
{\bf q} & := \tilde\alpha \nabla \frac 1 T = - \alpha \nabla  T,\label{closure-2}\\[1.5ex]
\stress^\circ & := 2 \eta {\bf D}^\circ,\label{closure-3}
\end{align}
where the phenomenological coefficients $\lambda, \alpha, \eta \geq 0$ depend on the basic thermodynamic variables which are, after a Legendre transformation, $T,\rho_1, \ldots ,\rho_N$.
Since a closure without cross-effects renders every individual binary product non-negative for any thermodynamic process, it is already clear that in equilibrium, i.e.\ for vanishing entropy production, the co-factors must also vanish, i.e.\
\begin{equation}\label{equilibrium}
\div {\bf v} =0, \quad \nabla T = 0, \quad {\bf D}^\circ = {\bf 0}.
\end{equation}
Therefore, the mechanical pressure at equilibrium satisfies
\begin{equation}\label{mech-press-equi}
P_{|equ} = - \rho e + \rho s T + \sum_{i=1}^N \rho_i \mu_i.
\end{equation}
Motivated by \eqref{mech-press-equi}, we define the \emph{thermodynamic pressure} as
\begin{equation}\label{pressure}
p := - \rho e + \rho s T + \sum_{i=1}^N \rho_i \mu_i.
\end{equation}
In order to interchange $\rho e$ with $T$ as independent variables, we introduce the free energy density as
\begin{equation}\label{free-energy}
\rho \psi := \rho e - \rho s T.
\end{equation}
It is straightforward to show that with $\psi = \psi (T,\rho, y_1, \ldots ,y_N)$, where the mass fractions
\begin{equation}\label{mass-fractions}
y_i := \frac{\rho_i}{\rho}
\end{equation}
satisfy the constraint $\sum_{i=1}^N y_i = 1$,
the thermodynamic pressure fulfills the relation
\begin{equation}\label{pressure-MR}
p = \rho^2 \frac{\partial \psi}{\partial \rho}.
\end{equation}
Conversely, if we define the thermodynamic pressure, as is often done, according to \eqref{pressure-MR} than $p$
necessarily satisfies the Gibbs-Duhem (also called Euler) relation, i.e.\
\begin{equation}\label{Gibbs-Duhem}
\rho e - \rho s T + p = \sum_{i=1}^N \rho_i \mu_i.\vspace{-0.1in}
\end{equation}
If we now let\vspace{-0.2in}\\
\begin{equation}\label{dyn-pressure}
\pi := P-p
\end{equation}
denote the dynamic pressure, the closure \eqref{closure-1}
simply reads as
\[
\pi = - \lambda \, \div {\bf v}
\]
and models a viscous pressure contribution due to volume variation, where the material dependent parameter $\lambda$ is called
bulk viscosity.

Next, we recall the closure of the mass production rates according to \cite{BD},
where $R_a$ is decomposed into forward minus backward rate according to
\[
R_a=R_a^f -R_a^b.
\]
Since chemical reactions are activated processes which often occur far from equilibrium,
a linear (in the affinities) closure for $R_a$ is not appropriate.
Instead, we use the nonlinear closure
\begin{equation}
\label{closure-chem-rates}
\ln \frac{R_a^f}{R_a^b}= - \gamma_a \frac{\mathcal{ A}_a}{RT}
\quad \mbox{ with } \gamma_a >0
\end{equation}
which implies
\[
\zeta_{\rm CHEM} = R \sum_{a=1}^{N_R} \frac{1}{\gamma_a}
(R_a^f - R_a^b) (\ln R_a^f - \ln {R_a^b}) \geq 0,
\]
since the logarithm is monotone increasing. Notice that still one of the rates --
either for the forward or the backward path -- needs to be
modeled, while the form of the other one then follows from \eqref{closure-chem-rates}.
This logarithmic closure not only allows to include standard mass action kinetics into this framework,
but to provide thermodynamically consistent extensions;
cf.\ \cite{Dreyer-Guhlke-M}, where this is employed for a rigorous derivation and
an extension of the Butler-Volmer equation for fluid interfaces.

Because of the strict monotonicity of the logarithm, 
the reactive contribution to the entropy
production only vanishes if {\it all reaction are separately in
equilibrium}, i.e.\ all forward and corresponding backward rates
coincide. This is an instance of the {\it principle of detailed
balance}, called Wegscheider's condition in the context of
chemical reaction kinetics.\\[1ex]
In addition to the required choice of appropriate phenomenological coefficients, the free energy $\rho \psi$
of the mixture needs to be modeled in order to arrive at a closed system of partial differential equations.
If the free energy is given as a function $\rho \psi = (\rho \psi) (T,\rho_1, \ldots , \rho_N)$, then
the chemical potentials can be obtained from
\begin{equation}
\mu_i (T,\rho_1, \ldots , \rho_N) = \frac{\partial (\rho \psi)}{\partial \rho_i} (T,\rho_1, \ldots , \rho_N)
\end{equation}
and the thermal equation of state, i.e.\ the relation $p=p(T,\rho_1, \ldots , \rho_N)$, is obtained from \eqref{pressure}.
While the modeling of appropriate free energies is a non-trivial important topic on its own, it lies outside of the scope of the present paper.
Let us only note in passing that a procedure for the construction of a consistent free energy, building on a given
thermal equation of state and on partial information on the chemical potentials, has been introduced in \cite{BD} and is currently being developed further \cite{BD-FE}.

To give a relevant prototype class for the underlying free energy model, let us introduce the so-called \emph{ideal mixtures} which are present in most of the classical references. We emphasize at the same time that the results of this paper are valid for much wider classes of mixtures than ideal ones. We call a mixture ideal if the chemical potentials obey the additive splitting
\begin{align}\label{idealchempot}
\mu_i = g_i(T, \, p) + \frac{c(T)}{M_i} \, \ln x_i \, .
\end{align}
The functions $g_1, \ldots,g_N$ are the Gibbs energies of the constituents, meaning that $\partial_{p}g_i(T, \, p) = 1/\hat{\rho}_i(T,p)$, where $\hat{\rho}_i(T,p)$ is the bulk density of the $i^{\rm th}$ constituent at $(T,p)$. The number $c(T)$ denotes some function of temperature only, usually $c(T) = RT$.

An equivalent characterization, shown in the upcoming paper \cite{BD-FE}, states that a mixture is ideal if and only if it is volume additive and simple with respect to the composition variable. The volume-additivity refers to the linear additive representation of the average molar volume
\begin{align}\label{volumeadditive}
\frac{1}{c} = \sum_{i=1}^N \frac{M_i}{\hat{\rho}_i(T,p)} \, x_i = \sum_{i=1}^N \partial_{p}g_i(T, \, p) \, M_i \, x_i \, .
\end{align}
It can be shown rigorously that the property of volume-additivity is equivalent with separation of the pressure and the composition variables according to
\begin{align}\label{volumeadditivechempot}
\mu_i = g_i(T, \, p) + a_i(T, \, {\bf x}) \, .
\end{align}
Here $a_i$ are certain maps related to the activities/fugacities of the species. Requiring also that $\mu_i$ depends on the composition of species $A_i$ only via $x_i$, which is the meaning of a simple dependence, it follows that $a_i(T, \, {\bf x}) = a_i(T, \, x_i)$. Then, rigorous arguments show that $a_i(T, \, x_i) =  c(T)/ M_i \, \ln x_i$, i.\ e.\, the representation \eqref{idealchempot} is necessary. We refer the interested reader to \cite{BD-FE} for complete proofs.

For more details about the classical theory of irreversible processes we refer in particular to \cite{dGM}. Concerning an extended
thermodynamical description involving partial momentum balances see \cite{BD} and further references given there.
Further information on the continuum thermodynamics of mixtures can be found, e.g., in the monographs
\cite{CT69}, \cite{M85},\cite{Raja}, \cite{Hutter-book}, \cite{KB-book} and \cite{PekarSamohyl}.

\section{Fick-Onsager Closure}\label{Fick-Onsager}
In 1945, Lars Onsager in \cite{Onsager-Diffusion} generalized Fick's law \eqref{Fick-1} to a non-dilute multicomponent liquid mixture by employing the closure
\begin{equation}
\label{Fick-MC}
{\bf j}_i^{\rm mol} = - \sum_{k=1}^N D_{ik} \nabla c_k
\end{equation}
with so-called binary diffusivities $D_{ik}$.
Recall that the molar concentration of species $A_k$ is given as
$c_k=\rho_k / M_k$ with corresponding balance equation \eqref{molar-mass-balance}.
In case $D_{ik}\neq 0$ for $i\neq k$, the flux of species $A_i$ obviously couples to the gradient of
concentrations of $A_k$, a phenomenon named \emph{cross-diffusion}.
In general, the closure above is not thermodynamically consistent:
\begin{enumerate}
\item
It does not comply with mass conservation unless, essentially, all
diffusivities are the same (cf.\ Chapter~7, Section~5 in \cite{Giovan});
\item
The co-factors of the mass fluxes in the entropy production, eq.\ \eqref{entropy-inequality-1},
are not (directly) involved, hence the second law will only be fulfilled in very particular cases.
\end{enumerate}

The classical theory of irreversible processes (T.I.P.), essentially in the form as briefly recalled in section~\ref{TIP} above, has been employed to generalize the Fickian closure from dilute (hence, in particular, ideal) systems to the non-dilute, non-ideal case.
Furthermore, it provides a consistent coupling between the partial mass and the total momentum balances.
The latter is extremely important, since pressure effects can lead to significant changes of diffusion fluxes,
mediated by the pressure dependence of the chemical potentials.
Even if no macroscopic flow is observed,
the pressure counteracts external forces and pressure gradients can be significant.
This is, e.g., relevant for modeling the transport processes in ultra-centrifuges, see \cite{TK-book},
and for the modeling of molecular transport in electrolytes, especially in the vicinity of electrodes, see \cite{Dreyer-Guhlke-M}.

In the present paper, the notion of ''Fick-Onsager diffusion fluxes'' shall be used \underline{not}
for \eqref{Fick-MC}, but for a thermodynamically consistent closure which yields expressions
for the ${\bf j}_i$ based on \eqref{entropy-inequality-1}.
In contrast to other approaches, the Fick-Onsager closure directly yields expressions for the diffusion fluxes, where the constraint \eqref{flux-constraint} is built in by elimination of one of them, say ${\bf j}_N$.
To keep this paper self-contained, we briefly show how to derive the relevant equations.
For better readability, we
specialize to the case ${\bf b}_i={\bf b}$ which can easily be generalized afterwards,
by replacing $\nabla \frac{\mu_i}{T}$ with $\nabla \frac{\mu_i}{T}-\frac{{\bf b}_i}{T}$.
With this simplification, the diffusional contribution to the entropy production reads as
\begin{equation}\label{entropy-diffusion}
    \zeta_{\rm DIFF}  \, =\,
- \sum_{i=1}^N {\bf j}_i \cdot \nabla \frac{\mu_i}{T},
\end{equation}
hence elimination of ${\bf j}_N$ via \eqref{flux-constraint} yields
\begin{equation}\label{entropy-Fick-i}
    \zeta_{\rm DIFF}  \, =\,
- \sum_{i=1}^{N-1} {\bf j}_i \cdot \nabla \frac{\mu_i -\mu_N}{T}.
\end{equation}
Restricting again to a linear (in the co-factors) closure, we let
\begin{equation}\label{Closure-Fick-i}
{\bf j}_i  \, :=\, - \sum_{k=1}^{N-1} L_{ik} \nabla \frac{\mu_k -\mu_N}{RT} \quad \mbox{ for } \; i=1,\ldots ,N-1
\end{equation}
with a symmetric and positive definite $(N-1)\times(N-1)$-matrix of phenomenological
(Onsager) coefficients $L_{ik}$.
These coefficients are sometimes called mobilities (note the different physical dimension
compared to the diffusivities introduced above) and
are functions of the state variables, i.e.\ $L_{ik}=L_{ik}(T,\rho ,{\bf y}' )$
In \eqref{Closure-Fick-i}, the universal gas constant $R$ has been inserted in order to simplify the
physical dimension of the $L_{ik}$. Note that $\mu_i^{\rm mol} /RT$ with the molar-based chemical
potential $\mu_i^{\rm mol}=M_i \mu_i$ is a dimensionless quantity.
The diffusion flux ${\bf j}_N$ follows from \eqref{flux-constraint}, resulting in
\begin{equation}\label{Closure-Fick-N}
{\bf j}_N  \, =\, - \sum_{k=1}^{N-1} L_{Nk} \nabla \frac{\mu_k -\mu_N}{RT},
\end{equation}
where
\begin{equation}\label{Lext1}
L_{Nk}:=- \sum_{i=1}^{N-1} L_{ik}.
\end{equation}
Evidently, the sums in \eqref{Closure-Fick-i} and \eqref{Closure-Fick-N} can also run up to $k=N$ for arbitrary
$L_{iN}$. We let
\begin{equation}\label{Lext2}
L_{iN}:=- \sum_{k=1}^{N-1} L_{ik} \quad \mbox{ for } \; i=1,\ldots ,N,
\end{equation}
which leads to a symmetric and positive semi-definite extended $N \times N$-matrix $[L_{ik}]$, and obtain
the symmetric (w.r.\ to the $\mu_k$) closure
\begin{equation}\label{Closure-Fick-ib}
{\bf j}_i  \, =\, - \sum_{k=1}^{N} L_{ik} \nabla \frac{\mu_k}{RT}
\quad \mbox{ for } \; i=1,\ldots ,N.
\end{equation}
The fluxes from \eqref{Closure-Fick-ib} are called ''Fick-Onsager diffusion fluxes'' throughout this paper.
Symmetry of $[L_{ik}]$ is assumed in accordance with the Onsager reciprocal relations and \eqref{Closure-Fick-ib}
is sometimes also referred to as the Fick-Onsager closure in the literature.

To facilitate the comparison of different closures, we rewrite the Fick-Onsager diffusion fluxes
employing a condensed tensor notation. For this purpose, we abbreviate the full system of fluxes as
\[
{\bf J} := [\, {\bf j}_1 | \cdots |\, {\bf j}_N ]^{\sf T}
\]
and let
\[
\nabla \frac{\boldsymbol\mu}{RT}  := [\, \nabla \frac{\mu_1}{RT} | \cdots |\, \nabla \frac{\mu_N}{RT} ]^{\sf T}.
\]
Then
\begin{equation}\label{Closure-Fick}
{\bf J} \, =\, - {\bf L} \, \nabla \frac{\boldsymbol\mu}{RT} \,
=\, - {\bf P}_N^{\sf T} \, {\bf L}' \,  {\bf P}_N \, \nabla \frac{\boldsymbol\mu}{RT},
\end{equation}
where ${\bf L}:=[L_{ik}] \in \R^{N\times N}$, ${\bf L}':=[L_{ik}]_{i,k=1}^{N-1} \in \R^{N-1\times N-1}$ and
\begin{equation}\label{proj-N}
{\bf P}_N = [\, {\bf I}_{N-1} | -{\bf e}' ]\in \R^{N-1\times N}
\end{equation}
with ${\bf e}'=(1,\dots , 1)^{\sf T} \in \R^{N-1}$ and ${\bf I}_{N-1}$ the identity on $\R^{N-1}$.
Column by column, \eqref{Closure-Fick} means
\[
\vec{j}^k = \, - {\bf P}_N^{\sf T} \, {\bf L}' \,  {\bf P}_N \, \partial_k \frac{\vec{\mu}}{RT}
\quad \mbox{ for $k=1,2,3$}
\]
with $\vec{j}^k$ the vector of all $k^{\rm th}$ components of the fluxes, $\partial_k$ denotes the partial derivative w.r.\ to the $k^{\rm th}$ spatial coordinate and $\vec{\mu}=(\mu_1,\ldots ,\mu_N)^{\sf T}={\boldsymbol\mu}^{\sf T}$.\\[1ex]
This Fick-Onsager closure within classical T.I.P.\ provides a thermodynamically consistent coupling
between the partial mass balances and the momentum and energy balances. It also yields explicit formulas for the ${\bf j}_i$,
i.e.\ no tedious inversion of a system of equations is required.
But this form of the closure also has two main disadvantages:
\begin{enumerate}
\item
For real mixtures, the phenomenological coefficients $L_{ik}$ show a complex nonlinear dependence on the composition. 
This is observed in experiments (cf.\ \cite{TK-book}),
but also related to the qualitative behavior of solutions of the resulting PDE systems.
E.g., positivity of solutions may, in general, not persist globally in time for
constant $L_{ik}$.
\item
Since the symmetry w.r.\ to the $\mu_k$ in the closure above is only restored afterwards by extension of ${\bf L}'$
to $\bf L$ via \eqref{Lext1} and \eqref{Lext2}, a full matrix ${\bf L}'$, hence also $\bf L$, is to be expected.
In particular, the special case of a diagonal matrix ${\bf L}'$, while, e.g., appropriate
for dilute solutions, does not cover the general case.
In other words, possible cross-effects between the constituents interfere with the necessary coupling enforced by \eqref{flux-constraint}.
\end{enumerate}
The aspect of positivity of solutions will be discussed in more detail in section~\ref{positivity} below,
leading to structural information about the $L_{ik}$.

In order to arrive at a closed system of partial differential equations (PDEs), there remains the
highly non-trivial task to model the functional dependencies of all $L_{ik}$ on $(T,\rho_1, \ldots ,\rho_N)$.
In our opinion, the main purpose of the Maxwell-Stefan approach is precisely to provide a framework
for this modeling of the functional dependencies of the phenomenological coefficients
in a physically consistent manner; in particular in such a way
that positivity of concentrations is preserved during the time evolution. 
\section{Maxwell-Stefan Closure}\label{section-MS}
In his classical paper \cite{Max}, James Clerk Maxwell used kinetic theory of gases
to derive a relation for the diffusion velocity of a binary mixture of simple gases.
Shortly after this, Josef Stefan essentially gave a continuum mechanical
derivation in \cite{Stef}, valid for a system of $N$ constituents. He employed the assumption
that every particle of a gas, if it is moving, encounters a
resistive force by every other gas, being proportional to the density of that gas and to the
relative velocity between the two.
He used this to formulate partial momentum balances which, in the diffusion approximation, lead to
\begin{equation}\label{MS-historical}
-\sum_{k\neq i} f_{ik} \rho_i \rho_k ({\bf u}_i -{\bf u}_k) = \nabla p_i,
\end{equation}
where $p_i$ are the partial pressures.

Later on, building on the work of Ludwig Boltzmann, i.e.\ the Boltzmann equations,
use of statistical mechanics gave rise to several types
of approximate solutions to the multi-species Boltzmann equations.
Hirschfelder, Curtiss and Bird in particular obtained the
so-called generalized driving forces ${\bf d}_i$ which were to
replace the partial pressure $p_i$; see
\cite{Hirsch} and cf.\ also \cite{Bird}. These have then been used
to formulate the reduced force balances (diffusional approximation), i.e.\ equation \eqref{MS-historical} but with right-hand side
${\bf d}_i$. An equivalent version reads as
\begin{equation}
\label{MS-naiv2}
-\sum_{k\neq i} \frac{x_k {\bf j}_i^{\rm mol} - x_i {\bf j}_k^{\rm mol}}{c\, \D_{ik}} = {\bf d}_i.
\end{equation}
The system \eqref{MS-naiv2} is nowadays referred to as the Maxwell-Stefan equations.

A short derivation of the Maxwell-Stefan equations
within T.I.P.\ employs the so-called resistance form in which
the role of the co-factors in the diffusional entropy production is exchanged; cf.\ \cite{KB-book}.
To start with, we have
\[
    \zeta_{\rm DIFF}  \, =\,
- \sum_{i=1}^{N} {\bf j}_i \cdot \big( \nabla \frac{\mu_i}{T} - \Lambda \big)
\]
for any vector field $\Lambda$ due to \eqref{flux-constraint}, which we rewrite as
\begin{equation}\label{entropy-MSa}
    \zeta_{\rm DIFF}  \, =\,
- \sum_{i=1}^{N} {\bf u}_i \cdot \rho_i \big( \nabla \frac{\mu_i}{T} - \Lambda \big).
\end{equation}
By choosing $\Lambda = \sum_k y_k \nabla \frac{\mu_k}{T}$,
this yields the alternative representation of the diffusional entropy production as
\begin{equation}\label{entropy-MSb}
    \zeta_{\rm DIFF}  \, =\,
- R \,\sum_{i=1}^{N} {\bf u}_i \cdot {\bf d}_i,
\end{equation}
where
\begin{equation}\label{di}
{\bf d}_i = \rho_i \big( \nabla \frac{\mu_i}{RT} - \sum_{k=1}^{N} y_k \nabla \frac{\mu_k}{RT} \big)
\end{equation}
satisfies
\begin{equation}\label{d-sum}
\sum_{i=1}^N {\bf d}_i =0.
\end{equation}
Above, we again expanded by the universal gas constant $R$ in order to get the combination $RT$ inside
the driving forces.

Eliminating ${\bf d}_N$ by means of \eqref{d-sum}, the Maxwell-Stefan system follows from the linear closure
\begin{equation}\label{Closure-MS-i}
{\bf d}_i  \, =\, - \sum_{k=1}^{N-1} \tau_{ik} \big( {\bf u}_k - {\bf u}_N \big) \quad \mbox{ for } \; i=1,\ldots ,N-1
\end{equation}
with a positive definite matrix ${\boldsymbol \tau}' =[\tau_{ik}]\in \rr^{N-1\times N-1}$ of phenomenological coefficients.
Extending the matrix ${\boldsymbol \tau}'$ to an $N\times N$-matrix ${\boldsymbol \tau}$
in a manner fully analogous to the extension
of ${\bf L}'$ to ${\bf L}$, i.e.\ such that
\begin{equation}\label{tau-extension}
\sum_{i=1}^{N} \tau_{ik}=0 \; \mbox{ for all $k$ \quad and } \quad \sum_{k=1}^{N} \tau_{ik}=0 \; \mbox{ for all $i$},
\end{equation}
we obtain
\begin{equation}\label{Closure-MS-ii}
{\bf d}_i  \, =\, \sum_{k=1}^{N} \tau_{ik} ( {\bf u}_i - {\bf u}_k )  \quad \mbox{ for } \; i=1,\ldots ,N.
\end{equation}
The ''interaction'' coefficients $\tau_{ik}$ are functions of the thermodynamic state variables,
i.e.\ $\tau_{ik}=\tau_{ik}(T,\rho_1, \ldots , \rho_N)$; in particular, they depend on the composition.

From here on, we will also assume symmetry of the interaction coefficients, i.e.\ $\tau_{ik} = \tau_{ki}$.
Usually, as in case of the Fick-Onsager closure, this assumption is added referring to Onsager symmetry.
Let us note that this symmetry necessarily holds in case of binary-type interactions, i.e.\ if
\begin{equation}
\label{binary-interaction}
\tau_{ik}=\tau_{ik}(T, \rho_i, \rho_k) \to 0 \quad \mbox{ whenever } \; \rho_i \rho_k \to 0+.
\end{equation}
Indeed, adapting an argument from \cite{CT62}, insertion of \eqref{Closure-MS-ii} into \eqref{d-sum} yields
\[
0 = \sum_{i,k=1}^N \tau_{ik} ( {\bf u}_i - {\bf u}_k ) = \sum_{1\leq i<k \leq N} (\tau_{ik} - \tau_{ki} )( {\bf u}_i - {\bf u}_k )
\]
for any thermodynamic process which the mixture is undergoing.
Considering processes in which  $\rho_j =0$ for all $j\neq i,k$, \eqref{binary-interaction} implies
\[
(\tau_{ik} - \tau_{ki} )( {\bf u}_i - {\bf u}_k ) = 0.
\]
Hence $\tau_{ik} = \tau_{ki}$ if $\rho_j=0$ for all $j\neq i,k$. Now, since $\tau_{ik}$ only depends on $\rho_i$ and $\rho_k$
under the assumption \eqref{binary-interaction}, this yields $\tau_{ik} = \tau_{ki}$ independently of $(\rho_1, \ldots ,\rho_N)$.

Evidently, the assumption \eqref{binary-interaction} of binary interactions (hence symmetry of $[\tau_{ik} ]$)
is satisfied by letting
\begin{equation}\label{fric-coeff}
\tau_{ik} = -\rho f_{ik} y_i y_k \quad \mbox{ for all $i,k=1,\ldots ,N$  with } i\neq k,
\end{equation}
where the phenomenological coefficients $f_{ik}$ ($i\neq k$) satisfy $f_{ik}=f_{ki}$.
Below, we relax the assumption that the $f_{ik}$ only depend on $\rho_i,\, \rho_k$ but allow $f_{ik}$ to depend on the thermodynamic state variables $(T, \rho, {\bf y}')$.


This dependence is assumed to be \underline{regular}, by
which we here mean that the $f_{ik}$ are smooth in all variables
and bounded in ${\bf y}'$. The explicit factor $\rho$ in \eqref{fric-coeff} refers to the fact that we model mass diffusion fluxes $\rho_i {\bf u}_i$ and leads to a slightly simpler final form.
Let us note in passing that more general relations, for instance $\tau_{ik} = -\rho f_{ik} y_i^{\alpha} y_k^{\alpha}$ with $\alpha > 0$, would also be possible
and might be interesting if fast or slow diffusion is to be modeled.

Given such $f_{ik}$ for $i\neq k$, the $f_{ii}$ 
have to be chosen in such a way that
\eqref{tau-extension} is fulfilled, i.e.\ such that
\begin{equation}
f_{ii}y_i + \sum_{k\neq i} f_{ik} y_k =0 \;\mbox{ for $y_1,\ldots ,y_N >0$ with } \sum_{l=1}^N y_l =1.
\end{equation}
This is always possible since the $f_{ii}$ are irrelevant for the sum in \eqref{Closure-MS-ii}.

Let us sum up the assumptions on ${\boldsymbol \tau}=[\tau_{ik} ]\in \rr^{N\times N}$:
\begin{equation}\label{fik-assumptions}
{\boldsymbol \tau}={\boldsymbol \tau}^{\sf T},\;\;
{\boldsymbol \tau} {\bf e}=0,\;\;
\langle {\boldsymbol \tau} {\bf z}, {\bf z} \rangle >0
\;\; \forall \; 0\neq {\bf z}\in \{{\bf e} \}^\perp,\;\;
\tau_{ik} = -\rho f_{ik} y_i y_k \; (i\neq k)
\end{equation}
with $f_{ik}$ being regular functions of $(T,\rho, {\bf y}')$, resp.\ of $(T,\rho_1, \ldots ,\rho_N)$, for all $i\neq k$.
Concerning the positive definiteness of ${\boldsymbol \tau}$ on $\{\bf e \}^\perp$, observe that this follows from
the positive definiteness of ${\boldsymbol \tau}'$ on $\rr^{N-1}$ since
\begin{equation}\label{tau-tau}
\langle {\boldsymbol \tau} {\bf z}, {\bf z} \rangle =
\sum_{i,k=1}^N \tau_{ik} z_i z_k = \sum_{i,k=1}^{N-1} \tau_{ik} (z_i -z_N) ( z_k-z_N)=
\langle {\boldsymbol \tau}' {\bf z}', {\bf z}' \rangle
\end{equation}
with $z_i' = z_i -z_N$ for $i=1,\ldots ,N-1$, and $0\neq {\bf z}\in \{\bf e \}^\perp$
implies ${\bf z}'\neq 0$ for ${\bf z}' \in \rr^{N-1}$.
Vice versa, this also shows that ${\boldsymbol \tau}'$ is positive definite on $\rr^{N-1}$
if ${\boldsymbol \tau}$ is positive definite on $\{\bf e \}^\perp$, i.e.\ consistency of \eqref{tau-tau}
with the original closure
in \eqref{Closure-MS-i}: given $0\neq {\bf z}' \in \rr^{N-1}$, let $z_N=-\frac 1 N \langle {\bf z}',{\bf e}' \rangle$,
$z_i=z_i' +z_N$ for $i<N$ and apply \eqref{fik-assumptions} and \eqref{tau-tau}.

The $f_{ik}$ are usually interpreted as ''friction coefficients'', hence $f_{ik}>0$ would be a reasonable assumption,
explaining the minus sign in \eqref{fric-coeff}.
Let us note that, while $f_{ik}>0$ is \underline{not} implied by \eqref{fik-assumptions}, the converse holds
for non-vanishing mass fractions $y_i$ for all constituents:
if ${\boldsymbol \tau}$ has off-diagonal entries given by \eqref{fric-coeff} with $f_{ik}=f_{ki}>0$ for all $i\neq k$
and diagonal entries such that \eqref{tau-extension} holds, then ${\boldsymbol \tau}$ has all properties listed in
\eqref{fik-assumptions}. Indeed, this follows from
\begin{equation}\label{entropy-MSc}
\langle {\boldsymbol \tau} {\bf z}, {\bf z} \rangle =
 \,\sum_{1\leq i<k \leq N} \rho f_{ik} y_i y_k (z_i - z_k)^2,
\end{equation}
which shows that $\langle {\boldsymbol \tau} {\bf z}, {\bf z} \rangle\geq 0$ and $\langle {\boldsymbol \tau} {\bf z}, {\bf z} \rangle=0$ only if $z_i=z_k$ for all $i\neq k$ in which case ${\bf z}\in {\rm span}({\bf e})$.

Combining \eqref{di}, \eqref{Closure-MS-ii} and \eqref{fric-coeff}, the Maxwell-Stefan equations in
a mass-based form read as
\begin{equation}\label{MS-mass-based}
- \sum_{k=1}^{N} f_{ik} (y_k  {\bf j}_i - y_i {\bf j}_k )   \, =\, \rho_i \big( \nabla \frac{\mu_i}{RT} - \sum_{k=1}^{N} y_k \nabla \frac{\mu_k}{RT} \big) \quad \mbox{ for } i=1,\ldots ,N.
\end{equation}
The right-hand side of \eqref{MS-mass-based} can be rewritten, using the Gibbs-Duhem relation \eqref{Gibbs-Duhem}.
For this purpose, let $g:=\psi +p/\rho$
denote the specific Gibbs free energy which satisfies $g=\sum_{k=1}^N y_k \mu_k$ due to \eqref{Gibbs-Duhem}.
As is well-known, $g$ is related to $\psi$ (and $s$) by means of a Legendre transform such that
\eqref{entropy-derivatives} implies
\[
dg = -s \, dT + \frac 1 \rho dp + \sum_{k=1}^N \mu_k dy_k \quad \mbox{ with } \sum_{k=1}^N y_k =1.
\]
Together with $dg = d(\sum_{k=1}^N y_k \mu_k)$ this yields
\[
\sum_{k=1}^N  y_k \nabla \mu_k = -s \nabla T + \frac 1 \rho \nabla p,
\]
hence
\begin{equation}\label{GD-grad}
\sum_{k=1}^N  y_k \nabla \frac{\mu_k}{T} = \frac 1 {\rho T} \nabla p + h \nabla \frac 1 T,
\end{equation}
where $h:=e+p/\rho$ is the specific enthalpy.
Consequently, a second form of the mass-based Maxwell-Stefan equations reads as
\begin{equation}\label{MS-mass-based2}
- \sum_{k=1}^{N} f_{ik} (y_k  {\bf j}_i -  y_i  {\bf j}_k )   \, =\, \rho_i \nabla \frac{\mu_i}{RT}
- \frac{y_i}{RT} \nabla p - \frac{\rho_i h}{R}  \nabla \frac 1 T \quad \mbox{ for } i=1,\ldots ,N.
\end{equation}
The right-hand side of \eqref{MS-mass-based2} defines--up to the factor $1/R$--the so-called generalized thermodynamic driving forces.
The different contributions are attributed to compositional (also called molecular) diffusion, pressure diffusion and thermal diffusion, in the order of their appearance. If individual body forces ${\bf b}_i$ are present, the additional term
$\rho_i \frac{{\bf b}_i - {\bf b}}{R T}$ appears on right-hand side, inducing so-called forced diffusion.
The latter is for instance present in transport processes involving charged species (ions) due to the intrinsic electrical field.
Recall that we do not include thermo-diffusive coupling for technical simplicity; otherwise, additional terms would appear in
\eqref{MS-mass-based2}.

Let us note in passing that a more refined class-II model yields the same expression but with the partial enthalpy
$\rho_i h_i$ instead of $\rho_i h$ in the last term; see \cite{BD}.
This is consistent with kinetic gas theory in terms of the multi-species Boltzmann equations, from which
the generalized Maxwell-Stefan equations have been originally derived; cf.\ \cite{Hirsch}.
Let us also note that in Chemical Engineering, the molar-based variant of the Maxwell-Stefan equations is more common.
Later, we will need this form as well which is therefore included in Appendix~\ref{MS-ChemEng}.

We rewrite \eqref{MS-mass-based} in tensorial notation as
\begin{equation}\label{Closure-MSa}
- {\bf B} \, {\bf J} \, =\,
\mathbf{R} \, {\bf P} \, \nabla \frac{\boldsymbol{\mu}}{RT}
\end{equation}
with ${\bf B}={\bf B}(T,\rho, {\bf y})=[B_{ij}(T,\rho, {\bf y})]$, where ${\bf y}=(y_1,\ldots ,y_N)$ with $\sum_{i=1}^N y_i =1$,
\begin{equation}\label{MS-matrix}
B_{ij}=- y_i f_{ij} \;\mbox{ for } i\neq j, \quad B_{ii}= \sum_{k\neq i} y_k f_{ik},
\end{equation}
$\mathbf{R}={\rm diag}(\rho_1, \ldots , \rho_N)$ and ${\bf P}$ denotes the projection
\begin{equation}\label{projection}
{\bf P}= {\bf I} - {\bf e}\otimes {\bf y} \; \mbox{ with } {\bf e}=(1, \ldots, 1)^{\sf T} \in \rr^N.
\end{equation}
From here on we write ${\bf B}({\bf y})$ to stress the fact that the entries $B_{ij}$ depend in particular on the composition.
Note also that, with ${\bf Y}={\rm diag} (y_1, \ldots ,y_N)$,
\begin{equation}\label{Commute}
\mathbf{R} \, {\bf P} = \rho {\bf Y} \big[ {\bf I} - {\bf e}\otimes {\bf y} \big] =
\rho \big[ {\bf I} - {\bf y}\otimes {\bf e} \big] {\bf Y} = {\bf P}^{\sf T} \mathbf{R}.
\end{equation}
Since
\[
{\bf e}^{\sf T} {\bf J}=0
\quad \mbox{and} \quad
{\bf e}^{\sf T} \mathbf{R} \, {\bf P} \, \nabla \frac{\boldsymbol{\mu}}{RT}
= {\bf e}^{\sf T} {\bf P}^{\sf T} \, \mathbf{R} \, \nabla \frac{\boldsymbol{\mu}}{RT}=0,
\]
equation \eqref{Closure-MSa} means to solve
\begin{equation}\label{Bzd}
- {\bf B}({\bf y}) \, {\bf z} = {\bf d}
\end{equation}
for given right-hand side ${\bf d}\in \{\bf e \}^\perp$ such that the solution satisfies ${\bf z}\in \{\bf e \}^\perp$.
Since
\begin{equation}\label{imkerB}
{\rm im} ({\bf B}({\bf y}))=\{\bf e \}^\perp \quad \mbox{and}\quad {\rm ker} ({\bf B}({\bf y}))={\rm span} \{ \bf y \},
\end{equation}
equation \eqref{Closure-MSa} cannot be solved by inversion of ${\bf B}({\bf y})$ as a map on all of $\rr^N$.
To resolve the Maxwell-Stefan equations we hence make use of generalized inverse matrices. In the context of multicomponent diffusion, this approach was introduced by Giovangigli, cf.\ \cite{Giovan}. In the case of positive fractions $y_i$, the relations \eqref{imkerB} show that ${\bf B}({\bf y})$ possesses rank $N-1$, and that the zero eigenvalue is associated with the strictly positive right-eigenvector ${\bf y}$ and left-eigenvector ${\bf e}$. In this situation, we can introduce the unique \emph{group inverse} ${\bf B}^{\sharp} = {\bf B}^{\sharp}({\bf y}) $ of ${\bf B}({\bf y})$. Among other properties, it satisfies ${\bf B}^{\sharp} \, {\bf B} = {\bf I} - {\bf y} \otimes {\bf e} = {\bf B} \, {\bf B}^{\sharp}$, and ${\bf B}^{\sharp} {\bf y} = 0 = ({\bf B}^{\sharp})^{\sf T} {\bf e}$ (see the Appendix, section \ref{Drazin}).

Applying the group inverse to \eqref{Closure-MSa}, we obtain that
\begin{equation}\label{Closure-MSJbbis}
{\bf J}= - {\bf B}^{\sharp} ({\bf y}) \, \mathbf{R} \, {\bf P} \, \nabla \frac{\boldsymbol{\mu}}{RT}
= - {\bf B}^{\sharp} ({\bf y}) \, \mathbf{R} \, \nabla \frac{\boldsymbol{\mu}}{RT}  \, .
\end{equation}
Using the properties of the generalized inverse, it is a short exercise to prove the relationship
\begin{align}\label{Bandtau}
{\bf B}^{\sharp}({\bf y}) = {\bf P}^{\sf T} \, {\bf R} \, \boldsymbol{\tau}^{\sharp} \, {\bf P}^{\sf T} \, ,
\end{align}
where $\boldsymbol{\tau} = {\bf B}({\bf y}) \, {\bf R}$ is the matrix satisfying \eqref{fik-assumptions}.
Consequently, \eqref{Commute} shows that
\begin{align}\label{Bandtau2}
{\bf B}^{\sharp}({\bf y}) \, {\bf R} = {\bf P}^{\sf T} \, {\bf R} \, \boldsymbol{\tau}^{\sharp} \, {\bf R} \, {\bf P} \, ,
\end{align}
which establishes the symmetry and the positivity of ${\bf B}^{\sharp}({\bf y}) \, {\bf R}$ from the natural properties of $\boldsymbol{\tau}$.
Appendix~\ref{Drazin} provides additional facts on the generalized inverse. Estimates concerning the representation \eqref{Closure-MSJbbis} and the relationship of ${\bf B}^{\sharp}$ to the coefficients of the Fick-Onsager matrix are discussed in our main theorems below.
\section{A Novel Consistent Closure Scheme}\label{sec-novel-scheme}
While constant $f_{ik}$ are, in principle, admissible and--as will become clear in section~\ref{positivity} below--also lead to preservation of positivity of the solutions of the partial mass balances, data from experimental measurements as well as from
molecular dynamics simulations show that the $f_{ik}$ depend on the composition in a non-trivial manner; see, e.g., \cite{TK-book}.
Furthermore, a disadvantage of the Maxwell-Stefan approach is that the fluxes are defined implicitly, requiring
the inversion of the Maxwell-Stefan equations which is computationally expensive.
For small systems, the inversion is usually done by eliminating one of the fluxes,
i.e.\ in an analogous way as in the generalized Fick-Onsager approach. This, again, breaks the symmetry w.r.\ to the constituents and complicates the required linear algebra considerably.

We aim for a closure as simple as the Fick-Onsager one, but having the advantages of the Maxwell-Stefan approach.
The new scheme is based on two simple ideas, both related to a certain symmetry aspect:
\begin{enumerate}
\item
The constraint \eqref{flux-constraint} on the fluxes should be incorporated without breaking
the symmetry w.r.\ to the constituents
\item
The decomposition of the binary products in the diffusive entropy production into co-factors should be done in a symmetric way.
\end{enumerate}

To avoid breaking the symmetry w.r.\ to the $A_i$, instead of incorporating the constraint \eqref{flux-constraint},
we prevent interference of the latter with the core closure process
by starting with general diffusion velocities taken against an undetermined reference velocity ${\bf v}^\ast$.
This way, we also keep an advantage present in the Maxwell-Stefan approach, namely a closure which is independent of a specific reference system chosen to define specific diffusion fluxes.
We therefore consider diffusion fluxes according to
\[
{\bf j}_i^\ast = \rho_i ({\bf v}_i -{\bf v}^\ast).
\]
This introduces one more unknown vector field, compensating for the additional equation \eqref{flux-constraint}.
Let us also note in passing that we need to consider a relative velocity for obtaining an objective constitutive quantity.

Evidently, the barycentric diffusion fluxes ${\bf j}_i$ are then given as
\begin{equation}\label{velo-relation}
{\bf j}_i =   {\bf j}_i^\ast -  y_i \sum_{k=1}^N {\bf j}_k^\ast, \; \mbox{ or } \;
{\bf J} = {\bf P}^{\sf T}\, {\bf J}^\ast
\end{equation}
with the projection
\[
{\bf P}^{\sf T}={\bf I}-{\bf y}\otimes {\bf e};
\]
recall that the transposed projection
${\bf P}={\bf I}-{\bf e}\otimes {\bf y}$ has already been introduced in the context of the Maxwell-Stefan closure.
We then have
\begin{equation}\label{entropy-diffusion4}
    \frac 1 R \zeta_{\rm DIFF}  \, =\,
- \langle {\bf P}^{\sf T} \,  {\bf J}^\ast ,   \nabla \frac{\boldsymbol{\mu}}{RT} \rangle
= - \langle  {\bf J}^\ast  , {\bf P} \, \nabla \frac{\boldsymbol{\mu}}{RT}\rangle.
\end{equation}
Since the ${\bf j}_i^\ast$, in contrast to the barycentric diffusion fluxes ${\bf j}_i$,
are unconstrained, we can employ the linear (in the co-factors) closure for the full system of diffusion fluxes,
i.e.\ we let
\begin{equation}\label{closure-BP0}
{\bf J}^\ast  \, =\, - {\bf L}\, {\bf P} \, \nabla \frac{\boldsymbol{\mu}}{RT}
\end{equation}
with a positive definite and symmetric matrix ${\bf L}=[L_{ij}]$. Hence
\begin{equation}\label{closure-BP0b}
{\bf J}  \, =\, - {\bf P}^{\sf T}\, {\bf L}\, {\bf P} \,\nabla \frac{\boldsymbol{\mu}}{RT}.
\end{equation}
At this point note that, here, in contrast to both the Fick-Onsager and the Maxwell-Stefan closure,
a diagonal closure (i.e., with diagonal ${\bf L}$) is possible, ignoring additional cross-effects which are not covered by the projections. This ''core-diagonal'' special case is interesting in itself and will be studied in section~\ref{Darken} below.

But one difficulty remains to be resolved: in the constitutive relation \eqref{closure-BP0},
the phenomenological coefficients $L_{ij}$ are not diffusivities, having the physical units ${\rm kg}^2 {\rm m}^{-1}  {\rm s}^{-1} {\rm mol}^{-1}$ instead of ${\rm m}^2 {\rm s}^{-1}$.
This indicates that the $L_{ij}$ contain, in particular, factors with the dimension of mass densities.
Such factors are in fact required since without additional structure, positivity of solutions to the resulting final PDE system cannot persist as will be explained in more detail in section~\ref{positivity}.
Consequently, in order to minimize the need for composition-dependence of the phenomenological coefficients,
we have to slightly adjust the approach.
At this point, observe that the physical dimension of the phenomenological coefficients can be changed by shuffling factors within the binary products. Recall that the latter was an important ingredient for the Maxwell-Stefan closure in section~\ref{section-MS}. Employing the corresponding representation of the diffusional entropy production according to
\begin{equation}\label{entropy-diffusion5}
    \frac 1 R \zeta_{\rm DIFF}  \, =\,
- \langle  {\bf U}  ,  {\bf R} \, \nabla \frac{\boldsymbol{\mu}}{RT}\rangle \, =\,
 - \langle  {\bf U}^\ast  , {\bf P}^{\sf T} \, {\bf R} \, \nabla \frac{\boldsymbol{\mu}}{RT}\rangle,
\end{equation}
where  ${\bf U}=[{\bf u}_1 | \cdots | {\bf u}_N]^{\sf T}$ and
${\bf U}^\ast=[{\bf u}_1^\ast | \cdots | {\bf u}_N^\ast]^{\sf T}$ with  ${\bf u}_i^\ast={\bf v}_i-{\bf v}^\ast$,
we arrive at the closure
\begin{equation}\label{closure-BP1}
{\bf J}  \, =\, - {\bf R}\, {\bf P} \, {\bf L}\, {\bf P}^{\sf T} \, {\bf R} \,\nabla \frac{\boldsymbol{\mu}}{RT}.
\end{equation}
In this constitutive relation, the $L_{ij}$ have physical dimension ${\rm m}^5 {\rm s}^{-1} {\rm mol}^{-1}$, i.e.\ according to these units, they shall contain reciprocals of concentrations as factors.

In order to obtain diffusivities as phenomenological coefficients, thus avoiding the need for hidden factors
(mass densities or molar concentration) which introduce strong dependence on composition,
we once again start from
\begin{equation}\label{entropy-diffusion5b}
    \frac 1 R \zeta_{\rm DIFF}  \, =\,
 - \langle  {\bf U}^\ast  , {\bf P}^{\sf T} \, {\bf R} \, \nabla \frac{\boldsymbol{\mu}}{RT}\rangle  \, =\,
 - \langle  {\bf U}^\ast  , {\bf R} \, {\bf P} \, \nabla \frac{\boldsymbol{\mu}}{RT}\rangle
\end{equation}
but change from $\mu_i$ to $\mu_i^{\rm mol}=M_i \mu_i$,
leading to
\begin{equation}\label{entropy-diffusion6}
    \frac 1 R \zeta_{\rm DIFF}  \, =\,
 - \langle  {\bf U}^\ast  , {\bf C} \, {\bf P}_{\rm mol} \, \nabla \frac{\boldsymbol{\mu}^{\rm mol}}{RT}\rangle
\end{equation}
with the projection
\[
{\bf P}_{\rm mol} \, =\, {\bf M} \, {\bf P} \, {\bf M}^{-1}.
\]
According to (ii) above, we now distribute the factor ${\bf C}={\rm diag}(c_1,\ldots ,c_N)$ symmetrically between the co-factors in \eqref{entropy-diffusion6}.
In order to avoid fractional physical dimensions, we factor out the total concentration $c$, employing ${\bf C}=c {\bf X}$. Hence, we build our closure on the representation
\begin{equation}\label{entropy-diffusion7}
\frac 1 {Rc} \zeta_{\rm DIFF}  \, =\,
 -  \, \langle  {\bf X}^{1/2} {\bf U}^\ast  , {\bf X}^{1/2}\, {\bf P}_{\rm mol} \, \nabla \frac{\boldsymbol{\mu}^{\rm mol}}{RT}\rangle.
\end{equation}
Linear (in the co-factors) closure gives
\begin{equation}\label{closure-B0}
{\bf X}^{1/2} {\bf U}^\ast  \, =\, - {\bf D}\, {\bf X}^{1/2}\, {\bf P}_{\rm mol} \, \nabla \frac{\boldsymbol{\mu}^{\rm mol}}{RT}
\end{equation}
with a symmetric matrix ${\bf D}$ of diffusivities $D_{ij}$ which is positive definite on $\{ \sqrt{\bf x} \}^\perp$,
where $\sqrt{\bf x}:=(\sqrt{x_1},\ldots , \sqrt{x_N})$; note that
\begin{equation}\label{rangexhalfpm}
{\rm im}({\bf X}^{1/2}\, {\bf P}_{\rm mol})={\rm ker}({\bf P}_{\rm mol}^{\sf T} \, {\bf X}^{1/2})^\perp
=\{ {\bf z}: {\bf P}^{\sf T}\,{\bf M}\,{\bf X}^{1/2}\, {\bf z}=0 \}^\perp
=\{ \sqrt{\bf x} \}^\perp.
\end{equation}
This yields diffusion velocities according to
\begin{equation}\label{closure-B1}
 {\bf U}  \, =\, -{\bf P} \, {\bf X}^{-1/2}\, {\bf D}\, {\bf X}^{1/2}\, {\bf P}_{\rm mol} \, \nabla \frac{\boldsymbol{\mu}^{\rm mol}}{RT}.
\end{equation}
In this formulation, where the necessary couplings due to conservation of total mass are accounted for
by the projections, the off-diagonal elements of ${\bf D}$ model what one might call ''true cross-effects''
which are not enforced by \eqref{flux-constraint}.
Concerning these cross-effects, we assume them--similar to interactions of binary type but allowing for dependence on the full composition--to vanish if one of the involved constituent is absent, i.e.
\begin{equation}
D_{ij}\to 0 \;\mbox{for $i\neq j$ whenever }\; x_i x_j \to 0.
\end{equation}
Because of symmetry of ${\bf D}$ and the form of the factors immediately left and right of ${\bf D}$, we incorporate this by assuming
\begin{equation}\label{novel-D}
{\bf D} = {\mathcal D} + {\bf X}^{1/2}\, {\bf K}\, {\bf X}^{1/2},
\end{equation}
where ${\mathcal D}={\rm diag}(d_1, \ldots ,d_N)$, ${\bf K}={\bf K}^{\sf T}$ and ${K}_{ii}=0$ for all $i$.
Insertion of this structure for ${\bf D}$ into \eqref{closure-B1} yields the closure for mass fluxes according to
\begin{equation}\label{closure-B2}
 {\bf J}  \, =\, -{\bf P}^{\sf T} \, {\bf R}\, [{\mathcal D}+ {\bf K}\, {\bf X}]\, {\bf M}\, {\bf P} \, \nabla \frac{\boldsymbol{\mu}}{RT}.
\end{equation}
Notice that the full coefficient matrix is symmetric since ${\bf X}\, {\bf M}={\bf R}/c$.
Equivalently, the fluxes according to \eqref{closure-B2} are of the form
\begin{equation}\label{closure-B3}
 {\bf J}  \, =\, -{\bf P}^{\sf T} \,  [\tilde{\mathcal D} + {\bf Y}\, \tilde{\bf K} ]\, {\bf P}^{\sf T} \, {\bf R}\, \nabla \frac{\boldsymbol{\mu}}{RT}
\end{equation}
with $\tilde{\mathcal D}:={\mathcal D}\, {\bf M}$ and $\tilde{\bf K}:=\frac{\rho}{c} \, {\bf K}$.

Now notice that $\tilde{\mathcal D}+{\bf Y}\, \tilde{\bf K}$ contains $N+(N-1)N/2$ model parameters, while there are
only $(N-1)N/2$ parameters in both the Fick-Onsager and the Maxwell-Stefan model.
This is consistent, since only the restriction of $\hat{\mathcal D}+{\bf Y}\, \tilde{\bf K}$
to the $(N-1)$-dimensional subspace ${\rm im}({\bf P}^{\sf T})$ is relevant.
Hence the system of diffusion fluxes ${\bf J}$ does not uniquely determine the coefficients of $\tilde{\mathcal D}$ and $\tilde{\bf K}$, respectively of ${\mathcal D}$ and ${\bf K}$,
and there are different options on how to make this choice unique.
While we comment on other options later on, we here exploit the fact that
$\tilde{\bf K}$ in \eqref{closure-B3} can be replaced by
\begin{equation}\label{Kmodify}
\hat{\bf K}:=\tilde{\bf K} + {\bf a}\otimes {\bf e} + {\bf e}\otimes {\bf a}
\end{equation}
without changing ${\bf J}$. The latter follows from
\[
{\bf P}^{\sf T}\, {\bf Y}\, [ {\bf a}\otimes {\bf e} + {\bf e}\otimes {\bf a} ] \, {\bf P}^{\sf T}\, = \, 0
\]
for arbitrary ${\bf a}\in \rr^N$, because
\[
{\rm im}({\bf P}^{\sf T})={\rm ker}({\bf P})^\perp=\{ {\bf e} \}^\perp
\;\mbox{ and } \;
{\rm ker}({\bf P}^{\sf T})= {\rm span}({\bf y}).
\]
Now a natural choice for ${\bf a}$, which does not break the symmetry of the components while preserving diagonal diffusion, is the one which yields $\hat{\bf K}_{\rm off}\, {\bf e}=0$ for the off-diagonal part $\hat{\bf K}_{\rm off}$ of the resulting tensor $\hat{\bf K}$; note that this off-diagonal part will
finally replace $\tilde{\bf K}$, while the diagonal part will be incorporated into $\tilde{\mathcal D}$.
We hence aim at
\[
\hat{\bf K}_{\rm off} {\bf e} \, = \, \big( \hat{\bf K} - 2 \, {\rm diag}({\bf a}) \big) {\bf e} = 0.
\]
The latter means
\[
\tilde{\bf K}{\bf e} + N {\bf a} + \langle {\bf a} , {\bf e}  \rangle {\bf e}  - 2 \, {\bf a} =0,
\]
or, in the relevant case $N>2$,
\[
[ \, {\bf I} + \frac{ {\bf e} \otimes {\bf e}}{N-2} \,]\,  {\bf a}   = - \frac{1}{N-2}  \tilde{\bf K}{\bf e}.
\]
This is invertible to the result
\begin{equation}\label{a-correct}
{\bf a}   = - \, \frac{1}{N-2} \, [ \, {\bf I} - \frac{ {\bf e} \otimes {\bf e}}{2 \, (N-1)} \,]\,  \tilde{\bf K}{\bf e}.
\end{equation}
Hence, modifying $\tilde{\bf K}$ according to \eqref{Kmodify} with ${\bf a}$ from \eqref{a-correct}
and shuffling the additional diagonal part to $\tilde{\mathcal D}$,
we can impose the additional condition that the off-diagonal part has zero row-sums.

Summing up, the novel closure scheme yields diffusion
fluxes of the structure as given in \eqref{closure-B2}, i.e.\
\[
 {\bf J}  \, =\, -{\bf P}^{\sf T} \, {\bf R}\, [{\mathcal D}+ {\bf K}\, {\bf X}]\, {\bf M}\, {\bf P} \, \nabla \frac{\boldsymbol{\mu}}{RT},
\]
respectively \eqref{closure-B3}, i.e.\
\[
 {\bf J}  \, =\,
 -{\bf P}^{\sf T} \,  [{\mathcal D}{\bf M} + \frac \rho c {\bf Y}\, {\bf K} ]\, {\bf P}^{\sf T} \, {\bf R}\, \nabla \frac{\boldsymbol{\mu}}{RT},
\]
where ${\mathcal D}={\rm diag}(d_1, \ldots ,d_N)$
and ${\bf K}$ is off-diagonal (all ${K}_{ii}=0$) with ${\bf K}={\bf K}^{\sf T}$ and ${\bf K} {\bf e}= {\bf 0}$.
Moreover, ${\bf D}= {\mathcal D} + {\bf X}^{1/2}\, {\bf K}\, {\bf X}^{1/2}$ from \eqref{novel-D}
is positive definite on $\{ \sqrt{\bf x} \}^\perp$.
This closure contains $(N-1)N/2$ model parameters as in the other approaches.

Due to the appearance of the projections left and right of the inner matrix product, the diffusion matrix from
\eqref{novel-D} is not uniquely determined by the diffusion fluxes. Correspondingly, different options exist in order to eliminate the $N$ superfluous parameters above.
One possibility would be to require ${\mathcal D}=0$, resulting in a fully off-diagonal inner matrix.
This would be somewhat similar to the off-diagonal closure originally introduced in \cite{CB}, but which has later been regarded as disadvantageous as mentioned in the review \cite{CB-review}.


Another possibility would be to ask for a positive (semi-)definite diffusion matrix ${\bf D}$ in \eqref{novel-D}, for which the additional condition ${\bf D} \, \sqrt{{\bf x}} = \lambda \, \sqrt{{\bf x}}$ with $\lambda \geq 0$ is fulfilled. This choice guarantees that the diagonal part $\mathcal{D}$ has positive entries, but it possesses the severe drawback that diagonal diffusion can occur only in the form $\mathcal{D} = d \, {\bf I}$ for some positive scalar $d$.
In other words, assuming that in \eqref{closure-B2}, the matrix ${\bf D}$ coincides with a diagonal matrix $\mathcal{D}$ on $\{{\bf \sqrt{x}}\}^{\perp}$, the condition ${\bf D} \, \sqrt{{\bf x}} = \lambda \, \sqrt{{\bf x}}$ implies $\mathcal{D} = d \, {\bf I}$. This type of diagonal diffusion has been investigated in Section 7.5 of \cite{Giovan}; see, in particular, Corollary 7.5.6 there.

%
%
%
%
\section{Positivity Requirements}\label{positivity}
The phenomenological coefficients which appear in the different closures for continuum thermodynamical diffusion fluxes
need to fulfill certain structural assumptions concerning their dependencies on the composition in order to allow for positive solutions of the resulting partial
differential equations. Such structural properties will also be required in order to prove the equivalence of the different closures.
There are two different approaches to treat the question of positivity of solutions.
Either the models are only formulated for compositions with all constituents present, i.e.\ $y_i >0$ for all $i=1,\ldots ,N$
and all $(t,x)$, or the models are extended to cover the cases in which partial densities may disappear, i.e.\ $y_i=0$ is allowed.
In the first case it is required to show strictly positive lower bounds on the $y_i$, while one needs to show that $y_i \geq 0$
in the second case.
Extensions to allow constituents to vanish, $y_i=0$, is a topic in itself which is not addressed here;
cf.\ Proposition 7.7.5 in \cite{Giovan}.

We shall employ the Maxwell-Stefan form of the closure to motivate the positivity requirements
since, with non-singular and non-degenerate $f_{ik}$, the Maxwell-Stefan diffusion fluxes
are such that positivity of the partial mass densities is sustained; this has been shown in \cite{DB-MS}, \cite{JP-MS}
for strong solutions.
To understand the structural reason behind, we rewrite \eqref{MS-mass-based} as
\begin{equation}\label{flux-structure}
{\bf j}_i =
- \frac{\rho_i }{\sum_{k\neq i} f_{ik} y_k} \big( \nabla \frac{\mu_i}{RT} - \sum_{k=1}^N y_k \nabla \frac{\mu_k}{RT} \big)
\, +\, y_i\, \frac{\sum_{k\neq i} f_{ik} {\bf j}_k}{\sum_{k\neq i} f_{ik} y_k}
\end{equation}
for $i=1,\ldots ,N$.
Since the chemical potential $\mu_k$ approaches the one for a dilute species  as $x_k\to 0+$,
it holds that $M_k \mu_k$ has the contribution $RT \ln x_k$ as the only singular (as $x_k\to 0+$) term. This condition is obviously rigorous for ideal mixtures in the sense of \eqref{idealchempot}, but it in fact covers a significantly wider class of mixtures. Therefore, the first summand on the right-hand side of \eqref{flux-structure}
is of the type
\begin{equation}\label{flux-structure0}
- d_i (T,\rho,{\bf y}) \nabla y_i + y_i \, {\bf f}_i (T,\rho,  \nabla \rho, {\bf y}, \nabla {\bf y}),
\end{equation}
where $d_i$ is non-degenerate as $y_i \to 0+$. To incorporate the remaining terms on the right-hand side of \eqref{flux-structure},
we plug in the fluxes ${\bf j}_k$ from \eqref{Closure-MSJbbis}.
Exploiting that these fluxes are given by non-singular functions
of $(T,\rho, \nabla \rho, {\bf y}, \nabla {\bf y})$, we get the same structure \eqref{flux-structure0} also
for the full fluxes. Since we always assume $\rho >0$, this can
also be rewritten as
\begin{equation}\label{flux-structure1}
{\bf j}_i \, = \, - d_i (T,\rho,{\bf y}) \nabla \rho_i + \rho_i \, {\bf f}_i (T,\rho,  \nabla \rho, {\bf y}, \nabla {\bf y}),
\end{equation}
where $d_i (T,\rho,{\bf y})\to d_i^0 (T,\rho,y_1,..,y_{i-1},y_{i+1},..,y_N)>0$ as $y_i\to 0+$ and the ${\bf f}_i$ are non-singular.

If fluxes of this structure are plugged into the partial mass balance \eqref{partial-mass2},
the resulting PDE is of the form
\begin{align}\label{PDE-structure}
\pt \rho_i - d_i (T,\rho,{\bf y}) \Delta \rho_i & = \rho_i \, F_i(T,\rho, \nabla \rho, {\bf y}, \nabla {\bf y}, {\bf \nabla v})\\
& + \nabla \rho_i \cdot {\bf G}_i (T,\rho, \nabla \rho, {\bf y}, \nabla {\bf y}, {\bf v})+ r_i.\nonumber
\end{align}
Now, if a sufficiently regular (classical, say)
solution starting from a strictly positive initial value reaches a time $t_0$ where
$\min_\Omega \rho_i (t_0,\cdot)=0$ for the first time and, for simplicity, zero is attained at an interior point $\xi_0$, then $\Delta \rho_i (t_0,\xi_0)\geq 0$, $\rho_i (t_0,\xi_0)=0$ and $\nabla \rho_i (t_0,\xi_0)=0$. Hence
\[
\pt \rho_i (t_0,\xi_0) \geq r_i (T(t_0,\xi_0), \rho(t_0,\xi_0), {\bf y}(t_0,\xi_0));
\]
this argument shows that, in fact, $d_i \geq 0$ as $y_i \rightarrow 0+$ is sufficient.
If we further impose the standard assumption that the reaction rates $r_i$ are quasi-positive, meaning that
\[
r_i (T,\rho, {\bf y})\geq 0 \quad \mbox{whenever } y_i=0,
\]
we obtain
\[
\pt \rho_i (t_0,\xi_0) \geq 0.
\]
With an additional approximation argument, this yields the non-negativity of $\rho_i$ for any given regular solution on its time interval of existence.
Let us note that quasi-positivity holds, in particular, in the realistic case that $r_i=r_i^f-x_i r_i^b$ with $r_i^f, r_i^b \geq 0$.
A mathematical theory which applies to much more general PDE-systems and yields non-negativity for less regular (weak $L^p$-) solutions under (positivity-)con\-ditions which are implied by the structural assumptions above can
be found in \cite{amann89, amann93}.

Since we are not considering cases which allow for finite-time extinction of species, it is interesting to see that, at least for regular solutions, the same structural property leads to a control of the lower bounds of strictly
positive solutions. For this purpose, consider the functions
\[
 m_i (t)=\min \{\rho_i (t,x):x\in \overline{\Omega}\}
\]
and employ the fact that
\begin{equation}\label{miderivative}
m_i'(t) \geq \min \{ \partial_t \rho_i (t,\xi ): \xi \in \overline{\Omega} \mbox{ such that } \rho_i (t,\xi )=m_i (t) \}
\;\mbox{ for a.e. } t
\end{equation}
for sufficiently regular functions and under homogeneous standard boundary conditions, e.g.\ Dirichlet, Neumann or Robin;
see \cite{ConstantinEscher} for a proof of \eqref{miderivative}.
Application of \eqref{miderivative} to classical solutions of \eqref{PDE-structure} leads to
\[
m_i'(t) \geq - \, K(t) \, m_i (t),
\]
where $K(t)$ is a bound for $F_i(T,\rho, \nabla \rho, \, {\bf y}, \nabla {\bf y}, \nabla {\bf  v})$.
In other words, under the structural property \eqref{flux-structure1}, $L^\infty$-bounds on
the $\rho_i$, their gradients and on $\nabla {\bf  v}$ imply strictly positive lower bounds for the $\rho_i$ such that
extinction of individual species is not possible in finite time.

\indent
To sum up, despite the fact that we do not consider the case in which constituents become absent (i.e., $y_i=0$),
positivity requirements lead us to impose \eqref{flux-structure1} to any closure for diffusion fluxes that claim equivalence to Maxwell-Stefan diffusion.
For the Fick-Onsager closure, we hence require $L_{ik}=l_{ik} y_k$ for all $i\neq k$ with coefficients $l_{ik}$ which are regular functions of the state variables, i.e.\ without singularities as $y_j\to 0+$ for any $j=1,\ldots ,N$. Due to symmetry of ${\bf L}$,
together with ${\bf L} {\bf e}=0$, this leads to a structure of the type
\begin{equation}\label{structure-Fick}
L_{ik} = \rho_i \, ( a_i \, \delta_{ik} \, + \, y_k \, S_{ik} ),
\end{equation}
where $\bf S$ is symmetric and where we may assume $S_{ii}=0$ for all $i$. Above we used $\rho_i$ instead of $y_i$ as pre-factor since the $L_{ik}$
appear for mass diffusion fluxes rather than diffusion fluxes of mass fractions.

Let us note in passing that a similar and--concerning positivity requirements--consistent structure has been inferred from
the representation of transport coefficients by velocity correlations based on the Green-Kubo formalism in \cite{Bardow2011}, using ad hoc arguments. There, the structure similar to \eqref{structure-Fick}
was used as a starting point to obtain a multicomponent generalization of the Darken equation. We shall come back to this point in section~\ref{Darken}.

In case of non-degenerate diffusion,
we also assume
$a_i ({\bf y}) \to a_i ({\bf y}_0)>0$ as ${\bf y} \to {\bf y}_0$ with $y_{0,i}=0$.
Note that $a_i ({\bf y}) > 0$ always follows from the positivity properties of ${\bf L}$, since the diagonals
$L_{ii} = \langle \, {\bf L} {\bf e}^i, \, {\bf e}^i \, \rangle$ satisfy
\begin{align}\label{A-pos}
L_{ii} = \langle \, {\bf L} \, ({\bf e}^i - \frac{1}{N} {\bf e} ), \, ({\bf e}^i - \frac{1}{N} {\bf e} ) \, \rangle \geq \ell_0 \, (1-\frac{1}{N})^2
\end{align}
with $\ell_0 = \inf_{{\bf z} \in \{{\bf e}\}^{\perp}} \langle\, {\bf L} {\bf z}, \, {\bf z}\, \rangle/|{\bf z}|^2 >0$. However, the number $\ell_0$ might tend to zero for ${\bf y} \to {\bf y}_0$ with $y_{0,i}=0$.
\\[1ex]

\indent
Next, let check that the novel closure automatically leads to the desired structure. We hence let
${\bf L}:= {\bf P}^{\sf T} \, {\bf R}\, [{\mathcal D}+ {\bf K}\, {\bf X}]\, {\bf M}\, {\bf P}$, which yields
\begin{align*}
 {\bf R}^{-1} \, {\bf L} = {\bf P} \, \mathcal{D} \, {\bf M} \, {\bf P} + \frac{\rho}{c} \, {\bf P} \, {\bf K} \, {\bf P}^{\sf T} \, {\bf Y} \, .
\end{align*}

After a straightforward computation, using
\begin{equation}\label{PDMP}
{\bf P}\,\mathcal{D} \,{\bf M}\,{\bf P}\, = \,\mathcal{D} \,{\bf M}
+ \big( \langle \mathcal{D} \,{\bf M}\,{\bf y}, {\bf e}\rangle ({\bf e}\otimes {\bf e})
- (\mathcal{D} \,{\bf M}\,{\bf e}) \otimes {\bf e}- {\bf e}\otimes (\mathcal{D} \,{\bf M}\,{\bf e})\big) \,{\bf Y},
\end{equation}
we see that \eqref{structure-Fick} indeed holds with
\[
{\bf A} :={\rm diag}(a_i)= [{\bf P}\,\mathcal{D} \,{\bf M}\,{\bf P}]_{\rm diag} + \frac \rho c [{\bf P}\,{\bf K}\,{\bf P}^{\sf T}]_{\rm diag} {\bf Y}\vspace{-0.15in}
\]
and
\[
{\bf S} = [\langle {\bf D}\,{\bf M}\,{\bf y}, {\bf e}\rangle ({\bf e}\otimes {\bf e})
- ({\bf D}\,{\bf M}\,{\bf e}) \otimes {\bf e}- {\bf e}\otimes ({\bf D}\,{\bf M}\,{\bf e})]_{\rm off}
+ \frac \rho c [{\bf P}\,{\bf K}\,{\bf P}^{\sf T}]_{\rm off},
\]
where $[\, \cdot \,]_{\rm diag}$ and $[\, \cdot \,]_{\rm off}$ denote the diagonal and the off-diagonal part, respectively.
Evidently, the $L_{ik}$ as well as the $a_i=A_{ii}$ and $S_{ik}$ have the same regularity as the $d_i$ and $K_{ik}$
and no singularities are introduced. Finally, exploiting again \eqref{PDMP}, we see that ${\bf y}\to {\bf y}_0$
with $y_{0,i}=0$ implies
\[
A_{ii} ({\bf y}) \to M_i d_i({\bf y}_0)
\]
and $d_i({\bf y}_0)>0$ by the assumption of non-degenerate diffusion.\\

\indent
Concerning the  Maxwell-Stefan formulation, a rigorous proof, showing that the inverted Maxwell-Stefan equations lead to a matrix of phenomenological coefficients of the structure \eqref{structure-Fick}, will be postponed to the next section, where we show the equivalence of all three diffusion closures.
\section{Equivalence of the Different Diffusion Closures}

In order to prove the equivalence of the above closures for multicomponent diffusion, let us first summarize the different forms together with the assumptions on the phenomenological coefficients.\\[1ex]
\noindent
{\bf Form (A): Fick-Onsager diffusion fluxes.} According to \eqref{Closure-Fick}, these are given as
\[
{\bf J}^{\rm FO} \, =\, - {\bf L} \, \nabla \frac{\boldsymbol\mu}{RT},
\]
where ${\bf L} \in \R^{N\times N}$ is symmetric, positive definite on $\{ {\bf e} \}^\perp$
and ${\bf L} {\bf e}={\bf 0}$. It further possesses the structure
\begin{equation}\label{L-Fick}
{\bf L} \, = \, {\bf R}\, [ \, {\bf A}\, +\, {\bf S}\, {\bf Y} \, ],
\end{equation}
with ${\bf A}={\rm diag}(a_1,\ldots ,a_N)$ and off-diagonal ${\bf S}={\bf S}^{\sf T}$, where the coefficients $a_i$ and $S_{ij}$ (for $i<j$) are regular functions of the state variables $(T,\rho, \, y_1,\ldots ,y_N)$. We require $a_i ({\bf y}) \to a_i ({\bf y}_0)>0$ as ${\bf y} \to {\bf y}_0$ with $y_{0,i}=0$. Note that the weaker inequality $a_i({\bf y}_0) \geq 0$ would already follow from \eqref{L-Fick} and the positivity properties of ${\bf L}$ (cf.\ \eqref{A-pos}) \\[1ex]
\noindent
{\bf Form (B): Maxwell-Stefan diffusion fluxes.} According to \eqref{Closure-MSJbbis}, these are given as
\[
{\bf J}^{\rm MS}= - \,{\bf B}^{\sharp}({\bf y}) \mathbf{R} \, \nabla \frac{\boldsymbol{\mu}}{RT},
\]
where ${\bf B}=B_{ij}\in \rr^{N\times N}$ with
\[
B_{ij}=- y_i f_{ij} \;\mbox{ for } i\neq j, \quad B_{ii}= \sum_{k\neq i} y_k f_{ik},
\]
where $f_{ik}=f_{ki}$ ($i\neq k$) are regular functions of $(T, \, \rho, \, {\bf y})$ such that ${\bf B} {\bf Y}$ is positive definite on $\{ {\bf e} \}^\perp$. \\[1ex]
%
\noindent
{\bf Form (C): Novel form of diffusion fluxes.} According to \eqref{closure-B3}, these are given as
\[
 {\bf J}  \, =\, -{\bf P}^{\sf T} \, {\bf R}\, [{\mathcal D}+ {\bf K}\, {\bf X}]\, {\bf M}\, {\bf P} \, \nabla \frac{\boldsymbol{\mu}}{RT},
\]
where ${\mathcal D}={\rm diag}(d_1, \ldots ,d_N)$, ${\bf K}$ is off-diagonal with ${\bf K}={\bf K}^{\sf T}$ and ${\bf K} {\bf e}= {\bf 0}$. The coefficients $d_i$ and $K_{ij}$ (for $i<j$) are regular functions of the state variables $(T,\rho, y_1,\ldots ,y_N)$.
Moreover, ${\mathcal D}\,+\,{\bf X}^{1/2} \, {\bf K} \, {\bf X}^{1/2}$ is positive definite on $\{\sqrt{\bf x}\}^\perp$, and $d_i ({\bf y}) \to d_i ({\bf y}_0)>0$ as ${\bf y} \to {\bf y}_0$ with $y_{0,i}=0$. \\[1.5ex]
\noindent
The following results establish the equivalence of the three different closures. We split the proof into two steps. In the next Theorem \ref{Th-1} we discuss the structural aspects, while the second statement \ref{Th-2} focuses on the regularity of the coefficients and estimates for the eigenvalues of the involved matrices. \\[0.5ex]

\begin{theorem}\label{Th-1}
The forms {\rm (A)}, {\rm (B)} and {\rm (C)} of the diffusion fluxes (cf.\ also \eqref{Closure-Fick}, \eqref{Closure-MSJbbis} and \eqref{closure-B2}) are equivalent.
%
\end{theorem}
\begin{proof}
(A) $\Rightarrow$ (C). Let the closure for ${\bf J}^{\rm FO}$ be given with coefficient matrix ${\bf L}$ having all properties as stated above, in particular being of type \eqref{L-Fick}.
Note that ${\bf L} \, {\bf e}={ \bf 0}$ together with symmetry of ${\bf L}$ implies
\begin{equation}\label{PTLP}
{\bf P}^{\sf T} \,{\bf L}\, {\bf P} = {\bf P}^{\sf T} \, [ {\bf L} - ({\bf L} {\bf e})\otimes {\bf y}]=
[{\bf I}-{\bf y}\otimes {\bf e}] \,{\bf L} = {\bf L} \, .
\end{equation}
%
Starting with \eqref{L-Fick}, we re-express
\begin{align*}
{\bf L} = & {\bf P}^{\sf T} \,{\bf R}\, [  {\bf A} + {\bf S}\, {\bf Y}\,] \,   \, {\bf P} = \rho \, {\bf P}^{\sf T} \, {\bf Y} \, [\,{\bf Y}^{-1} \, {\bf A} + {\bf S}\, ] \,  {\bf Y} \, {\bf P}\\
 = &  \rho \, {\bf P}^{\sf T} \, {\bf Y} \, [\,{\bf Y}^{-1} \, {\bf A} + {\bf S}\, +   {\bf b} \otimes {\bf e} + {\bf e} \otimes {\bf b}] \,  {\bf Y} \, {\bf P} \, ,
\end{align*}
where ${\bf b}$ is arbitrary. The vector ${\bf b}$ is next computed as to ensure, for the matrix ${\bf T} := {\bf Y}^{-1} \, {\bf A} + {\bf S}\, +   {\bf b} \otimes {\bf e} + {\bf e} \otimes {\bf b}$, that ${\bf T}_{\text{off}} \, {\bf e} = 0$. The latter means nothing else but $([N-2] \, {\bf I} + {\bf e} \otimes {\bf e}) \, {\bf b}  = - {\bf S} \, {\bf e}$ (compare with \eqref{a-correct}), so that
\begin{align}\label{b-correct}
{\bf b} = - \frac{1}{N-2} \, ({\bf I} -\frac{1}{2 \, (N-1)} \, {\bf e} \otimes {\bf e}) \, {\bf S} \, {\bf e} \, .
\end{align}
We observe that the entries of ${\bf b}$ are regular functions of the state variables since the entries of ${\bf S}$ are assumed regular.

To compute the diffusion matrices $\mathcal{D}$ and ${\bf K}$, we decompose ${\bf T} = \text{diag}({\bf T}) + {\bf T}_{\text{off}}$ in which
\begin{align*}
 \text{diag}({\bf T})  =  {\bf Y}^{-1} \, {\bf A} + 2 \, \text{diag} ({\bf b}) \, , \qquad {\bf T}_{\text{off}} =  {\bf S} + {\bf b} \otimes {\bf e} +  {\bf e} \otimes {\bf b} -  2 \, \text{diag}({\bf b})\, .
\end{align*}
We obtain the representations
\begin{align*}
{\bf L} = & {\bf P}^{\sf T} \, {\bf R} \, {\bf T} \,  {\bf Y} \, {\bf P}
= {\bf P}^{\sf T} \, {\bf R} \, [{\bf A} + 2 \, \text{diag}({\bf b}) \,{\bf Y}  + {\bf T}_{\text{off}} \, {\bf Y}  ]\,   {\bf P} \, .
\end{align*}
We define
\begin{align}\label{DfromAandS}
 \mathcal{D} = & {\bf A} \, {\bf M}^{-1} + 2 \,\frac{c}{\rho} \, \text{diag}({\bf b}) \,{\bf X}\, ,\\
 \label{KfromAandS} {\bf K} = & \frac{c}{\rho} \, ({\bf S} + {\bf b} \otimes {\bf e} +  {\bf e} \otimes {\bf b} -  2 \, \text{diag}({\bf b})) \, ,
\end{align}
verifying easily that ${\bf K}$ is off-diagonal, symmetric and satisfies ${\bf K} \, {\bf e} = {\bf 0}$. Moreover, we see that ${\bf L} = {\bf P}^{\sf T} \, {\bf R} \, ( \mathcal{D} +  {\bf K}  \, {\bf X} ) \, {\bf M} \, {\bf P}$ and that the entries of $\mathcal{D}$ and ${\bf K}$ are regular functions of the state variables.
The $i^{\rm th}$ diagonal entry is $d_i = a_i/M_i + 2 \, \frac{c}{\rho} \, b_i \, x_i$. Invoking that the components of ${\bf b}$ are regular functions, i. e. non-degenerate as $y_i \rightarrow 0+$, the components of $d_i$ converges to $a_i({\bf y}_0)/M_i >0$ as $y_i\to 0+$.\\[1ex]
%
%
%
%
%
%
%
%
%
\noindent
(C) $\Rightarrow$ (B).
We rewrite the fluxes from \eqref{closure-B2} as
\begin{equation}\label{fluxes-rewritten}
{\bf J}  \, =\, -{\bf P}^{\sf T} \, {\bf M}\, {\bf X}^{1/2} \, {\bf D}\, {\bf X}^{-1/2} \, {\bf R}\, {\bf P} \, \nabla \frac{\boldsymbol{\mu}}{RT},
\end{equation}
where ${\bf D} = {\mathcal D} + {\bf X}^{1/2}\, {\bf K}\, {\bf X}^{1/2}$, with ${\mathcal D}={\rm diag}(d_1, \ldots ,d_N)$,
${\bf K}={\bf K}^{\sf T}$ and ${K}_{ii}=0$ for all $i$, is positive definite on $\{ \sqrt{\bf x} \}^\perp$.
%


It is possible to replace ${\bf D}$ in \eqref{fluxes-rewritten} with ${\bf D}_0 = {\bf D} + {\bf b} \otimes \sqrt{\bf x}  + \sqrt{\bf x} \otimes {\bf b}$ for arbitrary vector ${\bf b}$. We choose ${\bf b}$ in such a way as to ensure that the modified diffusion matrix ${\bf D}_0$ satisfies ${\bf D}_0 \, \sqrt{\bf x} = 0$. This is true for
\begin{equation*}
{\bf b} := - {\bf D} \, \sqrt{\bf x} + \frac{\langle \, {\bf D} \, \sqrt{\bf x}, \, \sqrt{\bf x} \, \rangle}{2} \, \sqrt{\bf x} \, .
\end{equation*}
Since ${\bf D}_0$ inherits positivity on $\{ \sqrt{\bf x} \}^\perp$ and ${\bf D}_0 \, \sqrt{\bf x} = 0$, this matrix is symmetric and positive semi-definite on $\mathbb{R}^N$, with rank $N-1$. Hence, the remark \ref{GROUP} guarantees that the group inverse ${\bf D}_{0}^{\sharp}$ is well defined, symmetric, positive definite on $\{ \sqrt{\bf x} \}^\perp$ and ${\rm ker}({\bf D}_{0}^{\sharp}) = \{\sqrt{\bf x}\}$. We define
\begin{align}\label{Bcandidate}
 {\bf B} := {\bf X}^{1/2} \, {\bf D}_{0}^{\sharp}\, {\bf X}^{-1/2} \, {\bf M}^{-1} \,  ,
\end{align}
and we apply this ${\bf B}$ to \eqref{fluxes-rewritten} with ${\bf D}$ replaced by ${\bf D}_0$, resulting into
\begin{align}\label{fluxes-rewritten-2}
{\bf B} \, {\bf J}  \, =\, -{\bf X}^{1/2} \, {\bf D}_{0}^{\sharp}\, {\bf X}^{-1/2} \, {\bf M}^{-1} \, {\bf P}^{\sf T} \, {\bf M}\, {\bf X}^{1/2} \, {\bf D}_0\, {\bf X}^{-1/2} \, {\bf R}\, {\bf P} \, \nabla \frac{\boldsymbol{\mu}}{RT},
\end{align}
Observe that
\begin{align}
{\bf X}^{-1/2} \, {\bf M}^{-1} \, {\bf P}^{\sf T} \, {\bf M}\, {\bf X}^{1/2} = {\bf I} - \frac{c}{\rho} \, {\bf \sqrt{x}} \otimes {\bf M}  {\bf \sqrt{x}}  \, .
\end{align}
Since ${\bf D}_0^{\sharp} {\bf \sqrt{x}} =0$, the product of ${\bf D}_0^{\sharp}$ with the latter matrix is again ${\bf D}_0^{\sharp}$, and \eqref{fluxes-rewritten-2} gives
\begin{align*}
 {\bf B} \, {\bf J}  \, = & \, -{\bf X}^{1/2} \, {\bf D}_{0}^{\sharp}\,  {\bf D}_0\, {\bf X}^{-1/2} \, {\bf R}\, {\bf P} \, \nabla \frac{\boldsymbol{\mu}}{RT} \\
 = & -{\bf X}^{1/2} \, ({\bf I} - {\bf \sqrt{x}} \otimes  {\bf \sqrt{x}})\, {\bf X}^{-1/2} \, {\bf R}\, {\bf P} \, \nabla \frac{\boldsymbol{\mu}}{RT} = - {\bf R}\, {\bf P} \, \nabla \frac{\boldsymbol{\mu}}{RT} \, .
\end{align*}
We next check that $\boldsymbol \tau :=\rho {\bf B} {\bf Y}$ possesses the desired properties. In view of \eqref{Bcandidate}
\begin{equation}\label{tau-rep}
\boldsymbol \tau =
c \,{\bf X}^{1/2} \, {\bf D}_0^{\sharp} \, {\bf X}^{1/2} \, .
\end{equation}
The symmetry of $\boldsymbol \tau$ follows from the symmetry of ${\bf D}_0$ and Remark \ref{GROUP}. With this representation it is also clear that $\boldsymbol \tau {\bf e}=0$.
To show that $\boldsymbol \tau$ is positive definite on $\{\bf e \}^\perp$, it suffices to show that ${\bf D}_0^{\sharp}$ is positive definite on $\{{\bf \sqrt{x}}\}^{\perp}$, which again follows from Remark \ref{GROUP}. The regularity of the entries of the matrix \eqref{Bcandidate} is discussed in Theorem \ref{Th-2}.\\[1ex]
\noindent
(B) $\Rightarrow$ (A). Let the diffusion fluxes ${\bf J}^{\rm MS}$ be given by \eqref{Closure-MSJbbis}. Then, evidently,
\[
{\bf J}^{\rm MS}=- {\bf L}\, \nabla \frac{\boldsymbol{\mu}}{RT}
\]
with ${\bf L}:= {\bf B}^{\sharp}({\bf y}) \, {\bf R}$. Recalling \eqref{Bandtau2}
\begin{align*}
{\bf P}^{\sf T} \, {\bf Y} \, ({\bf B}({\bf y}) \, {\bf Y})^{\sharp} \, {\bf R} \, {\bf P} = {\bf B}^{\sharp}({\bf y}) \, {\bf R} = {\bf L} \, ,
\end{align*}
proving that ${\bf L}$ is symmetric an positive definite on $\{{\bf e}\}^{\perp}$.
%
%
%

It remains to show that ${\bf L}$ has the structure as given in \eqref{L-Fick}.

We show here only how to choose the matrices ${\bf A}$ and ${\bf S}$. We will prove with the next theorem \ref{Th-2} under which conditions the entries of these matrices are regular functions of the state variables. By assumption, ${\bf B}({\bf y}) \, {\bf Y}$ is symmetric and positive semi-definite on $\mathbb{R}^N$, and even positive definite on $\{ {\bf e} \}^{\perp}$. We denote $c_0 = c_0({\bf y}) > 0$ the smallest singular value of ${\bf B}({\bf y}) \, {\bf Y}$ on $\{ {\bf e} \}^{\perp}$, and we have
$$\langle {\bf B}({\bf y}) \, {\bf Y} \, {\bf e}^i, \, {\bf e}^i \, \rangle \geq c_0 \, |{\bf e}^i - \frac{1}{N} \, {\bf e}|^2 = c_0 \, (1- \frac{1}{N})^2 \, .$$
Thus, the diagonal entries of ${\bf B}({\bf y})$ satisfy $b_{ii} \geq c_0 \, y_i^{-1} \, (1- 1/N)^2 > 0$.

Suppose that the matrix ${\bf B}({\bf y})$ is defined via (B). Then $b_{ii} = \sum_{k \neq i}^N f_{ik} \, y_k$. The assumption (B) (cf.\ also the mass-based Maxwell-Stefan system \eqref{MS-mass-based}) means nothing else but
\begin{align*}
 b_{ii} \, {\bf j}_i - y_i \, \sum_{k\neq i} f_{ik} \, {\bf j}_k = - {\bf d}_i \, .
\end{align*}
Thus, choosing arbitrary $f_{ii} > 0$ ($i=1,\ldots,N$), we obtain that $(b_{ii} + f_{ii} \, y_i)\, {\bf j}_i - y_i \, \sum_{k = 1}^N f_{ik} \, {\bf j}_k = - {\bf d}_i$, and
\begin{align*}
{\bf j}_i = - \frac{1}{b_{ii} + f_{ii} \, y_i}\, {\bf d}_i + \frac{y_i}{b_{ii} + f_{ii} \, y_i} \, \sum_{k = 1}^N f_{ik} \, {\bf j}_k \, .
\end{align*}
We recall at this stage the equivalent form \eqref{Closure-MSJbbis} of the fluxes. With ${\bf F}  := \{f_{ik}\}_{i,k=1,\ldots,N}$ denoting the extended matrix of friction coefficients, it follows that
\begin{align}\label{AandSfromB}
  \mathbf{J} = & - (\text{diag}(\mathbf{F} \, {\bf y}))^{-1} \, {\bf d} - (\text{diag}(\mathbf{F} \, {\bf y}))^{-1} \,  \mathbf{Y} \, {\bf F}  \, \mathbf{B}^{\sharp} \,  \mathbf{R}  \, \nabla \frac{ \boldsymbol{\mu}}{RT} \nonumber\\
 = & -(\text{diag}(\mathbf{F} \, {\bf y}))^{-1}  \, \{  \mathbf{R}  -  \boldsymbol{\rho} \otimes  \mathbf{y} +  \mathbf{Y} \,  \mathbf{F} \,  \mathbf{B}^{\sharp} \,  \mathbf{R} \}\nabla \frac{ \boldsymbol{\mu}}{RT} \, .
\end{align}
We define $ \tilde{\mathbf{A }} = (\text{diag}(\mathbf{F} \, {\bf y}))^{-1}$ and $ \tilde{\mathbf{S}} = (\text{diag}(\mathbf{F} \, {\bf y}))^{-1}\, [ -\mathbf{e} \otimes  \mathbf{e} +  \mathbf{F} \, \mathbf{B}^{\sharp}]$ and obtain that $ \mathbf{J} = -  \mathbf{R} \, ( \tilde{\mathbf{A}}+ \tilde{\mathbf{S}} \,  \mathbf{Y}) \,   \nabla \frac{ \boldsymbol{\mu}}{RT}$. Hence, it also follows that ${\bf L} \, \nabla \frac{ \boldsymbol{\mu}}{RT} = \mathbf{R} \, ( \tilde{\mathbf{A}}+ \tilde{\mathbf{S}} \,  \mathbf{Y}) \, \nabla \frac{ \boldsymbol{\mu}}{RT}$. Note that $\nabla \frac{ \boldsymbol{\mu}}{RT}$ in fact stands for any element of $\mathbb{R}^{N\times 3}$ in the algebraic inversion, so we must have ${\bf L}  = \mathbf{R} \, ( \tilde{\mathbf{A}}+ \tilde{\mathbf{S}} \,  \mathbf{Y})$.

The entries of $\tilde{\mathbf{A }}$ are $\tilde{a}_i = 1/(b_{ii} + y_i \, f_{ii}) > 0$. It remains to shift the diagonal of $\tilde{\mathbf{S}}$. If the entries of $\tilde{{\bf S}}$ are regular functions of the state variables - which we show directly here below - this procedure does not affect the asymptotic properties for $y_i \rightarrow 0+$.

This finishes the proof of theorem \ref{Th-1}.
\end{proof}

\noindent
So far we have treated the question of structural equivalence. Also important is the question, that we want to discuss next, whether the entries of the involved matrices remain regular functions of the state variables when switching from one form of the fluxes to another, equivalent one.

To fix the ideas, let us call a function $f$ of the variables $(T, \, \rho, \, y_1, \ldots,y_N)$ \emph{regular} if it is of class $C^k(]0, \, + \infty[ \times ]0, \, + \infty[ \times \overline{S}^1_+)$ for a fixed $k \in \mathbb{N}_0$. Here $\overline{S}^1_+$ is the closed positive unit sphere of the one-norm, i.e.\
$$\{{\bf y} \in \mathbb{R}^N \, : \, 0 \leq y_i \leq 1, \, \quad \sum_{i=1}^N y_i = 1\} \, .$$

In the next theorem, we prove a quantitative form of the equivalence, where we also provide some composition--independent bounds for the coefficients and eigenvalues of the diffusion matrices. This extends the results in the section 7.7 of \cite{Giovan}, where only the smoothness in the interior of the domain $\rho, \, T, \, {\bf y} > 0$ is verified. Our main assumption is that diffusion is non-degenerate in the composition variable, meaning that diffusion matrices satisfy lower and upper bounds independent of the composition vector.
\begin{theorem}\label{Th-2}
 The following statements are equivalent:
 \begin{enumerate}[(i)]
\item   \label{fickonsager} Form {\rm (A)} of the fluxes, and there exists a strictly positive, regular function $d_0$ depending only on $(T, \, \rho)$ such that $ {\bf L} \geq d_0 \, {\bf P}^{\sf T} \,{\bf  M} \,{\bf  R} \, {\bf P}$;
 \item  \label{maxstef} Form {\rm (B)} of the fluxes, and ${\bf B} \, {\bf Y} \geq d_0 \,  {\bf P}^{\sf T} \, {\bf M}^{-1} \, {\bf  Y} \, {\bf P}$ with a strictly positive, regular function $d_0$ of the variables $(T, \, \rho)$;
 \item  \label{novel} Form {\rm (C)} of the fluxes, and there is a strictly positive, regular function $d_0 = d_0(T, \, \rho)$ such that $$\inf_{{\bf b} \in \{\sqrt{ {\bf x} }\}^{\perp}} \langle (\mathcal{D} + {\bf X}^{\frac{1}{2}} \, {\bf K} \, {\bf X}^{\frac{1}{2}}) \, \frac{{\bf b}}{|{\bf b}|}, \,  \frac{{\bf b}}{|{\bf b}|} \rangle \geq d_0 \, .$$
 \end{enumerate}
\end{theorem}
\begin{proof}
We prove first that \eqref{fickonsager} and \eqref{novel} are equivalent.

Assume \eqref{fickonsager}. Just as in the proof of Theorem \ref{Th-1} (see \eqref{DfromAandS}, \eqref{KfromAandS}), we obtain from the given ${ \bf A}$ and ${\bf S}$ with regular entries matrices $\mathcal{D}$ and ${\bf K}$ with regular entries. We might represent arbitrary ${\bf b}\in \{\sqrt{\bf x} \}^\perp$, via ${\bf b}={\bf X}^{1/2} {\bf M}{\bf P} {\bf z}$ with $${\bf z } = ({\bf I} - \frac{1}{N} \, {\bf e} \otimes {\bf e}) \, {\bf M}^{-1} \, {\bf X}^{-\frac{1}{2}} \, {\bf b} \in  \{\bf e \}^\perp \, .$$ We then have
\begin{align*}
& \langle [{\mathcal D}+ {\bf X}^{1/2}\, {\bf K}\, {\bf X}^{1/2}]\, {\bf b}, {\bf b} \rangle =  \frac 1 c \langle {\bf L}\, {\bf P}\,{\bf z}, {\bf P}\,{\bf z} \rangle
\geq \frac{d_0}{c} \, \langle {\bf P^T} \, {\bf  R} \,{\bf M}\, {\bf P}\,{\bf P}\,{\bf z}, \,{\bf P}\,{\bf z} \rangle \\
& \qquad =  \frac{d_0}{c} \, \langle  \, {\bf  R} \,{\bf M}\, {\bf P}\,{\bf z}, \,{\bf P}\,{\bf z} \rangle  = \frac{d_0}{c} \, \langle {\bf X}^{-\frac{1}{2}}\, {\bf M^{-1}}  \, {\bf R}  \, {\bf X}^{-\frac{1}{2}}  \,{\bf b}, \,{\bf b} \rangle =  d_0 \, |{\bf b}|^2 \, .
\end{align*}
This proves \eqref{novel}.

Assume now that \eqref{novel} is valid. Then, we obtain the Fick--Onsager form (A) with ${\bf L} = c \, {\bf P}^{\sf T} \, {\bf M}\, {\bf X}^{1/2} \, {\bf D}\, {\bf X}^{1/2} \, {\bf M}\, {\bf P}$. As shown in section \ref{positivity}, relation \eqref{PDMP}, the corresponding matrices ${ \bf A}$ and ${\bf S}$ are regular functions of the state variables. We next easily verify that
\begin{align*}
 {\bf L} \geq & c \, d_0 \, {\bf P}^{\sf T} \, {\bf M}\, {\bf X}^{1/2} \, {\bf X}^{1/2} \, {\bf M}\, {\bf P} =   d_0 \, {\bf P}^{\sf T} \, {\bf R} \, {\bf M}\, {\bf P} \, .
\end{align*}
This achieves to prove that \eqref{fickonsager} and \eqref{novel} are equivalent, and with the same $d_0$.\\[1ex]

We next prove that \eqref{maxstef} implies \eqref{fickonsager}. The proof consists in establishing two points: 1. Finding a lower bound for the ellipticity constant of ${\bf L}$ over $\{ {\bf e} \}^{\perp}$ and, 2. Proving that the entries of ${ \bf A}$ and ${\bf S}$ in the representation ${\bf L} = { \bf R} \, ({ \bf A} + {\bf S} \, {\bf Y})$ are regular functions.

Ad 1. By the assumptions in Form {\rm (B)}, the matrix $ \boldsymbol{\tau} := {\bf B}({\bf y}) \,{\bf R}$ is symmetric and positive semi-definite. Thus, ${ \bf R}^{-1} \, {\bf B}({\bf y}) = {\bf R}^{-1} \, \boldsymbol{\tau} \, {\bf R}^{-1}$ is likewise symmetric and positive semi-definite. The matrix ${\bf B}_{\alpha} = {\bf B}({\bf y}) + \alpha \, {\bf y} \otimes {\bf e}$ is invertible for all $\alpha >0$ and ${\bf R}^{-1} \, {\bf B}_{\alpha} = {\bf R}^{-1} \, {\bf B}({\bf y}) + \frac{\alpha}{\rho} \, {\bf e}\otimes {\bf e}$ is symmetric and positive definite.

We consider an auxiliary matrix ${\bf G} := {\bf R}^{\frac{1}{2}} \, {\bf R}^{-1} \, {\bf B}_{\alpha} \, {\bf R}^{\frac{1}{2}}$ which is again positive definite, symmetric, and possesses the entries
\begin{align*}
 g_{ij} = \begin{cases}
            (-f_{ij}+\alpha) \, \sqrt{y_i \, y_j} & \text{ for } i \neq j\\
            y_i \, (\frac{1}{y_i} \, \sum_{k\neq i} f_{ik} \, y_k + \alpha) & \text{ for } i = j
           \end{cases} \, .
\end{align*}
The latter shows that $\|{\bf G}\|_{\infty} \leq \sup_{i\neq j} |f_{ij}| + \alpha$. Thus the spectral radius of ${\bf G}$ is bounded by the same quantity, and
\begin{align*}
 \lambda_{\min}({\bf G}^{-1}) \geq \frac{1}{ \sup_{i\neq j} |f_{ij}| + \alpha  } =: \lambda(\alpha) \, .
\end{align*}
In other words, the matrix ${\bf G}^{-1} - \lambda(\alpha) \, {\bf I}$ is positive semi-definite, which means nothing else but ${\bf R}^{-\frac{1}{2}} \, {\bf B}_{\alpha}^{-1} \, {\bf R} \, {\bf R}^{-\frac{1}{2}} - \lambda(\alpha) \, {\bf I} \geq 0$.
We multiply from the left with the matrix ${\bf P}^{\sf T} \, {\bf R}^{\frac{1}{2}}$ and from the right with ${\bf R}^{\frac{1}{2}} \, {\bf P}$, preserving the inequality, and thus
\begin{align*}
{\bf P}^{\sf T} \, {\bf B}_{\alpha}^{-1} \, {\bf R} \, {\bf P} \geq & \lambda(\alpha) \, {\bf P}^{\sf T} \, {\bf R} \, {\bf P} \geq \frac{\lambda(\alpha)}{\|{\bf M}\|_{\infty}} \, {\bf P}^{\sf T} \, {\bf R} \, {\bf M} \, {\bf P} \, .
\end{align*}
The appendix, section \ref{Drazin}, establishes that ${\bf B}_{\alpha}^{-1} = {\bf B}(\mathbf{y})^{\sharp} + {\bf y} \otimes {\bf e}/\alpha$. Since ${\bf P}^{\sf T} \, {\bf y} = {\bf 0}$, we see that ${\bf P}^{\sf T} \, {\bf B}({\bf y})^{\sharp} \, {\bf R} \, {\bf P} \geq (\lambda(\alpha)/\|{\bf M}\|_{\infty}) \, {\bf P}^{\sf T} \, {\bf R} \, {\bf M} \, {\bf P}$. Finally, ${\bf B}({\bf y})^{\sharp} \, {\bf R} \, {\bf P} = {\bf B}({\bf y})^{\sharp} \, {\bf R}$ and ${\bf P}^{\sf T} \, {\bf B}({\bf y})^{\sharp} = {\bf B}({\bf y})^{\sharp}$ imply that
\begin{align}\label{inversefrombelow}
{\bf B}({\bf y})^{\sharp} \,{\bf R} \geq \frac{\lambda(\alpha)}{\|{\bf M}\|_{\infty}} \, {\bf P}^{\sf T} \, {\bf R} \, {\bf M} \, {\bf P} \, .
\end{align}
Thus, the Fick-Onsager form holds with ${\bf L} = {\bf B}({\bf y})^{\sharp} \,{\bf R}$ and, letting $\alpha \rightarrow 0$ in \eqref{inversefrombelow}, we prove ${\bf L} \geq d_0 \, {\bf P}^{\sf T} \, {\bf R} \, {\bf M} \, {\bf P}$ with $d_0 := (\sup_{i\neq j} |f_{ij}| \, \|{\bf M}\|_{\infty})^{-1}$. If the coefficients $f_{ij}$ are regular, they are bounded above independently the composition-variable ${\bf y}$, hence this $d_0$ is bounded from below by a function of $(T,\rho)$. This completes the proof of 1.

Ad 2. In the proof of Theorem \ref{Th-1}, relation \eqref{AandSfromB}, we showed that ${\bf L} = { \bf R} \, (\tilde{ \bf A} + \tilde{\bf S} \, {\bf Y})$ with $\tilde{\mathbf{A }}= (\text{diag}(\mathbf{F} {\bf y}))^{-1}$ and $ \tilde{\mathbf{S}} = (\text{diag}(\mathbf{F} {\bf y}))^{-1} \, [- \mathbf{e} \otimes  \mathbf{e} +  \mathbf{F} \, \mathbf{B}^{\sharp}]$. Here $\mathbf{F}$ denotes the matrix of the friction coefficients $f_{ik}$ arbitrarily extended in the diagonal with positive $f_{ii}$. In particular $(\text{diag}(\mathbf{F} {\bf y}))^{-1}$ is diagonal with entries $1/(b_{ii} + f_{ii} \, y_i)$.

The final form of the matrices ${ \bf A}$ and ${\bf S}$ is obtained by shuffling the diagonal of $\tilde{\bf S}$ into $\tilde{{\bf A}}$. Thus, showing that the entries of $\tilde{{\bf A }}$ and $\tilde{\mathbf{S}}$ are regular functions is already sufficient for the regularity of ${ \bf A}$ and ${\bf S}$.

Due to the assumption \eqref{maxstef}, the matrix ${\bf B}\, {\bf Y}$ is symmetric and positive semi-definite on $\mathbb{R}^N$ and satisfies $${\bf B} \, {\bf Y} \geq d_0 \,  {\bf P}^{\sf T} \, {\bf M}^{-1} \, {\bf  Y} \, {\bf P} \geq \frac{d_0}{\|{\bf M}\|_{\infty}} \,  {\bf P}^{\sf T} \, {\bf  Y} \, {\bf P} \, .$$ With $\tilde{d}_0 :=  d_0/\|{\bf M}\|_{\infty}$, it therefore follows that
\begin{align*}
 \langle {\bf B} \, {\bf Y} \, {\bf e}^i, \, {\bf e}^i \, \rangle \geq \tilde{d}_0 \, \sum_{k=1}^N y_k \, (\delta_{ki} - y_i)^2 = \tilde{d}_0 \, y_i \, (1 - y_i) \, .
\end{align*}
This shows that $b_{ii} \geq \tilde{d}_0  \, (1-y_i)$. Choosing $f_{ii} := \tilde{d}_0$, we obtain that $b_{ii} + f_{ii} \, y_i \geq \tilde{d}_0$, and it follows that
\begin{align*}
 \|(\text{diag}(\mathbf{F} {\bf y}))^{-1}\|_{\infty} = \frac{1}{\min_i b_{ii} + f_{ii} \, y_i}  \leq \tilde{d}_0^{-1} = \frac{\|{\bf M}\|_{\infty}}{d_0} \, .
\end{align*}
Thus, the matrix $\tilde{ {\bf A}}$ of \eqref{AandSfromB} consists of regular functions of the state variables. In order to show the same for $\tilde{ {\bf S}}$, it is now sufficient to prove that the generalized inverse $\mathbf{B}^{\sharp}$ consists of regular functions. Here we invoke the appendix, section \ref{Drazin}, and we see, in turn, that it is sufficient to find a positive lower bound from below for the quantity
\begin{align}\label{defdet}
D_0 := \sum_{i=1}^N \text{det} (\mathbf{B}[i|i]) \, ,
\end{align}
where $\mathbf{B}[i|i] \in \mathbb{R}^{(N-1) \times (N-1)}$ is the matrix obtained by canceling the row and column with index $i$. In order to derive the lower bound, note first that $(\mathbf{B} \, {\bf R} )[i|i] = \mathbf{B}[i|i]  \, {\bf R}[i|i]$ due to the fact that ${\bf R}$ is diagonal. This means that we first have
\begin{align}\label{premineq}
  \text{det} (\mathbf{B}[i|i]) = \text{det} \big( (\mathbf{B} \, {\bf R})[i|i]\big) \, \, \, \Pi_{j\neq i} \frac{1}{\rho_j} \, .
\end{align}
Next we exploit that $(\mathbf{B} \, {\bf R} )[i|i] = {\bf Q}^{\sf T} \, \mathbf{B} \, {\bf R} \, {\bf Q}$ with the rectangular projector ${\bf Q} \in \mathbb{R}^{N \times (N-1)}$ obtained by canceling from the identity matrix in $\mathbb{R}^N$ the $i^{\rm th}$ column and replacing the $i^{\rm th}$ line by the zero vector. In view of the assumptions in \eqref{maxstef}, we then have
\begin{align*}
 (\mathbf{B} \, {\bf R} )[i|i] \geq & \tilde{d}_0 \, {\bf Q}^{\sf T} \, {\bf P}^{\sf T} \, {\bf R} \,  {\bf P}\,  {\bf Q}\\
  = & \tilde{d}_0 \, \rho \, {\bf Q}^{\sf T} \, {\bf P}^{\sf T} \, {\bf Y} \,  {\bf P}\,  {\bf Q} = \tilde{d}_0 \, \rho \, {\bf Q}^{\sf T} \, ({\bf Y} - {\bf y} \otimes {\bf y})\,  {\bf Q} \\
  = & \tilde{d}_0 \, \rho \, ({\bf Y} - {\bf y} \otimes {\bf y})[i|i] \, .
\end{align*}
Making use of the Lemma \ref{Hadamard} just hereafter, we find $$\text{det} \big( (\mathbf{B} \, {\bf R} )[i|i] \big) \geq (\tilde{d}_0 \, \rho)^{N-1} \, \text{det} \big({\bf Y} - {\bf y} \otimes {\bf y})[i|i]\big) \, .$$
We express $({\bf Y} - {\bf y} \otimes {\bf y})[i|i] = {\bf Y}^{\frac{1}{2}}[i|i] \, ({\bf I} - \sqrt{{\bf y}} \otimes \sqrt{{\bf y}})[i|i] \, {\bf Y}^{\frac{1}{2}}[i|i]$. The eigenvalues of $({\bf I} - \sqrt{{\bf y}} \otimes \sqrt{{\bf y}})[i|i]$ are $\lambda =1$ with multiplicity $N-2$ and $\lambda = y_i$. It follows that
\begin{align*}
 \text{det} \big(({\bf Y} - {\bf y} \otimes {\bf y})[i|i]\big) = \text{det}\big({\bf Y}^{\frac{1}{2}}[i|i]\big) \, y_i \, \text{det}\big({\bf Y}^{\frac{1}{2}}[i|i]\big) = y_i \, \Pi_{j \neq i} y_j \, .
\end{align*}
We have therefore shown that
\begin{align*}
 \text{det}  (\mathbf{B} \, {\bf R} )[i|i] \geq (\tilde{d}_0 \, \rho)^{N-1} \, \Pi_{j =1}^N y_j = \frac{\tilde{d}_0^{N-1}}{\rho}  \, \Pi_{j =1}^N \rho_j \, .
\end{align*}
Thus, \eqref{premineq} yields $  \text{det} (\mathbf{B}[i|i]) \geq  \tilde{d}_0^{N-1} \, y_i$ while \eqref{defdet} implies that $D_0 \geq \tilde{d}_0^{N-1} = (d_0/\|{\bf M}\|_{\infty})^{N-1}$. This completes the proof of 2, and altogether, \eqref{maxstef} implies \eqref{fickonsager}.\\[1ex]

Finally, we  to prove that \eqref{fickonsager} implies \eqref{maxstef}. We thus assume that the diffusion fluxes are given as ${\bf J} = - {\bf L}  \, \frac{\nabla \boldsymbol{\mu}}{R \, T} = - ({\bf L} \, {\bf R}^{-1}) \, {\bf R} \, {\bf P} \,  \frac{\nabla \boldsymbol{\mu}}{R \, T} $. We observe that the matrix ${\bf L} \, {\bf R}^{-1}$ possesses the structure ${\bf A} + {\bf Y} \, {\bf S}$, in which ${\bf A}$ is a diagonal matrix with positive entries being regular functions of the state variables, and ${\bf S}$ is off-diagonal with, as well, regular entries. Owing to the property that ${\bf L}^{\sf T}\, {\bf e} = 0$, we compute that
\begin{align}\label{abestimmt}
 a_i + \sum_{k\neq i} s_{ik} \, y_k = 0 \text{ for all } i = 1,\ldots,N \, .
\end{align}
This means that ${\bf L} \, {\bf R}^{-1}$ is of exactly the same form of the Maxwell-Stefan matrix ${\bf B}$ with $f_{ik} := - s_{ik}$. We denote $\tilde{\bf B} :={\bf L} \, {\bf R}^{-1}$. By assumption we have $ \tilde{\bf B} \, {\bf R} = {\bf L} \geq d_0 \, {\bf P}^{\sf T} \, {\bf M} \, {\bf R} \, {\bf P}$, which implies that $\tilde{\bf B} \, {\bf R} \geq \tilde{d}_1 \, {\bf P}^{\sf T} \, {\bf M}^{-1} \, {\bf R} \, {\bf P}$ with $\tilde{d}_1 = d_0 \, (\inf_i M_i)^2$.

Now, we want to invert ${\bf J} = - \tilde{\bf B} \, {\bf R} \, {\bf P} \,  \frac{\nabla \boldsymbol{\mu}}{R \, T}$, for which we simply might apply the implication \eqref{maxstef} $\Rightarrow$ \eqref{fickonsager} exchanging the roles of ${\bf J}$ and ${\bf d}$ therein. We are done.
\end{proof}
Now the auxiliary Lemma, which probably is obvious.
\begin{lemma}\label{Hadamard}
Suppose that ${\bf A}$ and ${\bf B}$ are positive semi-definite and symmetric, and that ${\bf A} \geq {\bf B}$. Then $\text{det}({\bf A}) \geq \text{det}({\bf B})$.
\end{lemma}
\begin{proof}
Denote ${\bf D}$ the diagonal matrix consisting of the eigenvalues $\lambda_1,\ldots,\lambda_N$ of ${\bf A}$, and ${\bf O}$ a unitary matrix such that ${\bf O}^{\sf T} \, {\bf A} \, {\bf O} = {\bf D}$. Then ${\bf D} \geq {\bf O}^{\sf T} \, {\bf B} \, {\bf O}$ and, in particular, $\lambda_i \geq \langle {\bf O}^{\sf T} \, {\bf B} \, {\bf O}\, {\bf e}^i, \,{\bf e}^i \rangle$ for $i=1,\ldots,N$. According to the Hadamard determinant theorem, the determinant of a positive semi-definite matrix is bounded above by the product of the diagonal entries. Therefore
\begin{align*}
\text{det}({\bf A}) = & \Pi_{i=1}^N \lambda_i \geq \Pi_{i=1}^N\langle {\bf O}^{\sf T} \, {\bf B} \, {\bf O}\, {\bf e}^i, \, {\bf e}^i \rangle
\stackrel{\text{Hadamard th.}}{\geq} \text{det}({\bf O}^{\sf T} \, {\bf B} \, {\bf O}) = \text{det}({\bf B}) \, .
\end{align*}
\end{proof}

\noindent

\section{The Diagonal Case: Multicomponent Darken Equation}\label{Darken}
The Fick-Onsager phenomenological coefficients ${\bf L} =[L_{ij}]$ form a non-diagonal, typically dense matrix, where the
non-diagonality immediately follows from ${\bf L} {\bf e}=0$. Also the matrix ${\bf L}'$, which originally appeared in the closure
in section~\ref{Fick-Onsager}, cannot reasonably be assumed as diagonal, since a certain constituent $A_N$ was singled out for elimination of the respective
diffusion flux. Only in case of a dilute (w.r.\ to the $N$-th species, which hence acts as the solvent) system, the reduced system
of the first $N-1$ species equations might be approximated by a system with diagonal diffusion matrix. Evidently, cross-effects are
not expected in the dilute case.

In the Maxwell-Stefan approach, the MS-diffusivities can be interpreted as inverse friction coefficients, hence
the off-diagonal entries are all what counts and a diagonal matrix of Maxwell-Stefan diffusivities makes no sense at all.

Consequently, from the three equivalent approaches to model multicomponent diffusion within the continuum thermodynamical
framework, the novel closure has the unique feature to allow for a non-trivial diagonal closure,
i.e.\ a closure in which the diffusion matrix from \eqref{novel-D} is of the type ${\bf D}={\rm diag}(d_1, \ldots ,d_N)$
with $d_i >0$ that can all be different from each other.
It is interesting to study this closure in more detail and to understand its physical meaning.
For this purpose, we hence employ \eqref{closure-B2} with ${\bf K}:={\bf 0}$ and obtain
\begin{equation}\label{j-FO-diag}
{\bf J}  \, =\, -{\bf P}^{\sf T} \, {\bf R}\, {\mathcal D}\, {\bf M}\, {\bf P} \, \nabla \frac{\boldsymbol{\mu}}{RT}
\end{equation}
with diagonal ${\mathcal D}={\rm diag}(d_1, \ldots ,d_N)$, $d_i>0$.
For comparison with diffusion models from experiments or molecular dynamics, where the Maxwell-Stefan diffusivities
$\D_{ik}=\frac{1}{f_{ik}^{\rm mol}}$ (see Appendix~\ref{MS-ChemEng})
are modeled rather than the mass-based coefficients $f_{ik}$, we change for molar diffusion fluxes to the result
 \[
{\bf J}^{\rm mol}  \, =\, {\bf C} \,{\bf R}^{-1} \,{\bf J} \, =\, -{\bf C}  \,{\bf P} \, {\mathcal D}\, {\bf M}\, {\bf P} \, \nabla \frac{\boldsymbol{\mu}}{RT},
\]
where
\[
{\bf J}^{\rm mol} := [\, {\bf j}_1^{\rm mol} | \cdots |\, {\bf j}_N^{\rm mol} ]^{\sf T}.
\]
We further rewrite the molar fluxes as
\begin{equation}\label{jmol-diag}
{\bf J}^{\rm mol}= - \, [{\bf I}-\frac c  \rho {\bf x} \otimes {\bf m}]\, {\mathcal D}\, {\bf R}\, {\bf P} \, \nabla \frac{\boldsymbol{\mu}}{RT}
\end{equation}
with ${\bf m}:=(M_1,\ldots ,M_N)^{\sf T}$ being the vector of the molar masses of the $A_i$.
Let us note in passing that a completely molar version reads as
\[
{\bf J}^{\rm mol}= -\, [{\bf I}-\frac c  \rho {\bf x} \otimes {\bf m}]\, {\mathcal D}\, {\bf P}^{\sf T} {\bf C}\, \, \nabla \frac{\boldsymbol{\mu}^{\rm mol}}{RT},
\]
but will not be needed further.

We are now going to compute the Maxwell-Stefan diffusivities which correspond to this core-diagonal case.
For this purpose, we employ the molar-based analog of \eqref{Closure-MSa} and \eqref{MS-matrix} from Appendix~\ref{MS-ChemEng},
viz.\
\begin{equation}
- {\bf B}^{\rm mol} {\bf J}^{\rm mol} = {\bf R}\, {\bf P} \, \nabla \frac{\boldsymbol{\mu}}{RT}
\end{equation}
with
\begin{equation}\label{Bij-molar}
B_{ij}^{\rm mol}=- \frac{x_i}{\D_{ij}} \;\mbox{ for } i\neq j, \quad B_{ii}^{\rm mol}= \sum_{k\neq i} \frac{x_k}{\D_{ik}},
\end{equation}
in order to replace $- {\bf R}\, {\bf P} \, \nabla \frac{\boldsymbol{\mu}}{RT}$ by ${\bf B}^{\rm mol} {\bf J}^{\rm mol}$
in \eqref{jmol-diag}.
This yields
\[
{\bf J}^{\rm mol}\, =  \, [{\bf I}-\frac c  \rho {\bf x} \otimes {\bf m}]\, {\mathcal D}\, {\bf B}^{\rm mol} {\bf J}^{\rm mol},
\]
hence ${\mathcal D}\, {\bf B}^{\rm mol} {\bf J}^{\rm mol}$ is of the form
\[
{\mathcal D}\, {\bf B}^{\rm mol} {\bf J}^{\rm mol} = {\bf J}^{\rm mol} + {\bf x} \otimes {\bf a}
\]
for some vector field ${\bf a}$. To compute ${\bf a}$, we exploit the fact that ${\bf e}^{\sf T} {\bf B}^{\rm mol}=0$
and multiply the last equation with $({\mathcal D}^{-1} {\bf e})^{\sf T}$ from the left to obtain
\[
{\bf a} = -\, \frac{({\mathcal D}^{-1} {\bf e})^{\sf T} {\bf J}^{\rm mol}}{\langle {\mathcal D}^{-1} {\bf e}, {\bf x} \rangle}.
\]
Consequently,
\[
{\bf B}^{\rm mol} {\bf J}^{\rm mol} = {\mathcal D}^{-1} \Big( {\bf I} - \frac{ {\bf x} \otimes {\mathcal D}^{-1} {\bf e}}
{ \langle  {\bf x}, {\mathcal D}^{-1} {\bf e}  \rangle} \Big) {\bf J}^{\rm mol}.
\]
For $i\neq j$ we hence obtain
\[
B^{\rm mol}_{ij} = - \, \frac{x_i \, /\, d_i d_j}{\sum_{k=1}^N x_k \, / \, d_k}.
\]
Therefore, we  obtain
\begin{equation}\label{MS-from-diag}
\frac{1}{\D_{ij}} = \frac{d_i^{-1} \, d_j^{-1}}{\sum_{k=1}^N x_k d_k^{-1}}\qquad (i\neq j)
\end{equation}
for the Maxwell-Stefan diffusivities in case the novel closure is applied with diagonal diffusion
${\mathcal D}={\rm diag}(d_1,\ldots, d_N)$. Recall that the $d_i$ are functions of the thermodynamic state,
say $(T,c,{\bf x})$ in the molar-based variant. Hence, we finally write \eqref{MS-from-diag} as
\begin{equation}\label{MC-Darken}
\D_{ij}(T,c,{\bf x}) = d_i (T,c,{\bf x})\, d_j(T,c,{\bf x}) \, \sum_{k=1}^N  \frac{x_k}{d_k (T,c,{\bf x})}\qquad (i\neq j).
\end{equation}
Note that in the binary case $N=2$, \eqref{MC-Darken} reduces to the classical Darken equation
\[
\D_{12}=x_1 d_2 + x_2 d_1,
\]
introduced in \cite{Darken} for diffusion in solid metals.

We refer to equation \eqref{MC-Darken} as the \emph{multicomponent Darken equation}. It has been introduced in \cite{Bardow2011}
as a means to predict Maxwell-Stefan diffusivities, generalizing the classical Darken equation from binary to multicomponent systems.
In the notation of \cite{Bardow2011}, the $\D_{ij}$ are given by the relations
\begin{equation}
\D_{ij} = \frac{D_{i,\rm self}\, D_{j,\rm self}}{D_{\rm mix}}
\quad \mbox{ with } \quad
\frac{1}{D_{\rm mix}}= \sum_{k=1}^N \frac{x_k}{D_{k,\rm self}},
\end{equation}
i.e.\ the $d_i >0$ from the core-diagonal closure coincide with the so-called self-diffusion or (since the latter is a bit misleading)
tracer diffusion coefficients of the substance $A_i$ in the mixture. Note that \eqref{MC-Darken} yields $\D_{ij}(T,c,{\bf x})=d_i (T,c,{\bf x})$ in the limit as $x_j\to 1-$ (hence $x_i \to 0+$).
To obtain a complete description, the dependence of $d_i$ on $(T,c,{\bf x})$, in particular on the composition,  needs to be given.
There a several different empirical relationships employed in the literature.
In \cite{Bardow2011}, the self-diffusion coefficients are modeled via
\begin{equation}\label{Diself}
\frac{1}{D_{i,\rm self}}= \sum_{k=1}^N \frac{x_k}{D_{i,\rm self}^{x_k\to 1}},
\end{equation}
where the constant coefficients $D_{i,\rm self}^{x_k\to 1} >0$ denote the diffusivity of $A_i$ as a diluted species in a binary mixture of $A_i$ and $A_k$.
For other models see, e.g., \cite{TK-book}, \cite{KvB2005} and the recent publications cited in the introduction above.

The identification of the multicomponent Darken equation as the core-diagonal special case of the novel closure above
proves the consistency of this generalization. At this point it is interesting to observe that in \cite{Bardow2011},
the multicomponent Darken equation has been obtained via an ad hoc argument to infer a certain structure concerning the
dependence of the Fick-Onsager coefficients on the molar fractions (i.e., on the composition) and then to compute
the corresponding Maxwell-diffusivities for $N=2,3,4$, but keeping only the leading terms in the molar fractions.
While motivating the form of the multicomponent Darken equation, this is not a strictly consistent derivation:
the latter approximation, if done directly for the Fick-Onsager coefficients $L_{ij}$, would
lead to a diagonal matrix ${\bf L}$, but then the diagonal entries must vanish, too, due to the constraint ${\bf L}{\bf e}= {\bf 0}$.
Indeed, the core-diagonal closure \eqref{j-FO-diag} gives the corresponding Fick-Onsager coefficients as
\begin{equation}\label{FO-coeff-diag}
L_{ij} = \rho_i \big( \lambda_i \delta_{ij} - y_j (\lambda_i + \lambda_j - \sum_{k=1}^N y_k \lambda_k ) \big)
\end{equation}
with $\lambda_i := d_i M_i$ and the Kronecker symbol $\delta_{ij}$.
This is a dense matrix and there is no rational in order for ${\bf L}$ being close to a diagonal matrix.

The sound modeling of Maxwell-Stefan diffusivities is an important ongoing research question;
cf.\ \cite{KvB2005}, \cite{KvB2016},
\cite{Wolf2018}, \cite{G-Vrabec2020}, \cite{Peters2020}, the review article \cite{Krishna2019}
and further references given in these works.

\noindent

\section{\textcolor{black}{On the Sign of Diffusion Coefficients}}\label{limits}


So far we have proved the thermodynamic equivalence of the three proposed forms of the diffusion fluxes. A glance at the relevant literature shall however lead to amend this picture. Indeed, each of the three different closure approaches might also convey some specific information, mostly in the form of additional assumptions on the phenomenological coefficients involved. For instance, in any of the three closure schemes, we might assume that the diffusivities are constant or that certain coefficients possess a definite sign. Or, more generally, we might even ask which diffusion coefficients are really \emph{phenomenological diffusivities} such that we might prescribe or measure them.

The viewpoint of thermodynamical consistency does not suffice to answer these questions, and even the requirement of positivity for smooth solutions of PDEs for multicomponent diffusion does not exhaust this subject.

\subsection{Additional Sign Conditions on Phenomenological Coefficients}

In this paragraph we discuss issues concerning phenomenological coefficients that play a role in the practical use of multicomponent diffusion models. \\[0.2ex]

\noindent {\bf 1.\ Positivity of the MS-diffusivities.} It is common to impose to the Maxwell-Stefan form of the fluxes the additional condition that all friction coefficients $f_{ik}$ are strictly positive functions of the state variables. As recalled in the section \ref{Drazin}, this condition is always sufficient for the invertibility of the MS-equations. Whether negative Maxwell-Stefan diffusivities might occur or not is however a recurrent question in the applied literature, for instance in \cite{chakrawangeapen}, \cite{kraaijeveldwesselingh}. In the latter paper, it is proved that constant $f_{ik}$ are compatible with the thermodynamic requirement \eqref{entropy-MSc} only if $f_{ik} \geq 0$ for all $i \neq k$. This observation can be slightly generalized as follows.
\begin{lemma}\label{posMSdiffusivities}
 For all $1 \leq i < k \leq N$, assume that $f_{ik}$ is a regular function of the state variables defined at each $T > 0$ for all $\rho_1,\ldots,\rho_N \geq 0$. Assume moreover that for arbitrary $\rho_1,\ldots,\rho_N$, the inequality \eqref{entropy-MSc} is valid for all ${\bf z} \in \{{\bf e}\}^{\perp}$. In order that for all $i < k$, the function $f_{ik}(T, \cdot)$ is strictly positive at all $\rho_1,\ldots,\rho_N >0$, it is sufficient that all $f_{ik}$ describe binary interactions, that is $f_{ik} = f_{ik}(T,\rho_i,\rho_k)$ for $T>0$, all $i < k$ and all $\rho_1,\ldots,\rho_N \geq 0$.

Moreover, it is in general false that the conditions stated in {\rm (B)} for the matrix ${\bf B}$ imply that $f_{ik}$ is positive for all $i\neq k$.
\end{lemma}
\begin{proof}
 We first show that the condition of binary interactions ensures the strict positivity of $f_{ik}$. We thus fix $T>0$ and we suppose that, for all $i<k$, we have $f_{ik} = f_{ik}(T,\rho_i,\rho_k)$ for all $\rho_1,\ldots,\rho_N \geq 0$. Then for arbitrarily chosen $1 \leq \alpha < \beta \leq N$, ${\bf z} \in \{{\bf e}\}^{\perp}$ and $\rho_1, \ldots,\rho_N \geq 0$ such that $\rho_i = 0$ for $i \not\in \{\alpha, \beta\}$, the assumption \eqref{entropy-MSc} yields
 \begin{equation*}
0 <
 \,\sum_{1\leq i<k \leq N} f_{ik}(T,\rho_i,\rho_k) \, y_i \, y_k \, (z_i - z_k)^2 = f_{\alpha\beta}(T, \, \rho_{\alpha}, \, \rho_{\beta}) \, y_{\alpha} \, y_{\beta} \, (z_{\alpha} - z_{\beta})^2 \, .
\end{equation*}
Hence, constructing ${\bf z}$ via $z_i = 0$ for $i \not\in \{\alpha,\beta\}$ and $z_{\alpha} = - z_{\beta}$, it is obvious that $f_{\alpha\beta}(T,\rho_{\alpha},\rho_{\beta}) > 0$ for all $\rho_{\alpha}, \, \rho_{\beta} > 0$. This proves that $f_{\alpha\beta}$ is strictly positive.

Next we show the second claim by constructing a counterexample. We assume $N = 3$, and we show that, if $f_{13}$ is allowed to depend on $\rho_2$ it will assume negative values for some choices of $\rho_1,\rho_2,\rho_3$. It is easy to generalize the arguments to arbitrary large a number of components. We first let $\tau^1_{ik}(\rho_1,\rho_2,\rho_3) :=  \rho \, y_i \, y_k$ for all $1 \leq i < k \leq 3$ and fill the diagonal of the matrix ${\boldsymbol \tau}^1$ according to \eqref{tau-extension} in order to satisfy all conditions \eqref{fik-assumptions}. We introduce a second matrix ${\boldsymbol \tau}^2$ via
\begin{align*}
 {\boldsymbol \tau}^2(\rho_1,\rho_2,\rho_3) = \rho \, y_1 y_2 y_3 \, {\bf A}, \quad {\bf A} :=
\left[
\begin{array}{ccc}
a-1   & 1 & - a  \\[1ex]
1   & a-1 & - a \\[1ex]
 - a   &  - a & 2a
\end{array}
\right]
 \text{ with } a>2 \, .
\end{align*}
The matrix ${\boldsymbol \tau}^1$ results from the choice $f_{ik} \equiv 1$ in \eqref{fric-coeff}, hence it is symmetric and strictly positive definite on $\{{\bf e}\}^{\perp}$ for all ${\bf y}$ with $y_1,y_2,y_3>0$. The constant matrix ${\bf A}$ is symmetric, satisfies ${\bf A} {\bf e} = {\bf 0}$ and, for $a > 2$, it can be verified that the non-zero eigenvalues are strictly positive. Next we fix any state $\rho^0_1, \rho^0_2,\rho^0_3 > 0$, and we let
\begin{align*}
{\boldsymbol \tau} := |{\bf y} - {\bf y}^0|^2 \, {\boldsymbol \tau}^1 +  {\boldsymbol \tau}^2 \, ,
\end{align*}
which can be verified to be symmetric. Using that both ${\boldsymbol \tau}^1$ and ${\bf A}$ are positive definite on $\{{\bf e}\}^{\perp}$, the same follows for ${\boldsymbol \tau}$ and similarly we see that ${\boldsymbol \tau}{\bf e} = {\bf 0}$. Corresponding Maxwell-Stefan friction coefficients are next defined according to \eqref{fric-coeff}, meaning that
\begin{align*}
 f_{ik} = \frac{\tau_{ik}}{\rho \, y_i\, y_k} = |{\bf y} - {\bf y}^0|^2 + y_{[ik]} \, a_{ik} \, ,
\end{align*}
where $[ik] \neq i,k$ is the complementary index. These $f_{ik}$ are obviously regular functions of the state variables. Now, due to the special choice of ${\bf A}$, we among others obtain that
\begin{align*}
 f_{13} = |{\bf y} - {\bf y}^0|^2 - y_{2} \, a \,
\end{align*}
and, depending on the choice of $a > 2$, a neighborhood of the point ${\bf \rho}^0$ is readily constructed such that $f_{13} < 0$.
\end{proof}
Hence, thermodynamics does not require the positivity of the Maxwell-Stefan diffusivities. Note, moreover, that in a mixture like an electrolyte, where the species cannot vanish independently of each other, the arguments of the preceding Lemma do not apply. Even binary or constant diffusivities might then turn negative, hence such cases must be investigated separately.

From the viewpoint of the thermodynamic equivalence, the next natural question is to ask which kind of Fick-Onsager closure of form {\rm (A)} yields an equivalent Maxwell-Stefan form with positive friction coefficients. The next Lemma relates this problem to the interesting algebraic question of identifying the inverse of strict {\em Z}-matrices. Recall that a real square-matrix is a {\em Z}-matrix iff all off-diagonal elements are non-positive.
\begin{lemma}\label{Z}
 Let $\mathbf{L} \in \mathbb{R}^{N\times N}$ be symmetric, positive definite on $\{\mathbf{e}\}^{\perp}$ and $\mathbf{L}\, \mathbf{e} = {\bf 0}$. Let $\mathbf{J} = -  \mathbf{L} \, \nabla \frac{\boldsymbol{\mu}}{RT}$ and $\mathbf{L}^{\sharp}$ be the generalized inverse of $\mathbf{L}$. Then, the two following statements are equivalent:
 \begin{enumerate}[(1)]
  \item The matrix $\mathbf{B} := \mathbf{R} \, \mathbf{P} \, \mathbf{L}^{\sharp} \, \mathbf{P}^{\sf T} = (\mathbf{L} \, \mathbf{R}^{-1})^{\sharp}$ is a strict {\em Z}-matrix;
 \item $\mathbf{J}$ obeys the equations (4.15) with coefficients $f_{ik} := - b_{ik}/y_i > 0$ for all $i \neq k$.
 \end{enumerate}
 \end{lemma}
\begin{proof}
As seen, $\mathbf{J} = - \mathbf{L} \, \nabla \frac{\boldsymbol{\mu}}{RT}$ implies that $\mathbf{J} = - \mathbf{L} \, \mathbf{R}^{-1} \, \mathbf{R} \, \mathbf{P} \, \nabla \frac{\boldsymbol{\mu}}{RT}$. We apply the generalized inverse $\mathbf{B} = (\mathbf{L} \, \mathbf{R}^{-1})^{\sharp}$. Then $\mathbf{B} \, \mathbf{J} = - \mathbf{R} \, \mathbf{P} \, \nabla \frac{\boldsymbol{\mu}}{RT}$. We may compute that
\begin{align*}
 {\bf R} \, {\bf P} \, {\bf L}^{\sharp} \, {\bf P}^{\sf T} \, {\bf L} \, {\bf R}^{-1} = & {\bf R} \, {\bf P} \, {\bf L}^{\sharp} \, {\bf L} \, {\bf R}^{-1}
 = {\bf R} \, {\bf P} \,( {\bf I} - \frac{1}{N} \, {\bf e} \otimes {\bf e} ) \, {\bf R}^{-1} \\
 = & {\bf R} \, {\bf P}  \, {\bf R}^{-1} =  {\bf I} - {\bf y} \otimes {\bf e} \,.
\end{align*}
Similarly ${\bf L} \, {\bf R}^{-1} , {\bf R} \, {\bf P} \, {\bf L}^{\sharp} \, {\bf P}^{\sf T} = {\bf I} - {\bf y} \otimes {\bf e}$, showing that $ {\bf B} = ({\bf L} \, {\bf R}^{-1})^{\sharp} =  {\bf R} \, {\bf P} \, {\bf L}^{\sharp} \, {\bf P}^{\sf T}$.

We define $f_{ik} := -b_{ik}/y_i$, and we obtain that the friction coefficients are strictly positive iff $\mathbf{B}$ is a strict {\em Z}-matrix or, in other words, iff ${\bf L}^{\sharp} - {\bf L}^{\sharp} {\bf y} \otimes {\bf e} -{\bf e} \otimes {\bf L}^{\sharp} {\bf y} - \langle {\bf L}^{\sharp} {\bf y}, \, {\bf y} \rangle \,{\bf e}  \otimes {\bf e}$ has negative off-diagonals.
\end{proof}
The positivity of the Maxwell-Stefan diffusivities is thus equivalent to $\mathbf{L}\, \mathbf{R}^{-1}$ being the generalized inverse of a {\em Z}-matrix. Unfortunately, there is apparently no handy characterization for the generalized inverse of (strict) {\em Z}-matrices available in the literature (even in the regular case).
Therefore, computation of one of the matrices $(\mathbf{L} \, \mathbf{R}^{-1})^{\sharp}$ or $\mathbf{R} \,  \mathbf{P}\, \mathbf{L}^{\sharp} \, \mathbf{P}^{\sf T}$ seems the only way to verify the positivity of the corresponding Maxwell-Stefan coefficients.

Which subclass of Fick-Onsager coefficients leads to constant or binary friction coefficients in the Maxwell-Stefan form would be a question of the same quality.

Note however that the occurrence of Maxwell-Stefan diffusivities which are \emph{irregular (infinite) functions of the state-variables} is prohibited by our statement of the equivalence in Theorem \ref{Th-2}, unless the corresponding Fick-Onsager coefficients are already degenerated (enhanced diffusion); according to \cite{chakrawangeapen}, this phenomenon occurs in certain simulations.\\[0.2ex]

\noindent {\bf 2.\ {\em M}-Matrix property in the Fick-Onsager scheme.} If in the closure of type (A), the off-diagonal matrix ${\bf S}$ is elementwise negative, then ${\bf L}$ is a strict {\em Z}-Matrix of rank $N-1$ with kernel $\{{\bf e}\}$. Recalling that ${\bf L}$ is positive semi-definite and symmetric, we call $\lambda$ the largest eigenvalue of ${\bf L}$, implying the inequality $L_{ii} \leq \lambda$. Hence, ${\bf L}$ can also be written as $\lambda \, {\bf I} - {\bf G}$ with
\begin{equation*}
{\bf G} = \lambda \, {\bf I} - {\rm diag}\,({\bf L}) - {\bf L}_{\text{off}} \, .
\end{equation*}
If now $ S_{ij} < 0$, then $ - {\bf L}_{\text{off}} = {\bf R} \, {\bf S} \, {\bf Y}$ and this ${\bf G}$ are elementwise nonnegative which charaterises ${\bf L}$ as a so-called singular {\em M}-Matrix. The case of identical Maxwell-Stefan coefficients $f_{ik} = \bar{d}^{-1}$, leading to ${\bf B}^{\sharp} = \bar{d} \, {\bf P}^{\sf T}$ and to ${\bf L} = \bar{d} \, {\bf R} \, {\bf P}$, shows that this property can sometimes--in rare cases--be expected.

Here we restrict ourselves to the following simple observation: Constant $S_{ij}$ in (A) are possible only if $S_{ij} < 0$ for all $i\neq j$. To see this, it suffices to recall that $a_i = - \sum_{j\neq i} S_{ij} \, y_j$ for all $i$ (cf.\ \eqref{abestimmt}). We let $y_i \rightarrow 0+$. Since $a_i \rightarrow a_i({\bf y}_0)>0$ is required, we obtain that $\sum_{j\neq i} S_{ij} \, y_j <0$. This is now valid for all admissible ${\bf y}$ with $y_i = 0$, and clearly $S_{ij} < 0$ follows for all $j\neq i$.

As seen in the preceding paragraph, it remains in general open to determine which subclass of Maxwell-Stefan coefficients yield strictly negative $S_{ij}$ in the Fick-Onsager representation. This is, in fact, the same question as in Lemma \ref{Z}.\\[0.2ex]

\noindent {\bf 3.\ Elementwise diagonal positivity for the novel closure scheme.} The third example occurs if we ask that the novel form (C) of the diffusion fluxes be valid, where the diagonal part $\mathcal{D}$ is elementwise positive. We shall restrict our observations to the core-diagonal case.
\begin{lemma}
Let $\mathcal{D} = {\bf \rm diag}(d_1, \ldots,d_N)$.
 For all $1 \leq i\leq N$, assume that $d_{i}$ is a regular function of the state variables defined at each $T > 0$ for all $\rho_1,\ldots,\rho_N \geq 0$. Assume moreover that for arbitrary $\rho_1,\ldots,\rho_N$, the matrix ${\bf P}^{\sf T} \, {\bf R} \, \mathcal{D} \, {\bf M} \, {\bf P}$ is positive definite on $\{{\bf e}\}^{\perp}$. In order that for all $i$, the function $d_{i}(T, \cdot)$ is strictly positive at all $\rho_1,\ldots,\rho_N >0$, it is sufficient that $d_{i}$ is independent of $\rho_i$.

 Moreover, it is in general false that, if $\mathcal{D}$ satisfies all conditions stated in {\rm (C)} with ${\bf K} = {\bf 0}$, the functions $d_i$ are positive.
\end{lemma}
\begin{proof}
%
%

If $d_i$ does not depend on $\rho_i$, then the consequence of (C) that $d_i \rightarrow d_i({\bf y}_0) > 0$ for $y_i \rightarrow 0+$ implies that $d_i$ must be strictly positive over all compositions.

In order to prove the second claim, a counterexample is sufficient. We exploit the result of appendix~\ref{ternarysection} showing that, for $N = 3$, every Maxwell-Stefan closure satisfying (B) is core-diagonal with
\begin{align*}
 d_{[ik]} = f_{ik}/D_0, \quad D_0 = {\rm trace} \, ({\rm adj} ({\bf B})) > 0
\end{align*}
with the Maxwell-Stefan matrix ${\bf B}$ as in \eqref{MS-matrix} and $[ik] \neq i,k$ the complementary index. Now it suffices to apply the counterexample of Lemma \ref{posMSdiffusivities}, and we obtain in particular that $\mathcal{D}_0 \, d_2 = f_{13}$ is negative whenever $f_{13}$ is negative.
%
\end{proof}
More generally, it is not possible to show that every Fick-Onsager closure obeying (A) yields an equivalent novel form with elementwise positive $\mathcal{D}$. Thus, the latter condition introduces a new strict subclass among the thermodynamic consistent closures of type (A).
\begin{lemma}
 Let $\mathbf{L}= {\bf R} \, ({\bf A} + {\bf S} \, {\bf Y}) \in \mathbb{R}^{N\times N}$ be symmetric, positive definite on $\{\mathbf{e}\}^{\perp}$ with $\mathbf{L}\, \mathbf{e} = 0$. Let $\mathbf{J} = -  \mathbf{L} \, \nabla \frac{\boldsymbol{\mu}}{RT}$. Then, in order that $\mathbf{J}$ satisfies the novel closure equations {\rm (C)} with $\mathcal{D}$ elementwise positive, the following condition for $\mathbf{S}$ is necessary and sufficient: For $i = 1,\ldots,N$,
 \begin{align}\label{posdiag}
  \langle\, \mathbf{S} {\bf y}, \, {\bf e}^i \, \rangle + \frac{2}{N-2} \,  \langle \, \mathbf{S} {\bf e}, \, {\bf e}^i \, \rangle \,  y_i\leq \frac{\langle \mathbf{S} {\bf e}, \, {\bf e} \rangle }{(N-2) \, (N-1)} \,  y_i  .
  \end{align}
\end{lemma}
\begin{proof}
We have already computed in Theorem \ref{Th-1} (see \eqref{DfromAandS}, \eqref{KfromAandS}) that the matrix $\mathcal{D}$ is related to ${\bf A}$ and ${\bf S}$ via
\begin{align*}
 \mathcal{D} = & {\bf A} \, {\bf M}^{-1} + 2 \,\frac{c}{\rho} \, \text{diag}({\bf b}) \,{\bf X}
 \end{align*}
 with ${\bf b} = - \frac{1}{N-2} \, ({\bf I} -\frac{1}{2 \, (N-1)} \, {\bf e} \otimes {\bf e}) \, {\bf S} \, {\bf e}$.

 On the other hand, we can make use of ${\bf L} {\bf e} = {\bf 0}$ to compute that
 ${\bf A} \, {\bf e}  = - {\bf S} {\bf y}$ (cf.\ \eqref{abestimmt}). This allows to also compute
 \begin{align*}
  d_i = \frac{1}{M_i} \, \left( - \langle\, \mathbf{S} {\bf y}, \, {\bf e}^i \, \rangle - \frac{2}{N-2} \,  \langle \, \mathbf{S} {\bf e}, \, {\bf e}^i \, \rangle \,  y_i  + \frac{\langle \mathbf{S} {\bf e}, \, {\bf e} \rangle }{(N-2) \, (N-1)}  \,  y_i \right) \, .
 \end{align*}
Hence, $\mathcal{D}$ is positive iff the condition \eqref{posdiag} is valid.
\end{proof}
The core-diagonal new closure with elementwise positive matrix $\mathcal{D}$ induces a Maxwell-Stefan form with strictly positive coefficients. In this case, in fact, we have $f_{ik} = (\sum_{j} d_jy_j)/(d_i\, d_k)$. We refer to the section \ref{Darken} for details.\\[0.2ex]

\noindent

\subsection{Diffusion Matrices}

While the picture of thermodynamic equivalence seems clear and complete, the concept of a diffusion coefficient is essentially plural. This question is discussed for instance in \cite{Miller} and in the second paragraph of \cite{CB-review}.

In the literature, several different objects have been called \emph{diffusion matrix}. Following a classification proposed in \cite{Miller}, we distinguish between thermodynamic diffusion coefficients and Fickian diffusion coefficients. The thermodynamic diffusion coefficients describe proportionality relations between fluxes and driving forces, while the Fickian coefficients are proportionality factors between fluxes and gradients of concentrations or fractions.\\[0.2ex]

\noindent {\bf A. Fickian diffusivities.} The coefficient matrix ${\bf D}^{\rm Fick} = [D_{ik}]$ is associated with the representation \eqref{Fick-MC} of the diffusion fluxes. In the language of the present paper, we obtain these Fickian diffusivities starting from the Fick-Onsager representation ${\bf j}^{\text{mol}} = - {\bf M}^{-1} \, {\bf L} \, \nabla \frac{{\boldsymbol \mu}}{RT}$. Recalling that $\mu_i = \partial_{\rho_i}(\rho\psi)$, we introduce the Hessian ${\bf H} := {\boldsymbol D^2}_{{\boldsymbol\rho},{\boldsymbol \rho}}(\rho\psi)$ of the free energy. In the isothermal case, we obtain that
 \begin{align*}
  {\bf j}^{\text{mol}} = - \frac{1}{RT} \, {\bf M}^{-1} \, {\bf L} \, {\bf H} \, {\bf M} \, \nabla {\bf c} \, .
 \end{align*}
This means that ${\bf D}^{\rm Fick} = {\bf M}^{-1} \, {\bf L} \, {\bf H}\, {\bf M}/(RT)$.

The matrix of Fickian diffusivities is the one playing a role in the PDE analysis of diffusion and reaction--diffusion systems. If we follow \cite{Miller}, it is also the relevant matrix in measurements and experiments. From our general viewpoint, this matrix can be written as the product of ${\bf M}^{-1} \, {\bf L} \, {\bf M}^{-1}$, which is positive semi-definite, and of ${\bf M} \, {\bf H} \, {\bf M}/(RT)$, which is positive definite if, as required by the second law of thermodynamics, the free energy is a (strict) convex function of the partial mass densities. Hence, as shown in \cite{Miller} and \cite{DB-MS}, ${\bf D}^{\rm Fick}$ possesses only real positive eigenvalues. It always generates a normal elliptic operator in the PDEs (see \cite{JP-MS}, \cite{bothepruess}).

It is well known that the matrix ${\bf D}^{\rm Fick}$ is in general not symmetric.
\begin{lemma} For an ideal mixture, the matrix ${\bf D}^{\rm Fick}$ is non-symmetric.
\end{lemma}
\begin{proof}
Due to the properties of the Onsager matrix ${\bf L}$, the vector ${\bf m}$ of molar weights is a left eigenvector with trivial eigenvalue for ${\bf D}^{\rm Fick}$. Assuming ${\bf D}^{\rm Fick}$ symmetric, hence ${\bf D}^{\rm Fick} \, {\bf m} = ({\bf D}^{\rm Fick})^{\sf T} \, {\bf m} = 0$, the vector ${\bf H} \, {\bf M} \, {\bf m}$ must belong to the span of the vector ${\bf e}$.

It follows that $\sum_{j=1}^N H_{ij} \, M_j^2 = \alpha$ for some scalar function $\alpha$ of the state-variables and for all $i$. But note that the equations $\sum_{j=1}^N H_{ij} \, M_j^2 = \alpha$ also characterize $\alpha = \partial_{\rho_i} f$ for all $i$ with $f := \sum_{j=1}^N \mu_j \, M_j^2$. Hence, $\alpha \, {\bf e}$ is a gradient vector for the function ${\bf \rho} \mapsto f(T, \, {\bf \rho})$, implying that $\alpha = \alpha(T,\rho)$ is a function of the temperature and the total mass density only. Now, for ideal chemical potentials, it is particularly easy to show a contradiction.
We let $x_j \rightarrow 0+$ for some component, and we see that $\sum_{j=1}^N H_{ij} \, M_j^2$ must explode like $M_j/(M_ix_j)$ while $\alpha$, being independent on composition, would remain bounded.

%
\end{proof}
It has been widely discussed in the literature (see \cite{Miller}, \cite{chenengstromagren}, \cite{mutorufirooz}, \cite{Vrabec2019}) that ${\bf D}^{\rm Fick}$ can also not be expected to be diagonally positive. This is true for the most general case,
but this statement must be tempered. There are many interesting cases where diagonal positivity is to expect.

A so-called \emph{simple} mixture is characterized by the structure $\rho\psi = \sum_{i=1}^N f_i(T, \, c_i)$ of the free energy, which possesses a diagonal Hessian. In such cases, $D_{ii} = L_{ii} \, f_{i}^{\prime\prime}(T, \, c_i)/(RTM_i^2)$ is always strictly positive.

Another important point is that Fickian diffusion is often considered in isobaric, isochoric or related contexts, where some additional function of the densities - typically the concentration $c$, the pressure $p$, the volume or the specific volume - is assumed constant. In such cases, not the full Hessian ${\bf H}$ is relevant for the computation of ${\bf D}^{\rm Fick}$. Considering for instance an isothermal, ideal mixture which is moreover isobaric, the Hessian ${\bf H}$ reduces to $RT \, ({\bf M}^{-1} \, {\bf R}^{-1} - \frac{1}{c} \, {\bf M}^{-1} \, ({\bf e} \otimes {\bf e}) \, {\bf M}^{-1}])$ and for the matrix of Fickian diffusivities we obtain that
\begin{align*}
 {\bf D}^{\rm Fick} = & {\bf M}^{-1} \, ({\bf L} \,  {\bf R}^{-1} - \frac{1}{c} \, {\bf L} \, {\bf M}^{-1} \, {\bf e} \otimes {\bf e})\\
 = & {\bf M}^{-1} \, ({\bf B}^{\sharp} -  {\bf B}^{\sharp} \, {\bf X} \, {\bf e} \otimes {\bf e}) \, .
\end{align*}
For the last identity, we assumed the Maxwell-Stefan form of the Fick-Onsager matrix.
Hence, at least for isobaric, isothermal, ideal systems with Maxwell-Stefan diffusion, the matrix $ {\bf D}^{\rm Fick} $ consists of regular functions of the state variables. Further, the latter equations show obviously that the diagonal elements satisfy $D_{ii} \rightarrow d_i^+ > 0$ for $x_i \rightarrow 0+$.
Diagonal dominance can reasonably be expected, as shown by the example of identical Maxwell-Stefan coefficients $f_{ik} = \bar{d}^{-1}$, which yields ${\bf B}^{\sharp} = \bar{d} \, {\bf P}^{\sf T}$ and ${\bf L} = \bar{d} \, {\bf R} \, {\bf P}$. Hence
\begin{align*}
 D_{ik} = \bar{d} \, (\delta_{ik} - y_i - x_i \, (1 - \frac{c}{\rho} \, M_i ) ) =\bar{d} \,  (\delta_{ik} - x_i) \, ,
\end{align*}
which is a singular {\em M}-Matrix with positive diagonals. Whether similar properties can be expected for certain subclasses of MS-diffusivities or matrices ${\bf L}$ is an open question that we cannot exhaust in the context of this investigation.

From the point of view of the novel closure scheme, we wish to point out the following interesting property of Fickian diffusion in the case of core-diagonal closure.
\begin{lemma}
Consider an ideal isothermal and isobaric system subject to the core-diagonal closure relation ${\bf L} = {\bf P}^{\sf T} \, {\bf M} \, \mathcal{D} \, {\bf R} \, {\bf P}$. If $\mathcal{D} = {\rm diag}(d_1, \ldots,d_N)$ with $d_i \geq 0$, then the diagonal entries of ${\bf D}^{\rm Fick}$ satisfy
 \begin{align*}
  D_{ii} \geq d_i \, (1-y_i) \, (1-x_i) \text{ for } i = 1,\ldots,N \, .
 \end{align*}
\end{lemma}
\begin{proof}
By assumption, ${\bf H} = RT \,  \big({\bf M}^{-1} \, {\bf R}^{-1} - \frac{1}{c} \, {\bf M}^{-1} \, ({\bf e} \otimes {\bf e}) \, {\bf M}^{-1} \big)$. We calculate that
\begin{align*}
 {\bf D}^{\text{Fick}} = & \frac{1}{RT} \, {\bf M}^{-1} \, {\bf L} \, {\bf H} \, {\bf M} =  \mathcal{D} - \mathcal{D} \, {\bf x} \otimes {\bf e} - {\bf M}^{-1} \, { \bf y} \otimes (\mathcal{D} \, {\bf M} \, {\bf e} - \frac{\rho}{c} \, \sum_{k=1}^N d_k \, y_k \, {\bf e})  \, .
\end{align*}
Hence
\begin{align*}
 D_{ii} = & \, d_i \, (1-x_i) - d_i \, y_i + \frac{y_i}{M_i} \, \frac{\rho}{c} \, \sum_{k=1}^N d_k \, y_k\\
 = & \, d_i \, (1-x_i) - d_i \, y_i + x_i \, \sum_{k=1}^N d_k \, y_k \,
\end{align*}
and, estimating $\sum_{k=1}^N d_k \, y_k \geq d_i \, y_i$, the claim follows.
\end{proof}
Let us also refer to the interesting recent work \cite{Vrabec2019}, showing that diagonal positivity of the Fickian diffusion matrix could, in the ternary case, be a question also associated with the proper choice of the frame of reference for the diffusion velocity.

Let us remark that, still in the isothermal and isobaric (or similar) context, another usual representation of the diffusion flux is
\begin{align}\label{Fick-2}
j^{\rm mol}_i = - c \, \tilde{D}_{ik} \, \nabla x_k \, .
\end{align}
This representation is also equivalent, since an additional assumption of the type $p(c, \, x_1, \ldots,x_N) = p_0$ (isobaric case) is available, and allows to compute the complete thermodynamic driving forces in terms of the gradients $\nabla x_k$ only. The form of the diffusion matrix $\tilde{{\bf D}}^{\rm Fick}$ then depends on this equation too. Let us restrict to giving an example for the ideal, isobaric context. As $\nabla x_k = RT \, x_k \, \nabla (\mu_k/(RT))$ in this case, it is readily seen that the matrix $\tilde{{\bf D}}^{\rm Fick}$ is related to ${\bf L}$ via $c \, \tilde{{\bf D}}^{\rm Fick} = {\bf L} \, {\bf X}^{-1}$ and to ${\bf B}$ via $\tilde{{\bf D}}^{\rm Fick} = {\bf B}^{\sharp} \, {\bf M}$  which, again, guarantees diagonal positivity.

As a partial conclusion on the question of diagonal positivity for the Fickian diffusion matrix, let us point at two different aspects. On the one hand, from the viewpoint of the second law of thermodynamics, the Fickian diffusion matrix is the product of two positive (semi-)definite matrices, hence diagonal positivity is not to be expected--although all eigenvalues are non-negatice. On the other hand, representing a major difference between the concept of thermodynamic respectively of phenomenological diffusivity, all properties of ${\bf D}^{\rm Fick}$, resp.\ of $\tilde{{\bf D}}^{\rm Fick}$ depend strongly on the specific underlying free energy model and/or on additional constraint being possibly valid for the physical system. Hence, case by case, diagonal positivity might be observed for certain -- even possibly large -- classes of relevant systems, while it is certainly not a generic feature.\\[0.25ex]

\noindent {\bf B. Thermodynamic diffusivities.}

\begin{enumerate}[(1)]
 \item \label{diffus2} In eq.\ (2.6) of \cite{CB-review} and paragraph 2.5.2 of the book \cite{Giovan}, the symmetric diffusion matrix is defined to be ${\bf D}^{\rm CB} := c \, {\bf R}^{-1} \, {\bf L} \, {\bf R}^{-1}$ or, equivalently, ${\bf D}^{\rm CB} := c \, {\bf R}^{-1} \, {\bf B}^{\sharp}$ according to whether one starts from the Fick-Onsager or the Maxwell-Stefan form of the fluxes;
 \item \label{diffus3} In \cite{CB-review}, eq.\ (2.1), Bird and Curtiss define the non-symmetric, off-diagonal diffusion matrix $\hat{{\bf D}}^{\rm CB} := \frac{c}{\rho} \, [{\bf L} \, {\bf R}^{-1} + {\bf R}^{-1} \, {\bf L}_{\text{off}} {\bf e} \otimes {\bf e}]$;
 \item \label{diffus4} Our novel scheme introduces symmetric diffusivities ${\mathcal D}$, ${\bf K}$ (see \eqref{closure-B2} and (C)) and in particular a diffusion matrix $\mathcal{D} + {\bf X}^{\frac{1}{2}} \, {\bf K} \,  {\bf X}^{\frac{1}{2}}$ (see \eqref{novel-D}).
\end{enumerate}
In this section we restrict ourselves to a few remarks of general character concerning this plurality.

Ad \eqref{diffus2}. The matrix ${\bf D}^{\rm CB}$ is symmetric but, as the proportionality factor between \emph{diffusion velocities} and driving forces, it does not consist of regular functions of the state variables. As noted in \cite{Giovan}, Lemma 7.3.1, the definition ${\bf D}^{\rm CB} := c \, {\bf R}^{-1} \, {\bf B}^{\sharp}$ implies the blow-up of these coefficients for vanishing densities if Maxwell-Stefan closure with regular binary coefficients is assumed. {\color{black} Hence, this matrix might exhibit real drawbacks in practice.}

Ad \eqref{diffus3}. The matrix $\hat{{\bf D}}^{\rm CB}$ can be shown, using eq.\ (2.2) of \cite{CB-review}, to satisfy ${\bf L} \, {\bf R}^{-1} = \frac{\rho}{c} \, \hat{{\bf D}}^{\rm CB} \, {\bf P}^{\sf T}$. Hence also
\begin{align*}
 {\bf L} = \frac{\rho}{c} \, {\bf P}^{\sf T} \, \hat{{\bf D}}^{\rm CB} \, {\bf R} \, {\bf P} \, .
\end{align*}
Thus $\hat{{\bf D}}^{\rm CB}$ can be re-interpreted as introducing the diffusivities in a somewhat similar way as our novel closure scheme, choosing however another constraint--namely, the constraint of vanishing diagonal--in order to eliminate $N$ parameters. In this respect, the Ansatz of the novel closure scheme apparently exhibit some advantages, as it remains symmetric and preserves what we have called 'core-diagonal' diffusion.

\section{Concluding Remarks}
We conclude with some additional comments, mainly related to the new closure scheme.
The introduction of an undetermined velocity to avoid the constraint \eqref{flux-constraint}
can also be understood as the use of a Lagrange parameter to incorporate the dual constraint, i.e.\ \eqref{d-sum}.
Indeed, it is known that instead of evaluating $-\sum_{i=1}^N {\bf u}_i \cdot {\bf d}_i$ under the constraint
$\sum_{i=1}^N  {\bf d}_i =0$, one can equivalently evaluate $-\sum_{i=1}^N ({\bf u}_i + {\bf w})\cdot {\bf d}_i$
with a Lagrange parameter ${\bf w}$; cf.\ \cite{Liu} for more details.
Now, closing for ${\bf u}_i + {\bf w}$ instead of ${\bf d}_i$, this leads to
\[
{\bf U}=- {\bf P}^{\sf T} \, {\bf L} \, \vec{\bf d}
=- {\bf P}^{\sf T} \, {\bf L} \, {\bf R} {\bf P} \nabla \frac{\boldsymbol{\mu}}{RT},
\]
which resembles the first step in section~\ref{sec-novel-scheme}.

The advantage of a formulation being independent of a specific reference frame has been observed long before;
see \cite{Snell}.
In the present setting, the diffusion fluxes w.r.\ to a reference velocity of the type
\[
\hat{\bf v} = \sum_{i=1}^N \omega_i {\bf v}_i
\quad \mbox{ with } \omega_i\geq 0,\; \sum_{i=1}^N \omega_i =1
\]
is straightforward to express. Indeed,
\[
\hat{\bf u}_i = {\bf v}_i -\hat{\bf v} = {\bf u}_i^\ast- \sum_{k=1}^N \omega_k {\bf u}_k^\ast,
\]
hence
\[
\hat{\bf U} = {\bf P}_{\boldsymbol{\omega}} {\bf U}^\ast
\quad \mbox{ with }
{\bf P}_{\boldsymbol{\omega}} = {\bf I}-{\bf e}\otimes \boldsymbol{\omega}
\]
are the diffusion velocities in the reference frame corresponding to $\hat{\bf v}$.

In numerical simulations, the inversion of the Maxwell-Stefan system in every time step and in every mesh cell, resp.\ for every cell face is computationally expensive and, hence,
iterative schemes for approximate solutions have been developed; see \cite{Giovan} and the references given there.
Since the new closure scheme avoids such an inversion, it can provide an interesting alternative.
For the same reason, in many simulations of (reactive) multicomponent flows, the simple Fickian closure, i.e.\ ${\bf j}_i^\ast = - d_i \nabla c_i$, is employed.
In order to enforce the constraint \eqref{flux-constraint}, a correction velocity ${\bf w}$ is introduced such that $\sum_{i=1}^N ({\bf j}_i^\ast+\rho_i {\bf w})=0$; cf.\ \cite{Datta}. This leads to
\[
{\bf j}_i = {\bf j}_i^\ast - y_i \sum_{k=1}^N  {\bf j}_k^\ast,
\]
i.e.\ ${\bf J}={\bf P}^{\sf T} {\bf J}^\ast$ with ${\bf P}^{\sf T}={\bf I} - {\bf y}\otimes {\bf e}$.
While there hence is a relationship to the new closure scheme,
thermodynamic consistency can only be achieved if this correction by projection is already incorporated into the closure scheme as it has been introduced above. In other words, while this correction restores consistency with the continuity equation,
it is not consistent with the second law of thermodynamics.

The new approach also sheds additional light on the Maxwell-Stefan closure. Note that
\begin{equation}
\langle {\bf U},  {\bf R} \, \nabla \frac{\boldsymbol{\mu}}{RT} \rangle
= \langle  {\bf U}^\ast , {\bf P}^{\sf T} \, {\bf R} \, \nabla \frac{\boldsymbol{\mu}}{RT} \rangle
= \langle  {\bf U} , {\bf P}^{\sf T} \, {\bf R} \, \nabla \frac{\boldsymbol{\mu}}{RT} \rangle
= \langle  {\bf U} , \vec{\bf d} \rangle
\end{equation}
with ${\bf d}_i$ from \eqref{di},
since ${\bf P}^2={\bf P}$ and ${\bf U}={\bf P} \,{\bf U}^\ast$ with the unconstrained diffusion velocities ${\bf U}^\ast$.
While the second term allows for a direct closure of the (unconstrained) diffusion velocities,
the third term cannot be used for a direct (unconstrained) closure for ${\bf U}$.
Instead, the Maxwell-Stefan approach exploits the constraint on the ${\bf d}_i$ to get the implicit
relations for the ${\bf u}_i$ as shown above.

Out of the three closure schemes discussed above, the novel scheme is the only one
for which a diagonal coefficient matrix leads to sensible, actually realistic diffusion fluxes.
Of course all three closure schemes are equivalent in the sense that with fully occupied coefficient matrix with
entries depending on the primitive thermodynamic variables, the same classes of diffusion fluxes are admissible.
But the different approaches of course lead to different functional dependencies.
In this context, the novel closure scheme yields a better understanding of the multicomponent Darken equation,
showing that the main cross-effect is introduced by the projection ${\bf P}$ which is needed to account for the
constraint \eqref{flux-constraint}. Consequently, there is no ''true'' cross-diffusion in this case, but rather
a cross-coupling because of the continuity equation plus a weak cross-effect since the diagonal diffusivities
dependent on the composition.
Nevertheless, this link between the novel closure and the multicomponent Darken equation, together with
the fact that the latter describes simulated diffusivities rather accurately in several case,
indicates that the novel closure might constitute a more appropriate way to represent the intrinsic structure of the
continuum thermodynamical mass diffusion fluxes.
It might, therefore, also be employed for a more efficient and accurate fitting of diffusion coefficients
obtained from experiments or molecular dynamics simulations.
In particular, it would be very interesting to see how far the core-diagonal case from section~\ref{Darken}
with general diagonal entries $d_i=d_i(T,c,{\bf x})$ is already able to model cross-diffusivities obtained from
MD simulations for non-ideal, complex mixtures.

As a final remark, let us note that--surprisingly--it turns out that ternary systems are always core-diagonal as it is shown in the appendix. In terms of the Maxwell-Stefan diffusivities $\D_{ik}$,
the reciprocals $1/d_i({\bf y})$ of the diagonal elements $d_i({\bf y})$ of
${\mathcal D}={\rm diag}(d_1({\bf y}),d_2({\bf y}),d_3({\bf y}))$ are
\begin{equation}
\frac{y_1 \D_{23}}{\D_{12} \D_{13}} \!+\! \frac{y_2}{\D_{12}} \!+\! \frac{y_3}{\D_{13}},\quad
\frac{y_1}{\D_{12}} \!+\!  \frac{y_2 \D_{13}}{\D_{12} \D_{23}} \!+\! \frac{y_3}{\D_{23}},\quad
\frac{y_1}{\D_{13}} \!+\! \frac{y_2}{\D_{23}} \!+\!  \frac{y_3 \D_{12}}{\D_{13} \D_{23}}.
\end{equation}
So, for example, if $y_1\to 0+$, then
\[
\frac{1}{d_1({\bf y})}\to \frac{y_2}{\D_{12}} \!+\! \frac{y_3}{\D_{13}},
\]
apparently a reasonable expression for the diffusivity of the diluted component $A_1$ against the mixture of $A_2$ and $A_3$.
Note that the $\D_{ij}$ are themselves functions of the composition, so that the full dependence on ${\bf y}$ can be more complex,
but also more simple.
It would be desirable to understand why cross-diffusion in ternary systems is solely due to the constraint \eqref{flux-constraint}. More generally, it would be very interesting to understand the meaning of additional -- {\it true\,?!} -- cross-diffusion effects in
multicomponent mixtures with $N>3$ constituents.

If large classes of multicomponent diffusion systems turn out to be core-diagonal--a question to be studied especially
by means of molecular dynamics simulations--it would be very natural to also try benefiting from this structure for a rigorous mathematical wellposedness analysis\\[2ex]

{\bf Acknowledgment.}
The authors cordially thank Wolfgang Dreyer (WIAS) for fruitful and intense scientific exchange over several years.
They are also grateful to Jadran Vrabec (TU Berlin) for helpful discussions on the sign of Fickian diffusivities.

$\mbox{ }$\\[-3ex]
\newpage
\appendix
\section{The Molar-Based Maxwell-Stefan Equations}\label{MS-ChemEng}
In the Chemical Engineering literature, the molar-based variant of the Maxwell-Stefan equations is the common choice.
This variant follows if one uses molar fractions $x_i, x_k$ instead of $y_i, y_k$ in \eqref{fric-coeff}, i.e.\
\begin{equation}\label{fric-coeff-mol}
\tau_{ik} = - c f_{ik}^{\rm mol} x_i x_k \quad \mbox{ for } i,k=1,\ldots ,N
\end{equation}
with molar-based friction coefficients $f_{ik}^{\rm mol}$.
Since the $f_{ik}$ and the $f_{ik}^{\rm mol}$ are functions of the state variables $(T,\rho_1, \ldots , \rho_N)$ anyhow,
this equivalent form just means to let
\[
f_{ik}^{\rm mol} = f_{ik} M_i M_k  \frac c \rho  = f_{ik} M_i M_k  \frac{\sum_l \rho_l/M_l}{\sum_l \rho_l}.
\]
Now, noticing that the $f_{ik}^{\rm mol}$ have the physical dimension of reciprocal diffusivities, one introduces the so-called Maxwell-Stefan diffusivities as
\begin{equation}\label{MS-diffusivities}
\D_{ik}=\frac{1}{f_{ik}^{\rm mol}}.
\end{equation}
Then the so-called generalized Maxwell-Stefan equations result which read as
\begin{equation}\label{Closure-MS-iii}
- \sum_{k=1}^{N} \frac{x_k {\bf j}_i^{\rm mol} - x_i {\bf j}_k^{\rm mol}}{\D_{ik}}   \, =\,
c_i  \nabla \frac{\mu_i^{\rm mol}}{RT}  - y_i \sum_{k=1}^N c_k \nabla \frac{\mu_k^{\rm mol}}{RT}  \quad \mbox{ for } \; i=1,\ldots ,N,
\end{equation}
where ${\bf j}_i^{\rm mol}:={\bf j}_i/M_i=c_i {\bf u}_i$ denote the molar mass fluxes and
${\mu_i^{\rm mol}}=M_i \mu_i$ are the molar-based chemical potentials.
In condensed tensor notation, employing also the notation introduced in the main text, this reads as
\begin{equation}
- {\bf B}^{\rm mol} {\bf J}^{\rm mol} = {\bf P}^{\sf T} \,{\bf C}\,  \nabla \frac{\boldsymbol{\mu}^{\rm mol}}{RT}
\end{equation}
with
\begin{equation}\label{Bij-molar-app}
B_{ij}^{\rm mol}=- \frac{x_i}{\D_{ij}} \;\mbox{ for } i\neq j, \quad B_{ii}^{\rm mol}= \sum_{k\neq i} \frac{x_k}{\D_{ik}}.
\end{equation}
The system \eqref{Closure-MS-iii} is to be complemented by \eqref{flux-constraint}, i.e.\
by the constraint
\begin{equation}\label{molflux-constraint}
\sum_{i=1}^N M_i \, {\bf j}_i^{\rm mol} =0.
\end{equation}
Employing \eqref{GD-grad} in the form
\begin{equation}\label{GD-grad-mol}
\sum_{k=1}^N  c_k \nabla \frac{\mu_k^{\rm mol}}{T} = \frac 1 {R T} \nabla p + \rho h \nabla \frac 1 {RT},
\end{equation}
we obtain the equivalent version
\begin{equation}\label{Closure-MS-iv}
- \sum_{k=1}^{N} \frac{x_k {\bf j}_i^{\rm mol} - x_i {\bf j}_k^{\rm mol}}{\D_{ik}}   \, =\,
c_i \nabla \frac{\mu_i^{\rm mol}}{RT}  - \frac{y_i}{RT} \nabla p - \rho_i h \nabla \frac 1 {RT}
\quad \mbox{ for } \; i=1,\ldots ,N.
\end{equation}
Being mainly interested in the diffusion velocities, this yields
\begin{equation}\label{Closure-MS-v}
- \sum_{k=1}^{N} \frac{x_k ({\bf u}_i -  {\bf u}_k)}{\D_{ik}}   \, =\,
\nabla \frac{\mu_i^{\rm mol}}{RT}   - \frac{M_i}{\rho RT} \nabla p - M_i h \nabla \frac 1 {RT}
\quad \mbox{ for } \; i=1,\ldots ,N.\vspace{0.1in}
\end{equation}
The formulations \eqref{Closure-MS-iv} and \eqref{Closure-MS-v} of the Maxwell-Stefan equations,
but with the partial enthalpy $h_i$ instead of $h$ for the
reason explained above, are those which are common in the chemical engineering literature; see, e.g., \cite{Bird}.

\section{The Ternary Case}\label{ternarysection}
Elimination of ${\bf j}_N$ by means of (\ref{flux-constraint}) leads to the reduced system
\begin{equation}
\label{reduced-system}
-\,\tilde{{\bf B}} \, [\, {\bf j}_1 | \cdots | {\bf j}_{N-1} ]^{\sf T}
=
[\, {\bf d}_1 | \cdots | {\bf d}_{N-1}]^{\sf T}
\end{equation}
with ${\bf d}_i$ from \eqref{di}, where the $(N-1)\times (N-1)$-matrix $\tilde{{\bf B}}$ is given by
\begin{equation}
\label{B-def}
\tilde{B}_{ij}=
\left \{
\begin{array}{ll}
\ds
y_i (f_{iN} - f_{ij} ) & \mbox{for } i\neq j,\\[1ex]
\ds
y_i f_{iN} + \sum_{k\neq i} y_k f_{ik} & \mbox{for } i= j \; \mbox{ (with $\ds y_N=1-\sum_{m<N} y_m$)}.
\end{array}
\right .
\end{equation}
In the ternary case ($N=3$) this corresponds to
\begin{equation}
\label{B-ternary}
\tilde{{\bf B}}=
\left[
\begin{array}{cc}
(1-y_2) f_{13} + y_2 f_{12}
  & y_1 (f_{13} - f_{12} ) \\[1ex]
 y_2 ( f_{23}  - f_{12} )
 & (1-y_1) f_{23} + y_1 f_{12}
\end{array}
\right].
\end{equation}
It is easy to check that
\begin{equation}
\label{detB-ternary}
{\rm det}\, \tilde{{\bf B}}=
y_1 f_{12} f_{13}  + y_2 f_{12} f_{23} + y_3 f_{13} f_{23} = {\rm trace}\, ({\rm adj}({\bf B})) \, ,
\end{equation}
in which ${\bf B}$ is the original $3\times 3$ Maxwell-Stefan matrix of \eqref{MS-matrix} or condition (B). We thus see that
${\rm det}\, \tilde{{\bf B}} >0$: In the case of strictly positive $f_{ik}$, we clearly obtain that ${\rm det}\, \tilde{{\bf B}} \geq \min \{f_{12} f_{13}, f_{12} f_{23}, f_{13} f_{23}\} > 0$ while, if we start from the assumption that ${\bf B} \, {\bf Y} \geq d_0 \, {\bf P}^{\sf T} \, {\bf M}^{-1} \, {\bf Y} \, {\bf P}$ is strictly positive definite on $\{{\bf e}\}^{\perp}$, the techniques of Theorem \ref{Th-2} imply that ${\rm trace}\, ({\rm adj}({\bf B})) \geq (d_0/\|{\bf M}\|_{\infty})^{N-1}$.
Hence
\begin{equation}
\label{B-ternary-inverse}
\tilde{{\bf B}}^{-1}=
\frac{1}{{\rm det}\, \tilde{{\bf B}}}
\left[
\begin{array}{cc}
(1-y_1) f_{23} + y_1 f_{12}  & - y_1 (f_{13} - f_{12} ) \\[1ex]
 - y_2 ( f_{23}  - f_{12} )  & (1-y_2) f_{13} + y_2 f_{12}
\end{array}
\right].
\end{equation}
After some straightforward manipulations, this yields
\begin{equation}
\label{flux-ternary}
[\, {\bf j}_1 | {\bf j}_2 | {\bf j}_3 ]^{\sf T}=
\frac{1}{{\rm det}\, \tilde{{\bf B}}}
\left[
\begin{array}{ccc}
(1-y_1) f_{23}  & - y_1 f_{13} & - y_1 f_{12}  \\[1ex]
 - y_2 f_{23}   & (1-y_2) f_{13} & - y_2 f_{12} \\[1ex]
 - y_3 f_{23}   &  - y_3 f_{13} & (1-y_3) f_{12}
\end{array}
\right]
\, [\, {\bf d}_1 | {\bf d}_2 | {\bf d}_3 ]^{\sf T}.
\end{equation}
Interestingly, the diffusion fluxes are hence of the form \eqref{closure-B3}
without off-diagonal terms. Indeed,
\[
 {\bf J}  \, =\, -{\bf P}^{\sf T} \,  [\tilde{\mathcal D} + {\bf Y}\, \tilde{\bf K} ]\, {\bf P}^{\sf T} \, {\bf R}\, \nabla \frac{\boldsymbol{\mu}}{RT}
\]
with
\begin{equation}
\tilde{\mathcal D} \, = \,
\frac{1}{{\rm det}\, \tilde{{\bf B}}}
\left[
\begin{array}{ccc}
 f_{23}  & 0 & 0  \\[1ex]
0   &  f_{13} & 0 \\[1ex]
0   &  0 & f_{12}
\end{array}
\right]
\quad \mbox{ and } \quad
\tilde{\bf K} = {\bf 0}\, .
\end{equation}
At a first glance, the diagonal entries look somewhat strange, but notice the pre-factor. Rewritten in terms of Maxwell-Stefan diffusivities, the reciprocals $d_i({\bf y})^{-1}$ of the diagonal elements of $\tilde{\mathcal D}={\rm diag}(d_1({\bf y}),d_2({\bf y}),d_3({\bf y}))$ are
\begin{equation}
\frac{y_1 \D_{23}}{\D_{12} \D_{13}} \!+\! \frac{y_2}{\D_{12}} \!+\! \frac{y_3}{\D_{13}},\quad
\frac{y_1}{\D_{12}} \!+\!  \frac{y_2 \D_{13}}{\D_{12} \D_{23}} \!+\! \frac{y_3}{\D_{23}},\quad
\frac{y_1}{\D_{13}} \!+\! \frac{y_2}{\D_{23}} \!+\!  \frac{y_3 \D_{12}}{\D_{13} \D_{23}}.
\end{equation}
So, for example, if $y_1\to 0+$, then
\[
\frac{1}{d_1({\bf y})}\to \frac{y_2}{\D_{12}} \!+\! \frac{y_3}{\D_{13}},
\]
apparently a reasonable expression for the diffusivity of the diluted component $A_1$ against the mixture of $A_2$ and $A_3$.
Note that the $\D_{ij}$ are themselves functions of the composition, so that the full dependence on ${\bf y}$ could be different.

This result has a remarkable implication: the cross-diffusion in a ternary system are solely due to the constraint \eqref{flux-constraint}.
\section{The Group Inverse}\label{Drazin}
The application of generalized inverses in the context of Maxwell-Stefan closure equations was initiated in the book \cite{Giovan}. Here we apply the concept of \emph{group inverse} or \emph{Drazin inverse} of a matrix. For the definition and more background information, we refer to the book \cite{BenIsrael}, Chapter 4, or to \cite{Giovan}, Section 7.3.4. We will also use some properties exposed in the paper \cite{MS}. Here we recall only a few preliminaries directly needed in our proofs for the theorems \ref{Th-1} and \ref{Th-2}.

Let $\mathbf{A} \in \mathbb{R}^{N\times N}$. The index $\text{Ind}\, ({\bf A})$ of the matrix ${\bf A}$ is the smallest positive integer $k$ such that $\text{Dim}\, {\rm im}({\bf A}^{k+1}) = \text{Dim}\, {\rm im}({\bf A}^{k})$. The system of equations
\begin{align}\label{GROUPDEF}
{\bf A} \, {\bf X} \, {\bf A} = {\bf A}, \quad \, {\bf X} \, {\bf A} \, {\bf X} = {\bf X}, \quad {\bf A} \, {\bf X} = {\bf X} \, {\bf A} \,  ,
\end{align}
possesses a unique solution ${\bf X}$ if and only if $\text{Ind}\, ({\bf A}) = 1$ (see \cite{BenIsrael}, Ch. 4, Th. 2). The solution ${\bf X}$ is called the group inverse of ${\bf A}$, denoted by ${\bf A}^{\sharp}$. The \emph{Drazin inverse}, denoted by ${\bf A}^D$, is a generalization of the group inverse which is not needed in the present context. We afore mention some straightforward properties, to compare with Proposition 7.3.6 of \cite{Giovan}.
\begin{remark}\label{GROUP}
 Suppose that $\mathbf{A} \in \mathbb{R}^{N\times N}$ has index one.
\begin{enumerate}[(1)]
 \item \label{symm} If ${\bf A}$ is symmetric, then ${\bf A}^{\sharp}$ is symmetric;
 \item \label{pos} If ${\bf A}$ is moreover positive semi-definite, so is ${\bf A}^{\sharp}$;
 \item \label{kern} ${\rm ker}({\bf A}) = {\rm ker} ({\bf A^{\sharp}})$ and ${ \rm im} ({\bf A}) = {\rm im}({\bf A}^{\sharp})$.
\end{enumerate}
\end{remark}
Suppose that, (i), $\mathbf{A} \in \mathbb{R}^{N\times N}$ is a matrix of rank $N-1$, and that, (ii), there are two eigenvectors $\mathbf{b}, \, \mathbf{c} \in \mathbb{R}^N$ with strictly positive components such that $\mathbf{A}^T \mathbf{c}= 0 = \mathbf{A} \, \mathbf{b}$. Then, zero is a simple eigenvalue of $\mathbf{A}$. Since there are positive left and right eigenvectors, we might follow the argument of Lemma 1 in \cite{MS} showing that the index of ${\bf A}$ is equal to $1$, and that the group inverse $\mathbf{A}^{\sharp}$ also satisfies (i), (ii), together with the identities
\begin{align*}
 \mathbf{A}^{\sharp} \, \mathbf{A} = \mathbf{A} \, \mathbf{A}^{\sharp} = \mathbf{I} - \mathbf{b} \otimes \mathbf{c} \, ,
\end{align*}
where we assume that $\mathbf{b}, \, \mathbf{c}$ are normalized such that $\langle \mathbf{b}, \,  \mathbf{c} \rangle = 1$. The statements in the paper \cite{MS} concern the so-called \emph{Drazin inverse} of the matrix ${\bf A}$. For a matrix of index $1$, we have ${\bf A}^D = {\bf A}^{\sharp}$.

For all $ t \neq 0$, the matrix $\mathbf{A} + t \,\mathbf{ b} \otimes \mathbf{c}$ is invertible, and $(\mathbf{A} + t \, \mathbf{b} \otimes \mathbf{c})^{-1} = \mathbf{A}^{\sharp} + \mathbf{b} \otimes \mathbf{c}/t$ (see \cite{MS}, page 150). Thus, for $0 < t$ sufficiently small, the matrix $\mathbf{A} + t \, \mathbf{b} \otimes \mathbf{c}$ is always inverse positive (see \cite{MS}, Theorem 2).
We denote by ${\rm adj} ({\bf A})$ the adjugate of ${\bf A}$.
\begin{lemma}\label{lemdrazin}
Let ${\bf A}$ satisfy (i), (ii). Then for $t > 0$, $ {\rm det}(\mathbf{A} + t \, \mathbf{b} \otimes \mathbf{c}) = D_0  \, t$ with $D_0 = {\rm trace} ( {\rm adj} ({\bf A})) \neq 0$. Moreover
\begin{align}\label{reprdrazin}
 \mathbf{A}^{\sharp} = \frac{1}{t \,  D_0} \, ({\rm adj}(\mathbf{A} + t \, \mathbf{b} \otimes \mathbf{c}) - D_0 \, \mathbf{b} \otimes \mathbf{c}) \, .
\end{align}
\end{lemma}
\begin{proof}
 We note that $(\mathbf{A} + t \, \mathbf{b} \otimes \mathbf{c})^{-1} \,  \mathbf{b} = (\mathbf{A}^{\sharp} + \frac{1}{t} \, \mathbf{b} \otimes \mathbf{c}) \, {\bf b} = \frac{1}{t} \, \mathbf{b}$. Using Jacobi's differential formula for the determinant, we have
\begin{align*}
\frac{d}{dt} \text{det}(\mathbf{A} + t \, \mathbf{b} \otimes \mathbf{c}) = & \text{trace}\big(\text{adj}(\mathbf{A}+t \, \mathbf{b} \otimes \mathbf{c}) \cdot  \mathbf{b} \otimes \mathbf{c}\big)\\
= & \text{det}(\mathbf{A} + t \, \mathbf{b} \otimes \mathbf{c}) \, \text{trace}\big((\mathbf{A}+t \, \mathbf{b} \otimes \mathbf{c})^{-1} \cdot  \mathbf{b} \otimes \mathbf{c}\big)\\
= & \frac{\text{det}(\mathbf{A} + t \, \mathbf{b} \otimes \mathbf{c})}{t} \, \text{trace}(\mathbf{b} \otimes \mathbf{c}) \, .
\end{align*}
Since ${\rm trace}(\mathbf{b} \otimes \mathbf{c}) = 1$, the function $\text{det}(\mathbf{A} + t \, \mathbf{b} \otimes \mathbf{c}) =: g(t)$ satisfies the ordinary differential equation $g^{\prime} = g/t$. Moreover, $g(0) = \text{det}(\mathbf{A}) = 0$ implies that $g(t) = D_0 \, t$ for some constant $D_0 \neq 0$.

Since further $(\mathbf{A} + t \, \mathbf{b} \otimes \mathbf{c})^{-1} = \mathbf{A}^{\sharp} + \frac{1}{t} \, \mathbf{b} \otimes \mathbf{c}$, we have also
\begin{align*}
 \text{adj}(\mathbf{A} + t \, \mathbf{b} \otimes \mathbf{c}) = \text{det}(\mathbf{A} + t \, \mathbf{b} \otimes \mathbf{c}) \, (\mathbf{A}^{\sharp} + \frac{1}{t} \, \mathbf{b} \otimes \mathbf{c})
 = D_0 \, (t \, \mathbf{A}^{\sharp} + \mathbf{b} \otimes \mathbf{c}) \, .
\end{align*}
We directly obtain the representation \eqref{reprdrazin} of $\mathbf{A}^{\sharp}$.
Letting moreover $t \rightarrow 0$, we find that $ \text{adj}(\mathbf{A}) = D_0 \, (\mathbf{b} \otimes \mathbf{c})$. Thus, it also follows that $D_0 = \text{trace} (\text{adj}({\bf A}))$.
\end{proof}
Notice that Lemma \ref{lemdrazin} characterizes $D_0$ as the sum of the principal minors of leading order of ${\bf A}$: $D_0 = \sum_{i=1}^N \text{det}({\bf A}[i|i])$, where ${\bf A}[i|i]$ is the $(N-1) \times (N-1)$ matrix obtained by canceling rows $i$ and column $i$ of ${\bf A}$.

If the function $1/D_0$ is bounded, \eqref{reprdrazin} shows that the entries of $\mathbf{A}^{\sharp}$ are regular functions (polynomials) of the entries of $\mathbf{A}$, $\mathbf{b}$ and $\mathbf{c}$.

In the case that ${\bf A}$ is a singular \emph{M}-matrix (for instance ${\bf A} = \mathbf{B}({\bf y})$ is the Maxwell-Stefan matrix with positive friction coefficients), a bound for $1/D_0$ is obtained using the strict diagonal dominance by columns of each matrix by ${\bf A}[i|i]$. For the sake of completeness concerning the analysis of the Maxwell-Stefan equations, we sketch this application of the group inverse just hereafter.

Consider the Maxwell-Stefan equation $\mathbf{B}(\mathbf{y}) \, \mathbf{J} = - \mathbf{d}$ (see \eqref{MS-mass-based}) in which
\begin{align*}
 b_{ik}(\mathbf{y}) =  - y_i \, f_{ik}(\mathbf{y}) \text{ for } k \neq i, \qquad b_{ii}(\mathbf{y}) := \sum_{k\neq i} f_{ik}(\mathbf{y}) \, y_k \, ,
\end{align*}
in which $\mathbf{y}$ are the mass fractions and $\mathbf{d} := \mathbf{R} \, \mathbf{P}(\mathbf{y}) \, \nabla \frac{\boldsymbol{\mu}}{RT}$. For strictly positive $f_{ik}$ which are regular functions of $\mathbf{y}$, the rank of $\mathbf{B}(\mathbf{y})$ is $N-1$ and $\mathbf{e}^{\sf T} \, \mathbf{B}(\mathbf{y}) = 0 = \mathbf{B}(\mathbf{y}) \, \mathbf{y}$.


In order to show that the entries of $\mathbf{B}^{\sharp}(\mathbf{y})$ are regular functions of $\mathbf{y}$, it is, as seen, sufficient to show that $D_0(y) := \sum_{i=1}^N \text{det}(\mathbf{B}(y)[i|i])$ remains strictly positive. Now $(\mathbf{B}(y)[i|i])^{\sf T}$ being diagonally dominant by rows, classical results give $\|(\mathbf{B}(y)[i|i])^{-1}\|_{\infty} \leq 1/(y_i \, \min_{j\neq i} f_{ji})$ (see \cite{Var}).

Since $\text{det}((\mathbf{B}(y)[i|i])^{-1}) \leq \|(\mathbf{B}(y)[i|i])^{-1}\|_{\infty}^{N-1}$, we have
\begin{align*}
 [\text{det}(\mathbf{B}(y)[i|i])]^{\frac{1}{N-1}} \geq & y_i \, \min_{j\neq i} f_{ji} \, ,
\end{align*}
and this yields $D_0 \geq  c_0(N) \, (\min_{i\neq j} f_{ij})^{N-1}$ with a certain constant $c_0(N) > 0$ depending only on $N$.

For comparison, notice that in Theorem \ref{Th-2}, we show a way to estimate $D_0$ replacing the assumption of strictly positive friction \emph{coefficients} by positivity assumptions on the matrix ${\bf B} \, {\bf R}$.
%
%
%
%
%
%
%


\begin{thebibliography}{99}
\baselineskip11pt
\bibitem{amann89} H.~Amann: Dynamic theory of quasilinear parabolic systems. {III}. {G}lobal existence.
Math.\ Z.\ {\bf 202}, 219-250 (1989).
\bibitem{amann93} H.~Amann: Nonhomogeneous linear and quasilinear elliptic and parabolic boundary value problems.
pp.\ 9-126 in {\em Function spaces, differential operators and nonlinear analysis ({F}riedrichroda, 1992)}, Vol.\ 133 of {\em Teubner-Texte Math.}, Teubner, Stuttgart 1993.
\bibitem{BenIsrael} A.\ Ben Israel, Greville N.E.~: {\it Generalized Inverses. Theory and Applications.}, CMS Books in Mathematics, Springer, 2003.
\bibitem{Bird} R.B.\ Bird, W.E.\ Stewart, E.N.\ Lightfoot:
{\it Transport Phenomena} (2$^{\rm nd}$ edition). Wiley, New York 2007.
\bibitem{bondesanbriant} M.\ Briant, A.\ Bondesan: Perturbative Cauchy theory for a flux-incompressible Maxwell-Stefan system in a non-equimolar regime. Preprint avalaible at arXiv:1910.03279 [math.AP], 2019.
\bibitem{DB-MS} D.\ Bothe: On the Maxwell-Stefan equations to multicomponent diffusion,
pp.\ 81-93 in {\it Progress in Nonlinear Differential Equations and their Applications Vol. 60} (P.\ Guidotti, Chr.\ Walker et al., eds),  Springer, Basel 2011.
\bibitem{BD} D.\ Bothe, W.\ Dreyer: Continuum thermodynamics of chemically reacting fluid mixtures,
Acta Mechanica {\bf 226}, 1757-1805 (2015).
\bibitem{BD-FE} D.\ Bothe, P.-E.\ Druet:
Construction of thermodynamic potentials for compressible and incompressible multicomponent systems (in preparation).
\bibitem{BD-COMP} D.\ Bothe, P.-E.\ Druet:
Mass transport in multicomponent compressible fluids: local and global well-posedness in classes of strong solutions for general class-one models (submitted), Preprint available at {\verb|http://www.wias-berlin.de/preprint/2658/wias_preprints_2658.pdf|} and at arXiv:2001.08970 [math.AP], 2019.
\bibitem{bothepruess}
D.\ Bothe, J.\ Pr\"{u}ss: Modeling and analysis of reactive multi-component two-phase flows
  with mass transfer and phase transition -- the isothermal incompressible
  case. Discrete Contin. Dyn. Syst. Ser. S {\bf 10}, 673--696 (2017).
  \bibitem{chakrawangeapen} B.\ Chakraborty, J.\ Wang, J.\ Eapen: Multicomponent diffusion in molten {LiCl}-{KCl}: Dynamical correlations and divergent Maxwell-Stefan diffusivities. Phys. Rev. E {\bf 87}, 052312 (2013).
\bibitem{ConstantinEscher} A.\ Constantin, J.\ Escher: Global solutions for quasilinear parabolic problems.
J.\ Evol.\ Equ.\ {\bf 2}, 97-111 (2002).
\bibitem{chenengstromagren} Q.\ Chen, A.\ Engstr\"{o}m, J.\ Agren: On Negative Diagonal Elements in the Diffusion Coefficient
Matrix of Multicomponent Systems. J. Phase Equilib. Diffus., published online https://doi.org/10.1007/s11669-018-0648-x, 2018.
\bibitem{chenjuengel}
X.\ Chen, A.\ J{\"{u}}ngel: Analysis of an incompressible {N}avier-{S}tokes-{M}axwell-{S}tefan
  system. Commun. Math. Phys. {\bf 340}, 471--497 (2015).
\bibitem{CB-review}
C.\ F.~Curtiss, R.\ B.~Bird: Multicomponent Diffusion.
Ind.\ Eng.\ Chem.\ Res.\ {\bf 38}, 2515-2522 (1999).
\bibitem{CB}
C.\ F.~Curtiss, J.\ O.~Hirschfelder: Transport properties of multicomponent gas mixtures. J.\ Chem.\ Phys.\ {\bf 17}, 550-555 (1949).
\bibitem{Cussler}
E.~Cussler: {\em Diffusion: Mass Transfer in Fluid Systems} (3$^{\rm nd}$ edition), Cambridge Series in Chemical Engineering. Cambridge: Cambridge University Press 2009.
\bibitem{Darken} L.S.\ Darken:
Diffusion, mobility and their interrelation through free energy in binary
metallic systems. Trans. AIME {\bf 175}, 184-201 (1948).
\bibitem{Datta} R.\ Datta, S.\ A.\ Vilekar:
The continuum mechanical theory of multicomponent diffusion in fluid mixtures,
Chem.\ Eng.\ Sci.\ {\bf 65}, 5976--5989 (2010).
\bibitem{Dreyer-Guhlke-M}
W.\ Dreyer, C.\ Guhlke, R.\ M\"uller: Bulk-surface electrothermodynamics and applications to electrochemistry,
Entropy {\bf 20}, 939,1-44 (2018).
\bibitem{dredrugagu20}
W.\ Dreyer, P.-E.\ Druet, P.\ Gajewski, C.\ Guhlke: Existence of weak solutions for improved {N}ernst-{P}lanck-{P}oisson
  models of compressible reacting electrolytes. Z. Angew. Math. Phys. {\bf 71:119}, Open access. https://doi.org/10.1007/s00033-020-01341-5 (2020).
  \bibitem{PED-MS} P.-E.\ Druet: A Theory of Generalised Solutions for Ideal Gas Mixtures with Maxwell-Stefan Diffusion, Disc. Cont. Dyn. Sys. Ser. S, to appear (2020).
\bibitem{Toor62} J.B.\ Duncan, H.L.\ Toor:
An experimental study of three component gas diffusion. AIChE Journal {\bf 8}, 38-41 (1962).
\bibitem{Fick} A.\ Fick: \"Uber Diffusion. Annalen der Physik {\bf 170}, 59-86 (1855).
\bibitem{Giovan} V.\ Giovangigli, {\it Multicomponent Flow Modeling},
Birkh\"auser, Boston 1999.
\bibitem{dGM} S.R.\ de Groot, P.\ Mazur: {\it  Non-Equilibrium Thermodynamics}. Dover Publications, 1984.
\bibitem{G-Vrabec2020}
G.\ Guevara-Carrion, R.\ Fingerhut, J.\ Vrabec:
Fick Diffusion Coefficient Matrix of a Quaternary Liquid Mixture by Molecular Dynamics.
J.\ Phys.\ Chem.\ B {\bf 124}, 4527-4535 (2020).
\bibitem{JP-MS} M.\ Herberg, M.\ Meyries, J.\ Pr\"uss, M.\ Wilke:
Reaction-diffusion systems of Maxwell-Stefan type with reversible mass-action kinetics. Nonlinear Analysis {\bf 159},
264-284 (2017).
\bibitem{Hirsch} J.O.\ Hirschfelder, C.F.\ Curtiss, R.B.\ Bird:
{\it Molecular Theory of Gases and Liquids} (2$^{\rm nd}$ corrected printing). Wiley, New York 1964.
\bibitem{Hutter-book} K.\ Hutter, K.\ J\"ohnk, {\it Continuum Methods of Physical Modeling},
Springer, Heidelberg 2004.
\bibitem{Kerk} P.J.A.M.\ Kerkhof, M.A.M.\ Geboers:
Analysis and extension of the theory of multicomponent fluid diffusion. Chem.\ Eng.\ Sci.\ {\bf 60}, 3129-3167 (2005).
\bibitem{KB-book} S.\ Kjelstrup, D.\ Bedeaux, {\it Non-Equilibrium Thermodynamics
of Heterogeneous Systems}, in ''Series on Advances in Statistical Mechanics -- Volume 16'',
World Scientific, Singapore 2008.
\bibitem{Vrabec2019} S.\ Kozlova, A.\ Mialdun, I.\ Ryzhkov, T.\ Janzen, J.\ Vrabec, V.\ Shevtsova: Do ternary liquid mixtures exhibit negative main Fick diffusion coefficients? Phys. Chem. Chem. Phys. {\bf 21}, 2140 (2019).
\bibitem{kraaijeveldwesselingh} G.\ Kraaijeveld, J.A.\ Wesselingh: Negative Maxwell-Stefan Diffusion Coefficients. Ind.\ Eng.\ Chem.\ Res.\ {\bf 32}, 738-742 (1993).
\bibitem{Krishna2019}
R.\ Krishna: Diffusing uphill with James Clerk Maxwell and Josef Stefan.
Chem.\ Eng.\ Sci. {\bf 195}, 851-880 (2019).
\bibitem{KvB2005}
R.\ Krishna, J.M.\ van Baten:
The Darken Relation for Multicomponent Diffusion in Liquid Mixtures of Linear Alkanes: An Investigation Using Molecular
Dynamics (MD) Simulations.
Ind.\ Eng.\ Chem.\ Res.\ {\bf  44}, 6939-6947 (2005).
\bibitem{KvB2016}
R.\ Krishna, J.M.\ van Baten:
Describing diffusion in fluid mixtures at elevated pressures by combining the Maxwell–Stefan formulation with an equation of state.
Chem.\ Eng.\ Sci. {\bf 153}, 174-187 (2016).
\bibitem{KW} R.\ Krishna, J.A.\ Wesselingh:
The Maxwell-Stefan approach to mass transfer. Chem.\ Eng.\ Sci.\ {\bf 52}, 861-911 (1997).
\bibitem{Liu} I-Shih Liu, {\it Continuum Mechanics},
Springer, 2002.
\bibitem{Bardow2011} X.\ Liu, T.J.H.\ Vlugt, A.\ Bardow:
Predictive Darken equation for Maxwell-Stefan diffusivities in multicomponent mixtures.
Ind.\ End.\ Chem.\ Res.\ {\bf 50}, 10350-10358 (2011).
\bibitem{Max} J.C.\ Maxwell:
On the dynamical theory of gases,
Phil.\ Trans.\ R.\ Soc.\ {\bf 157}, 49-88 (1866).
\bibitem{Mehrer} H.\ Mehrer, N.A.\ Stolwijk: Heroes and highlights in the history of diffusion.
Diffusion Fundamentals {\bf 11}, 1-32 (2009).
\bibitem{MS} C.D.\ Meyer Jr., M.W.\ Stadelmaier:
Singular \emph{M}-Matrices and inverse positivity. Linear Algebra and its Applications {\bf 22}, 139-156 (1978).
\bibitem{Miller} D.\ G.\ Miller, V.\ Vitagliano, R.\ Sartorio: Some comments on multicomponent diffusion: Negative main term diffusion coefficients, second law constraints, solvent choices, and reference frame transformations. J. Phys. Chem. {\bf 90}, 1509-1519 (1986).
\bibitem{M85} I.\ M\"uller: {\it  Thermodynamics}.
Pitman 1985.
\bibitem{Muck} C.\ Muckenfuss:
Stefan-Maxwell relations for multicomponent diffusion and the Chapman Enskog solution
of the Boltzmann equations. J.\ Chem.\ Phys.\ {\bf 59}, 1747-1752 (1973).
\bibitem{Raja} K.R.\ Rajagopal, L.\ Tao:
{\it Mechanics of Mixtures}. World Scientific Publishers, Singapore 1995.
\bibitem{mupoza15} P.B.\ Mucha, M.\ Pokorny, E.\ Zatorska: Heat-conducting, compressible mixtures with multicomponent diffusion: construction of a weak solution. SIAM J. Math. Anal. {\bf 47}, 3747-3797 (2015).
\bibitem{mutorufirooz} J.\ W.\ Mutoru, A.\ Firoozabadi: Form of multicomponent Fickian diffusion coefficients matrix. J. Chem. Thermodynamics {\bf 43}, 1192-1203 (2011).
\bibitem{Onsager-Diffusion}
L.\ Onsager: Theories and problems of liquid diffusion,
Ann.\ N.\ Y.\ Acad.\ Sci.\ {\bf 46}, 241-265 (1945)
\bibitem{PekarSamohyl}
M.\ Pekar, I.\ Samohyl:
{\it The Thermodynamics of Linear Fluids and Fluid Mixtures}. Springer 2014.
\bibitem{Peters2020}
C.\ Peters, J.\ Thien, L.\ Wolff, H.-J.\ Ko\ss, A.\ Bardow:
Quaternary Diffusion Coefficients in Liquids from Microfluidics and Raman Microspectroscopy: Cyclohexane + Toluene + Acetone + Methanol.
J.\ Chem.\ Eng.\ Data {\bf 65}, 1273-1288 (2020).
\bibitem{piashiba19}
T.\ Piasecki, Y.\ Shibata, and E.\ Zatorska: On strong dynamics of compressible two-component mixture flow. SIAM J. Math. Anal {\bf 51}, 2793--2849 (2019).
\bibitem{Snell} F.M.\ Snell, R.A.\ Spangler: A phenomenological theory of transport in multicomponent systems.
J.\ Phys.\ Chem.\ {\bf 71}, 2503-2510 (1967).
\bibitem{Standart} G.L.\ Standart, R.\ Taylor, R.\ Krishna:
The Maxwell-Stefan formulation of irreversible thermodynamics for simultaneous
heat and mass transfer. Chem.\ Engng.\ Commun.\ {\bf 3}, 277-289 (1979).
\bibitem{Stef} J.\ Stefan:
\"Uber das Gleichgewicht und die Bewegung insbesondere die Diffusion von Gasgemengen,
Sitzber.\ Akad.\ Wiss.\ Wien {\bf 63}, 63-124 (1871).
\bibitem{TK-book} R.\ Taylor, R.\ Krishna:
{\it Multicomponent mass transfer}. Wiley, New York 1993.
\bibitem{CT69} C.\ Truesdell:
{\it Rational Thermodynamics}. McGraw-Hill Series in Modern Applies Mathematics, New York 1969.
\bibitem{CT62} C.\ Truesdell:
Mechanical basis of diffusion, J.\ Chem.\ Phys. {\bf 37}, 2336-2344 (1962).
\bibitem{Var} J.M.\ Varah:
A lower bound for the smallest singular value of a matrix, Linear Algebra and its Applications {\bf 11}, 3-5 (1975).
\bibitem{Wolf2018}
L.\ Wolff, S.H.\ Jamali, T.M.\ Becker, O.A.\ Moultos, T.J.H.\ Vlugt,  A.\ Bardow:
Prediction of Composition-Dependent Self-Diffusion Coefficients in Binary Liquid Mixtures: The Missing Link for Darken-Based Models
Ind.\ Eng.\ Chem.\ Res.\ {\bf  57}, 14784-14794 (2018).
\end{thebibliography}
\end{document}